\documentclass[11pt]{article}

\usepackage[utf8]{inputenc}

\usepackage{setspace}
\usepackage{appendix}
\usepackage{xfrac}
\usepackage{amssymb,amsfonts}
\usepackage[figuresright]{rotating}
\usepackage{amsmath,amssymb}
\usepackage{graphicx}
\usepackage{dcolumn}
\usepackage{bm}
\usepackage{amscd,amsthm}
\usepackage{ifthen}
\usepackage{mathtools}

\usepackage{pgf,tikz}
\usepackage{pgfplots}
\usetikzlibrary{calc,shadows}
\usetikzlibrary{pgfplots.groupplots}
\usepackage{subfigure}
\usepackage{url}
\usetikzlibrary{pgfplots.groupplots}

\usetikzlibrary{external}
\usepackage{amsthm}
\newtheorem{lemma}{Lemma}

\newtheorem{theorem}{Theorem}

\newtheorem{proposition}{Proposition}

\usepackage{hyperref}
\hypersetup{
    bookmarks=true,         
    unicode=false,          
    pdftoolbar=true,        
    pdfmenubar=true,        
    pdffitwindow=false,     
    pdfstartview={FitH},    
    pdfauthor={Reinhard Heckel},     
    pdfsubject={Subject},   
    pdfcreator={Reinhard Heckel},   
    pdfproducer={Producer}, 
    pdfnewwindow=true,      
    colorlinks=false,       
    linkcolor=red,          
    citecolor=green,        
    filecolor=magenta,      
    urlcolor=cyan           
}

 \newcommand{\mc}[0]{\mathcal }
 \newcommand{\mb}[0]{\mathbb }

\newcommand\norm[2][\Tnorm]{\ensuremath{{\left\Vert #2 \right\Vert}_{#1}}}

\newcommand\ind[2][\Tind]{\ensuremath{ {\mathbf{1}_{#1} (#2) } }}

\newcommand\Tinnerprod{}
\newcommand{\innerprod}[3][\Tinnerprod]{\ifthenelse{\equal{#1}{}}{\ensuremath{\left<#2,#3\right>}}{\ensuremath{\left<#2,#3\right>_{#1}}}}

\newcommand\PR[1]{\ensuremath{{\mathbb{P}}\!\left[#1\right]}}

\newcommand\vect[1]{\mathbf #1}

\renewcommand\Re[0]{\mathrm{Re}}

\newcommand{\val}{ {\bm \alpha} }
\newcommand{\vbe}{ {\bm \beta} }

\newcommand{\va}{\vect{a}}  
\newcommand{\vb}{\vect{b}}

\newcommand{\vf}{\vect{f}}  
\newcommand{\vg}{\vect{g}}

\newcommand{\vn}{\vect{n}}
  
\newcommand{\vp}{\vect{p}}  
\newcommand{\vq}{\vect{q}}
\newcommand{\vr}{\vect{r}}  
\newcommand{\vs}{\vect{s}}  

\newcommand{\vu}{\vect{u}}  
\newcommand{\vv}{\vect{v}}  

\newcommand{\vx}{\vect{x}}  
\newcommand{\vy}{\vect{y}}  
\newcommand{\vz}{\vect{z}}

\newcommand{\mA}{\vect{A}}  
\newcommand{\mB}{\vect{B}} 
\newcommand{\mC}{\vect{C}} 
\newcommand{\mD}{\vect{D}}

\newcommand{\mF}{\vect{F}}
\newcommand{\mG}{\vect{G}}
\newcommand{\mH}{\vect{H}}
\newcommand{\mI}{\vect{I}}

\newcommand{\mL}{\vect{L}}

\newcommand{\mQ}{\vect{Q}}
\newcommand{\mR}{\vect{R}}

\newcommand{\mU}{\vect{U}}
\newcommand{\mV}{\vect{V}}

\newcommand{\mZ}{\vect{Z}}

\newcommand{\safemath}[2]{\newcommand{#1}{\ensuremath{#2}}}
\safemath{\complexset}{\mathbb{C}}
\safemath{\reals}{\mathbb{R}}
\newcommand{\comp}[1]{\overline{#1}}

\newcommand\defeq{\coloneqq}
\newcommand{\abs}[1]{ \left| #1 \right| }

\newcommand\Tex{}
\newcommand\Ex[2][\Tex]{%
\ifthenelse{\equal{#1}{}}{{\mathbb E}[#2]}{\ensuremath{\underset{#1}{\mathbb E}\left[ #2\right]}}}

\newcommand\EX[2][\Tex]{
\ifthenelse{\equal{#1}{}}{{\mathbb E}[#2]}{\ensuremath{\underset{#1}{\mathbb E}\left[ #2\right]}}}

\DeclareMathOperator{\sinc}{sinc}

\newcommand{\inv}[1]{  {#1}^{ -1 } } 

\newcommand{\conj}[1]{ {#1}^* } 
\newcommand{\herm}[1]{{#1}^H} 
\newcommand{\transp}[1]{{#1}^T} 
\newcommand{\mtx}[1]{\mathbf #1}

\newcommand{\minlet}[1]{\tilde #1}

\newcommand{\sign}{\mathrm{sign}}

\newcommand{\sfunc}{s_{\hspace{-0.02cm}H}} 
\newcommand{\tfunc}{L_{\hspace{-0.02cm}H}} 

\newcommand{\prsig}[0]{x}  
\newcommand{\prsigF}[0]{X}  

\newcommand{\derd}{\partial}
\newcommand{\opnormss}{} 

\newcommand{\Tint}{T}
\newcommand{\Bint}{B}
\newcommand{\btlim}[1]{\overline{#1}} 
\newcommand\SRF{\mathrm{SRF}}
\newcommand\taum{\tau_{\max}}
\newcommand\num{\nu_{\max}}

\newcommand{\N}{N}
\newcommand{\K}{K}
\renewcommand{\S}{S}

\newcommand{\FK}{F} 

\newcommand{\x}{\tau}
\newcommand{\y}{\nu}
\renewcommand{\L}{{L}} 
\newcommand{\infdist}[1]{\left| #1\right|}

\newcommand{\tauc}{\bar \tau}
\newcommand{\nuc}{\bar \nu}
\newcommand{\T}{\mathcal S} 

\newcommand{\TL}[1]{\tilde{#1}}

\begin{document}

\title{Super-Resolution Radar}

\author{
Reinhard Heckel\thanks{
Corresponding author: \href{mailto:reinhard.heckel@gmail.com}{reinhard.heckel@gmail.com}
} \thanks{R.~Heckel was with the Department of Information Technology and Electrical Engineering, ETH Zurich, Switzerland and is now with IBM Research, Zurich, Switzerland} , Veniamin I.~Morgenshtern\thanks{Department of Statistics, Stanford University, CA} , and Mahdi Soltanolkotabi\thanks{Ming Hsieh Department of Electrical Engineering, University of Southern California, Los Angeles, CA }  
}

\date{Nov.~2014; revised Aug.~2015; last revised Dec.~2020}


\maketitle

\begin{abstract}
In this paper we study the identification of a time-varying linear system from its response to a known input signal. 
More specifically, we consider systems whose response to the input signal is given by a weighted superposition of delayed and Doppler shifted versions of the input. 
This problem arises in a multitude of applications such as wireless communications and radar imaging. Due to practical constraints, the input signal has finite bandwidth $B$, and the received signal is observed over a finite time interval of length $T$ only. This gives rise to a delay and Doppler resolution of $1/B$ and $1/T$.  
We show that this resolution limit can be overcome, i.e., we can exactly recover the continuous delay-Doppler pairs and the corresponding attenuation factors, by  solving a convex optimization problem. 
This result holds provided that the distance between the delay-Doppler pairs is at least $2.37/B$ in time or $2.37/T$ in frequency. Furthermore, this result allows the total number of delay-Doppler pairs to be linear up to a log-factor in $BT$, the dimensionality of the response of the system, and thereby the limit for identifiability. 
Stated differently, we show that we can estimate the time-frequency components of a signal that is $\S$-sparse in the \emph{continuous} dictionary of time-frequency shifts of a random window function, from a number of measurements, that is linear up to a log-factor in $\S$.
\end{abstract}


\section{Introduction}
The identification of \emph{time-varying} linear systems 
is a fundamental problem in many engineering applications. 
Concrete examples include radar and the identification of dispersive communication channels. 
%
%
In this paper, we study the problem of identifying a system $H$ whose response $y = Hx$  to the probing signal $x$ can be described by finitely many delays and Doppler shifts:
\begin{align} 
y(t)
=
\sum_{j=1}^\S  b_j \prsig(t-\tauc_j)  e^{i2\pi \nuc_j t}.
\label{eq:iorelintro}
\end{align}
Here, $b_j$ is the attenuation factor corresponding to the delay-Doppler pair  $(\tauc_j, \nuc_j)$. 
In radar imaging, for example, this input-output relation corresponds to a scene consisting of $\S$ moving targets modeled by point scatters, where the input $x$ is the probing signal transmitted by the radar, and the output $y$ is the superposition of the reflections of the  probing signal by the point scatters. 
The relative distances and velocities of the targets can be obtained from the delay-Doppler pairs $(\tauc_j, \nuc_j)$. 

In order to identify the system $H$ (e.g.,~to locate the targets in radar) we need to estimate the continuous time-frequency shifts $(\tauc_j, \nuc_j)$ and the  corresponding attenuation factors $b_j$ from a single input-output measurement, i.e., from the response $y$ to some known and suitably selected probing signal $x$. 
There are, however, important constraints on the type of input-output measurements that can be performed in practice: The probing signal $x$ must be band-limited and approximately time-limited. Also, the response $y$ can be observed only over a finite time interval. 
For concreteness, we assume that we observe the response $y$ over an interval of length $T$ and that $x$ has bandwidth $B$ and is approximately supported on a time interval of length proportional to $T$. 
This time- and band-limitation determines the ``natural'' resolution of the system, i.e., the accuracy up to which the delay-Doppler pairs can be identified is $1/\Bint$ and $1/\Tint$ in $\tau$- and $\nu$-directions, respectively. 
This resolution is achieved by a standard pulse-Doppler radar that 
performs digital matched filtering in order to detect the delay-Doppler pairs. 

From \eqref{eq:iorelintro}, it is evident that band- and approximate time-limitation of $x$ implies that $y$ is band- and approximately time-limited as well---provided that the delay-Doppler pairs are compactly supported. 
For example in radar, due to path loss and finite velocity of the targets or objects in the scene this is indeed the case \cite{strohmer_pseudodifferential_2006}. 
Throughout, we will therefore assume that 
$
(\tauc_j, \nuc_j) \in [-T/2,T/2]\times[-B/2,B/2]
$. This is not a restrictive assumption as the region in the $(\tau,\nu)$-plane where the delay-Doppler pairs are located can have area $BT \gg 1$, 
which is very large. In fact, for certain applications, it is reasonable to assume that the system is \emph{underspread}, i.e., that the delay Doppler pairs lie in a region of area $\ll 1$ \cite{taubock_compressive_2010,bajwa_learning_2008,bajwa_identification_2011}. We do not need the underspread assumption in this paper. 

Since $y$ is band-limited and approximately time-limited, by the $2WT$-Theorem \cite{slepian_bandwidth_1976,durisi_sensitivity_2012}, it is essentially characterized by on the order of $BT$ coefficients. 
We therefore sample $y$ in the interval $[-T/2, T/2]$ at rate $1/B$, so as to collect $\L \defeq BT$ samples\footnote{For simplicity we assume throughout that $\L = BT$ is an odd integer.}. 
Furthermore, we choose $x$ partially periodic by taking its samples $x_\ell = x(\ell/B)$ to be $\L$-periodic for $3L$ many samples, and zero otherwise, so that $x$ is essentially supported on an interval of length $3T$. 
For the readers familiar with wireless communication, we point out that the partial periodization of $x$ serves a similar purpose as the cyclic prefix used in OFDM systems. 
As detailed in Section \ref{sec:probform}, the corresponding samples $y_p \defeq y(p/B)$ in the interval $ p/B \in [-T/2, T/2]$ are given by 
\begin{align}
y_p
&= 
\sum_{j=1}^{\S} b_j  
[\mc F_{\nu_j}
\mc T_{\tau_j}
\vx ]_p
%
, \quad p = -\N,...,\N, \quad 
N \defeq \frac{L-1}{2},
\label{eq:periorel}
\end{align}
where 
\begin{align}
[\mc T_{\tau} \vx ]_p
\defeq
\frac{1}{L}
\sum_{k=-\N}^{\N} 
\left[ 
\left(
\sum_{\ell=-\N}^{\N}
x_{\ell} e^{- i2\pi \frac{\ell k}{L}  }
\right)
e^{-i2\pi k  \tau }   
\right]
e^{i2\pi \frac{p k}{L}  } 
\quad
\text{and}
\quad 
[\mc F_{\nu} \vx ]_p
\defeq x_p e^{i2\pi p \nu }.
\label{eq:deftimefreqshifts}
\end{align}
Here, we defined\footnote{To avoid ambiguity, from here onwards we refer to $(\tauc_j, \nuc_j)$ as delay-Doppler pair and to $(\tau_j, \nu_j)$ as time-frequency shift. 
} the time-shifts $\tau_j \defeq  \tauc_j/T$ and frequency-shifts $\nu_j \defeq \nuc_j/B$. Since $(\tauc_j, \nuc_j) \in [-T/2,\allowbreak T/2] \allowbreak \times[-B/2,B/2]$ we have $(\tau_j, \nu_j) \in [-1/2,1/2]^2$. Since $\mc T_{\tau}\vx$ and $\mc F_{\nu}\vx$ are $1$-periodic in $\tau$ and $\nu$, we can assume in the remainder of the paper that $(\tau_j, \nu_j) \in [0,1]^2$. 
The operators $\mc T_{\tau}$ and $\mc F_{\nu}$ can be interpreted as fractional time and frequency shift operators in $\complexset^\L$. If the $(\tau_j,\nu_j)$ lie on a $(1/L,1/L)$ grid, the operators $\mc F_{\nu}$ and $\mc T_{\tau}$ reduce to the ``natural'' time frequency shift operators in $\complexset^\L$, i.e., $[\mc T_{\tau} \vx ]_p = x_{p - \tau \L}$ and $[\mc F_{\nu} \vx ]_p = x_p e^{i2\pi p \frac{\nu\L}{\L} }$. 
The definition of a time shift in \eqref{eq:deftimefreqshifts} as taking the Fourier transform, 
modulating the frequency, and taking the inverse Fourier transform is a very natural definition of a \emph{continuous} time-shift $\tau_j \in [0,1]$ of a \emph{discrete} vector $\vx = \transp{[x_0,...,x_{\L-1}]}$. 
Finally note that to obtain \eqref{eq:periorel} (see Section \ref{sec:probform}) from \eqref{eq:iorelintro}, we approximate a periodic sinc function with a finite sum of sinc functions (this is where partial periodization of $x$ becomes relevant). Thus \eqref{eq:periorel} does not hold exactly if we take the probing signal to be essentially time-limited. However, in Section \ref{sec:probform} we show that the incurred relative error decays as $1/\sqrt{\L}$ 
and is therefore negligible for large $\L$. 
The numerical results in Section \ref{sec:numres} indeed confirm that this error is negligible. 
If we took $x$ to be $T$-periodic on $\reals$,  \eqref{eq:periorel} becomes exact, but at the cost of $x$ not being time-limited. 

The problem of identifying the system $H$ with input-output relation \eqref{eq:iorelintro} under the constraints that the probing signal $x$ is band-limited and the response to the probing signal $y=Hx$ is observed on a finite time interval now reduces to the  estimation of the triplets $(b_j, \tau_j, \nu_j)$ from the samples in \eqref{eq:periorel}. 
Motivated by this connection to the continuous system model, 
in this paper, we consider the problem of recovering the attenuation factors $b_j$ and the corresponding time-frequency shifts $(\tau_j, \nu_j )\in [0,1]^2, j=1,...,\S$, from the samples $y_p, p=-\N,...,\N$, in~\eqref{eq:periorel}. 
We call this the super-resolution radar problem, as recovering the exact time-frequency shifts $(\tau_j,\nu_j)$ ``breaks'' the natural resolution limit  of $(1/B,1/T)$ achieved by standard pulse-Doppler radar. 

Alternatively, one can view the super-resolution radar problem as that of recovering a signal that is $\S$-sparse in the continuous dictionary of time-frequency shifts of an $\L$-periodic sequence $x_\ell$. 
In order to see this, and  to better understand the super-resolution radar problem, it is instructive to consider two special cases.
\subsection{Time-frequency shifts on a grid \label{sec:ongrid}}
If the delay-Doppler pairs $(\tauc_j, \nuc_j)$ lie on a $(\frac{1}{B},  \frac{1}{T})$ grid or equivalently if the time-frequency shifts $(\tau_j, \nu_j)$ lie on a $(\frac{1}{L},\frac{1}{L})$ grid, 
the super-resolution radar problem reduces to a sparse signal recovery problem with a Gabor measurement matrix. 
To see this, note that in this case $\tau_j \L$ and $\nu_j\L$ are integers in $\{0,...,\L-1\}$, and  \eqref{eq:periorel} reduces to 
\begin{align}
y_p
&= \sum_{j=1}^{\S} b_j
x_{p - \tau_j \L}
e^{i2\pi \frac{ (\nu_j \L)   p}{\L}  },  \quad p = -\N,...,\N.\label{eq:periorel_gabor}
\end{align}
Equation \eqref{eq:periorel_gabor} can be written in matrix-vector form 
\[
\vy = \mG \vs. 
\]
Here, $[\vy]_p \defeq y_p$, $\mG \in \complexset^{\L \times \L^2}$ is the Gabor matrix with window $\vx$, where the entry in the $p$th row and $(\tau_j\L, \nu_j\L)$-th column is $x_{p - \tau_j \L} e^{i2\pi \frac{(\nu_j \L)   p}{\L} }$,   
and $\vs \in \complexset^{\L^2}$ is a sparse vector where the $j$-th non-zero entry is given by $b_j$ and is indexed by $(\tau_j\L, \nu_j\L)$. 
Thus, the recovery of the triplets $(b_j,\tau_j,\nu_j)$ amounts to recovering the $\S$-sparse vector $\vs \in \complexset^{\L^2}$ from the measurement vector $\vy \in \complexset^\L$. This is a sparse signal recovery problem with a Gabor measurement matrix. A---by now standard---recovery approach is to solve a simple convex $\ell_1$-norm-minimization program. 
From \cite[Thm.~5.1]{krahmer_suprema_2014} we know that, provided the $x_\ell$ are i.i.d.~sub-Gaussian random variables, and provided that $S\le c L/(\log L)^4$ for a sufficiently small numerical constant $c$,  with high probability, all $\S$-sparse vectors $\vs$ can be recovered from $\vy$ via $\ell_1$-minimization.
Note that the result \cite[Thm.~5.1]{krahmer_suprema_2014} only applies to the Gabor matrix $\mG$ and therefore does not apply to the super-resolution problem where the ``columns'' $\mc F_{\nu} \mc T_{\tau} \vx$ are highly correlated.

\subsection{Only time or only frequency shifts \label{sec:redsupres}}

Next, we consider the case of only time or only frequency shifts, and show that in both cases recovery of the $(b_j,\tau_j)$ and the $(b_j,\nu_j)$, is equivalent to the recovery of a weighted superposition of spikes from low-frequency samples. Specifically, if $\tau_j = 0$ for all $j$, \eqref{eq:periorel} reduces to 
\begin{align}
y_p = 
 x_p \sum_{j=1}^{\S} b_j
e^{i2\pi p \nu_j  }, \quad p = -\N,...,\N.
\label{eq:supres}
\end{align}
The $y_p$ above are samples of a mixture of $\S$ complex sinusoids, and  estimation of the $(b_j,\nu_j)$ corresponds to determining the magnitudes and the frequency components of these sinusoids. 
Estimation of the $(b_j,\nu_j)$ is known as a line spectral estimation problem, and can be solved using approaches such as Prony's method \cite[Ch.~2]{gershman_space-time_2005}. 
Recently, an alternative approach for solving this problem has been proposed, specifically in~\cite{candes_towards_2014} it is shown that exact recovery of the $(b_j,\nu_j)$ is possible by solving a convex total-variation norm minimization problem. This results holds provided that the minimum separation between any two $\nu_j$ is larger than $2/\N$. 
An analogous situation arises when there are only time shifts ($\nu_j = 0$ for all $j$) as taking the discrete Fourier transform of $y_p$ yields a relation exactly of the form \eqref{eq:supres}.

\subsection{Main contribution}

In this paper, we consider a random probing signal by taking the $x_\ell$ in \eqref{eq:periorel} to be i.i.d. Gaussian (or sub-Gaussian) random variables. We show that with high probability, the triplets $(b_j,\tau_j,\nu_j)$ can be recovered perfectly from the $L$ samples $y_p$ by essentially solving a convex program. This holds provided that two conditions are satisfied:
\begin{itemize}
\item \emph{Minimum separation condition:} We assume the time-frequency shifts $(\tau_j,\nu_j) \in [0,1]^2, j = 1,...,\S$, satisfy the minimum separation condition 
\begin{align}
\max(|\tau_j - \tau_{j'}|, |\nu_j - \nu_{j'}| ) \geq \frac{2.38}{\N}
, \text{ for all } j\neq j',
\label{eq:minsepcond}
\end{align}
where $|\tau_j - \tau_{j'}|$ is the wrap-around distance on the unit circle. For example, $|3/4-1/2|=1/4$ but $|5/6-1/6|=1/3\neq 2/3$. 
Note that the time-frequency shifts must not be separated in both time \emph{and} frequency, e.g., \eqref{eq:minsepcond} can hold even when $\tau_j = \tau_{j'}$ for some $j\neq j'$. 
\item \emph{Sparsity:}
We also assume that the number of time-frequency shifts $S$ obeys
\begin{align*}
S\le c\frac{L}{(\log L)^3}, 
\end{align*}
where $c$ is a numerical constant.
\end{itemize}
This result is essentially optimal in terms of the allowed sparsity level, as the number $\S$ of unknowns can be linear---up to a log-factor---in the number of observations $L$. Even when we are given the time-frequency shifts $(\tau_j,\nu_j)$, we can only hope to recover the corresponding attenuation factors $b_j, j=1,...,\S$, by solving \eqref{eq:periorel}, provided that $\S \leq \L$. 

We note that some form of separation between the time-frequency shifts is necessary for stable recovery. 
To be specific, we consider the simpler problem of line spectral estimation (cf.~Section \ref{sec:redsupres}) that is obtained from our setup by setting $\tau_j=0$ for all $j$. 
Clearly, any condition necessary for the line spectral estimation problem is also necessary for the super-resolution radar problem. 
Consider an interval of the $\nu$-axis of length $\frac{2\S'}{\L}$. 
If there are more than $\S'$ frequencies $\nu_j$ in this interval, then the problem of recovering the $(b_j,\nu_j)$ becomes extremely ill-posed when $\S'$ is large \cite[Thm.~1.1]{donoho1992superresolution}, \cite{morgenshtern2014stable}, \cite[Sec.~1.7]{candes_towards_2014}, \cite{beurling1989collected}. 
Hence, in the presence of even a tiny amount of noise, stable recovery is not possible. The condition in~\eqref{eq:minsepcond} 
allows us to have  $0.42\,\S'$ time-frequency shifts in an interval of length $\frac{2\S'}{\L}$, an optimal number of frequencies up to the constant $0.42$. 
We emphasize that while some sort of separation between the time-frequency shifts is necessary, the exact form of separation required in \eqref{eq:minsepcond} may not be necessary for stable recovery and less restrictive conditions may suffice.
Indeed, in the simpler problem of line spectral estimation (i.e., $\tau_j=0$ for all $j$), Donoho \cite{donoho1992superresolution} showed that stable super-resolution is possible via an exhaustive search algorithm even when condition~\eqref{eq:minsepcond} is violated locally as long as every interval of the $\nu$-axis of length $\frac{2\S'}{\L}$ contains less than $\S'/2$ frequencies $\nu_j$ and $\S'$ is small (in practice, think of $\S'\lesssim 10$). The exhaustive search algorithm is infeasible in practice and an important open question in the theory of line spectral estimation is to develop a feasible algorithm that achieves the stability promised in \cite{donoho1992superresolution}. In the special case when the $b_j$ are real and positive, stable recovery can be achieved by convex optimization \cite{morgenshtern2014stable,Caratheodory_ueber_1911,fuchs_sparsity_2005}, see also  \cite{schiebinger_superresolution_2015} for recent results on more general positive combination of waveforms.



Translated to the continuous setup, our result implies that with high probability we can identify the triplets $(b_j, \tauc_j,\nuc_j)$ perfectly provided that  
\begin{equation}
	\label{eq:mindist}
|\tauc_j -\tauc_{j'}| \geq \frac{4.77}{B} \,\text { or } \, |\nuc_j - \nuc_{j'}| \geq \frac{4.77}{T}, \quad \text{ for all } j\neq j'
\end{equation}
and $S\le c\frac{BT}{\left(\log (BT)\right)^3}$.
%
Since we can exactly identify the delay-Doppler pairs $(\tauc_j,\nuc_j)$, our result offers a significant improvement in resolution over conventional radar techniques. 
Specifically, with a standard pulse-Doppler radar that samples the received signal and performs digital matched filtering in order to detect the targets, 
the delay-Dopper shifts $(\tauc_j,\nuc_j)$ can only be estimated up to an uncertainty of about $(1/T,1/B)$. 

We hasten to add that in the radar literature, the term super-resolution is often used for  the ability to resolve very close targets, specifically even closer than the Rayleigh resolution limit \cite{quinquis_radar_2004} that is proportional to $1/B$ and $1/T$ for delay and Doppler resolution, respectively.  
Our work permits identification of a \emph{each} target with precision much higher than $1/B$ and $1/T$ as long as other targets are not too close,  specifically other targets should be separated by a constant multiple of the Rayleigh resolution limit (cf. \eqref{eq:mindist}).

Recall that $(\tau_j, \nu_j) \in [0,1]^2$ translates to $(\tauc_j,\nuc_j) \in [-T/2, T/2] \times [-B/2,B/2]$, i.e., the $(\tauc_j,\nuc_j)$ can lie in a rectangle of area $\L=BT \gg 1$, i.e., the system $H$ does not need to be underspread\footnote{A system is called underspread if its spreading function is supported on a rectangle of area much less than one.}. 
The ability to handle systems that are \emph{overspread} is important in radar applications. Here,  we might need to resolve targets with large relative distances and relative velocities, resulting in delay-Doppler pairs $(\tauc_j,\nuc_j)$ that lie in a region of area larger than $1$ in the time-frequency plane. 

We finally note that standard non-parametric estimation methods such as the MUSIC algorithm can in general not be applied directly to the super-resolution radar problem. 
That is, because MUSIC relies on \emph{multiple} measurements (snapshots) \cite[Sec.~6.3]{stoica_spectral_2005}, whereas we assume only a \emph{single} measurement $y_p, p = -\N, ..., \N$ to be available. 
However, by choosing the probing signal $\vx$ in \eqref{eq:periorel} to be periodic, the single measurement $y_p, p = -\N, ..., \N$ can be transformed into multiple measurements and MUSIC may be applied. This approach, discussed in detail in Appendix \ref{app:MUSIC}, however, requires the $\nu_j$ to be distinct, the time-shifts $\tau$ to lie in a significantly smaller range than $[0,1]$, $S < \sqrt{L}$, 
and is (significantly) more sensitive to noise than our convex programming based approach. 
For the case that multiple measurements are available, e.g., by observing distinct paths of a signal by an array of antennas, subspace methods have been studied for delay-Doppler estimation \cite{jakobsson_subspace-based_1998}.

\subsection{Notation}
We use lowercase boldface letters to denote (column) vectors and uppercase boldface letters to designate matrices. 
The superscripts $\transp{}$ and $\herm{}$ stand for transposition and Hermitian transposition, respectively. 
For the vector $\vx$, $x_q$ and $[\vx]_q$ denotes its $q$-th entry, $\norm[2]{\vx}$ its $\ell_2$-norm and $\norm[\infty]{\vx} = \max_q |x_q|$ its largest entry.  
For the matrix $\mA$, $[\mA]_{ij}$ designates the entry in its $i$-th row and $j$-th column, 
 $\norm[\opnormss]{\mA} \defeq\;$ $\max_{\norm[2]{\vv} = 1  } \norm[2]{\mA \vv}$ its spectral norm, $\norm[F]{\mA} \defeq (\sum_{i,j} |[\mA]_{ij}|^2 )^{1/2}$ its Frobenius norm, and $\mA \succeq 0$ signifies that $\mA$ is positive semidefinite. 
The identity matrix is denoted by $\mI$. 
For convenience, we will frequently use a two-dimensional index for vectors and matrices, e.g., we write $[\vg]_{(k,\ell)}, k,\ell=-\N,...,\N$ for 
$
\vg = \transp{[g_{(-\N,-\N)}, g_{(-\N,-\N+1)}, ...,g_{(-\N,\N)},g_{(-\N+1,-\N)},...,g_{(\N,\N)}]}. 
$
For a complex number $b$ with  polar decomposition $b = |b|e^{i2\pi \phi}$, $\sign(b) \defeq e^{i2\pi \phi}$. Similarly, for a vector $\vb$, $[\sign(\vb)]_k \defeq \sign([\vb]_k)$. For the set $\T$, $|\T|$ designates its cardinality and $\comp{\T}$ is its complement. The sinc-function is denoted as $\sinc(t)  = \frac{\sin(\pi t)}{\pi t}$. 
For vectors $\vr, \vr' \in [0,1]^2$, $\infdist{\vr - \vr'}=\max(|r_1-r'_1|,|r_2-r'_2|)$. Here, $|x-y|$ is the wrap-around distance on the unit circle between two scalars $x$ and $y$. For example, $|3/4-1/2|=1/4$ but $|5/6-1/6|=1/3\neq 2/3$. Throughout, $\vr$ denotes a 2-dimensional vector with entries $\x$ and $\y$, i.e., $\vr = \transp{[\x,\y]}$. 
Moreover $c,\tilde c, c', c_1,c_2,...$ are numerical constants that can take on different values at different occurrences. Finally, $\mathcal N(\mu, \sigma^2)$ is the Gaussian distribution with mean $\mu$ and variance $\sigma^2$. 

\section{Recovery via convex optimization \label{sec:recovery}}

In this section we present our approach to the recovery of the parameters $(b_j, \tau_j,\nu_j)$ from the samples $y_p$ in \eqref{eq:periorel}. Before we proceed we note that \eqref{eq:periorel} can be rewritten as (see Appendix \ref{app:discspfunc} for a detailed derivation) 
\begin{align}
y_p =
\sum_{j=1}^{\S} b_j
\sum_{k,\ell = -\N}^\N 
D_{\N} \! \left( \frac{\ell}{\L} - \tau_j \right) D_{\N} \! \left( \frac{k}{\L} - \nu_j \right)  x_{p- \ell}  e^{i2\pi \frac{ pk}{\L}}, \quad p = -\N,...,\N,
\label{eq:iowithdirchkernel1}
\end{align}
%
%
where 
\begin{align}
D_{\N}(t) 
\defeq 
\frac{1}{L} \sum_{k=-\N}^{\N} e^{i2\pi t k}
\label{eq:defDirichlet}
\end{align}
is the Dirichlet kernel.
\newcommand{\setA}{\mathcal A}
We define atoms $\va \in \complexset^{\L^2}$ as 
\begin{align}
[\va(\vr)]_{(k,\ell)} = 
D_{\N} \! \left( \frac{\ell}{\L} - \x \right) D_{\N} \! \left( \frac{k }{\L} -  \y \right), \quad \vr = \transp{[\x,\y]},\quad k,\ell = -\N,...,\N. 
\label{eq:defatoms}
\end{align}
Rewriting \eqref{eq:iowithdirchkernel1} in matrix-vector form yields 
\[
\vy = \mG \vz, \quad \vz =   \sum_{j=1}^{\S}  |b_j| e^{i2\pi \phi_j} \va(\vr_j), \quad \vr_j = \transp{[\x_j,\y_j]}, 
\]
where $b_j = |b_j|e^{i2\pi \phi_j}$ is the polar decomposition of $b_j$ and $\mG \in \complexset^{\L \times \L^2}$ is the Gabor matrix defined by 
\begin{align}
[\mG]_{p, (k,\ell)} \defeq x_{p- \ell}  e^{i2\pi \frac{k p}{\L}}, \quad k,\ell, p = -\N,...,\N. 
\label{eq:defgabormtx}
\end{align}
The signal $\mathbf{z}$ is a sparse linear combination of time and frequency shifted versions of the atoms $\va(\vr)$. A regularizer that promotes such a sparse linear combination is the atomic norm induced by these signals~\cite{chandrasekaran_convex_2012}. The atoms in the set $\setA \defeq \{ e^{i2\pi \phi} \va(\vr), \vr \in [0,1]^2,\phi \in [0,1]\}$ are the building blocks of the signal~$\vz$. The atomic norm $\norm[\setA]{\cdot}$ is defined as 
\[
\norm[\setA]{\vz} 
= \inf \left\{ t > 0\colon \vz \in t\, \mathrm{conv}(\setA) \right\}
= \inf_{b_j \in \complexset, \vr_j \in [0,1]^2} \left\{ \sum_j |b_j| \colon \vz = \sum_j b_j \va(\vr_j) \right\},
\]
where $\mathrm{conv}(\setA)$ denotes the convex hull of the set $\setA$. 
The atomic norm can enforce sparsity in $\setA$ because low-dimensional faces of $\mathrm{conv}(\setA)$ correspond to signals involving only a few atoms \cite{chandrasekaran_convex_2012,tang_compressed_2013}. 
A natural algorithm for estimating $\vz$ 
is the atomic norm minimization problem \cite{chandrasekaran_convex_2012}
\newcommand{\AN}{\mathrm{AN}}
\begin{align}
\AN(\vy) \colon \;\; \underset{\minlet{\vz}  }{\text{minimize}} \,  \norm[\setA]{\minlet{\vz} } \; \text{ subject to } \; \vy = \mG \minlet{\vz}.
\label{eq:primal}
\end{align}
Once we obtain $\vz$, the recovery of the time-frequency shifts is a 2D line spectral estimation problem that can be solved with standard approaches such as Prony's method, see e.g.~\cite[Ch.~2]{gershman_space-time_2005}. In Section \ref{sec:estfromdual}, we will present a more direct approach for recovering the time-frequency shifts $\mathbf{r}_j$. When the time-frequency shifts $\vr_j$ are identified, the coefficients $b_j$ can be obtained by solving the linear system of equations 
\[
\vy = \sum_{j=1}^\S b_j \mG \va(\vr_j). 
\]
Computation of the atomic norm involves taking the infimum over infinitely many parameters and may appear to be daunting. 
For the case of only time or only frequency shifts (cf.~Section \ref{sec:redsupres}), the atomic norm can be characterized in terms of linear matrix inequalities \cite[Prop.~2.1]{tang_compressed_2013}; this allows us to formulate the atomic norm minimization program as a semidefinite program that can be solved efficiently. 
The characterization \cite[Prop.~2.1]{tang_compressed_2013} relies on a classical Vandermonde decomposition lemma for Toeplitz matrices by Carath\'eodory and Fej\'er. 
While this lemma generalizes to higher dimensions \cite[Thm.~1]{yang_vandermonde_2015}, this generalization fundamentally comes with a rank constraint on the corresponding Toeplitz matrix. This appears to prohibit a characterization of the atomic norm paralleling that of \cite[Prop.~2.1]{tang_compressed_2013} which explains why no semidefinite programming formulation of the atomic norm minimization problem \eqref{eq:primal} is known, to the best of our knowledge. 
Nevertheless, based on \cite[Thm.~1]{yang_vandermonde_2015}, one can obtain a semidefinite programming \emph{relaxation} of $\AN(\vy)$. 

Instead of taking that route, and explicitly stating the corresponding semidefinite program, we show in Section \ref{sec:estfromdual} that the time-frequency shifts $\vr_j$ can be identified directly from the dual solution of the atomic norm minimization problem $\AN(\vy)$ in \eqref{eq:primal}, and propose a semidefinite programming \emph{relaxation}  that allows us to find a solution of the dual efficiently. 



\section{Main result \label{sec:prefres}}

Our main result, stated below, provides conditions guaranteeing that the solution to $\AN(\vy)$ in \eqref{eq:primal} is $\vz$ (perfect recovery).
%
As explained in Section \ref{sec:recovery}, from $\vz$ we can obtain the triplets $(b_j, \tau_j,\nu_j)$ easily. 

\begin{theorem} Assume that the samples of the probing signal $x_\ell, \ell =-N,...,N$, are i.i.d.~$\mathcal N(0,1/\L)$ random variables, $\L=2\N+1$. Let $\mathbf{y}\in\mathbb{C}^L$, with $L \geq 1024$, contain the samples of the output signal obeying the input-output relation \eqref{eq:periorel}, i.e., 
\[
\vy = \mG \vz, \quad \vz = \sum_{\vr_j \in \T}  b_j \va(\vr_j),
\]
where $\mathbf{G}$ is the Gabor matrix of time-frequency shifts of the input sequence $x_\ell$ defined in \eqref{eq:defgabormtx}. Assume that the $\sign(b_j)$ are i.i.d.~uniform on $\{-1,1\}$ and that the set of time-frequency shifts $\T = \{\vr_1, \vr_2,...,\vr_\S \} \allowbreak \subset [0,1]^2$ obeys the minimum separation condition 
\begin{equation}
\max(|\tau_j - \tau_{j'}|, |\nu_j - \nu_{j'}| ) \geq \frac{2.38}{\N}
\text{ for all } [\tau_j, \nu_j], [\tau_{j'}, \nu_{j'}] \in \T  \text{ with } j \neq j'.
\label{eq:minsepcond1}
\end{equation}
Furthermore, choose $\delta > 0$ and assume  
\begin{align*}
S\le c\frac{L}{(\log(L^6/\delta))^3}, 
\end{align*}
where $c$ is a numerical constant. Then, with probability at least $1-\delta$, $\vz$ is the unique minimizer of $\AN(\vy)$ in \eqref{eq:primal}. 
\label{thm:mainres}
\end{theorem}

Recall that the complex-valued coefficients $b_j$ in \eqref{eq:iorelintro} in the radar model describe the attenuation factors. 
Therefore, it is natural to assume that the phases of different $b_j$ are independent from each other and are uniformly distributed on the unit circle of the complex plane. Indeed, in standard models in wireless communication and radar~\cite{bello_characterization_1963}, the $b_j$ are assumed to be complex Gaussian. To keep the proof simple, in Theorem~\ref{thm:mainres} we assume that the $b_j$ are real-valued. The assumption that the coefficients $b_j$ have random sign is the real-valued analogue of the random phase assumption discussed above. Theorem~\ref{thm:mainres} continues to hold for complex-valued $b_j$ (only the constant $2.38$ in \eqref{eq:minsepcond1} changes slightly). While this random sign assumption is natural for many applications, we believe that is not necessary for our result to hold. Finally, we would like to point out that Theorem \ref{thm:mainres} continues to hold for sub-Gaussian sequences $x_\ell$ with trivial modifications to our proof. 

The proof of Theorem \ref{thm:mainres} is based on analyzing the dual of $\AN(\vy)$. We will prove that the recovery is perfect by constructing an appropriate dual certificate. The existence of this dual certificate guarantees that the solution to $\AN(\vy)$ in \eqref{eq:primal} is $\vz$. 
This is a standard approach, e.g., in the compressed sensing literature, the existence of a related dual certificate guarantees that the solution to $\ell_1$-minimization is exact \cite{candes_robust_2006}. 
Specifically, the dual of $\AN(\vy)$ in \eqref{eq:primal} is \cite[Sec.~5.1.16]{boyd_convex_2004}
\begin{align}
\underset{\vq}{\text{maximize}} \; \Re \innerprod{\vq}{\vy} \text{ subject to } \norm[\setA^\ast]{\herm{\mG} \vq} \leq 1, 
\label{eq:dual}
\end{align}
where $\vq = \transp{[ q_{-\N}, ..., q_{\N}]}$ and 
\begin{align*}
\norm[\setA^\ast]{\vv}  
= \sup_{\norm[\setA]{\vz} \leq 1} \Re \innerprod{\vv}{\vz} 
= \sup_{\vr \in [0,1]^2} \left| \innerprod{\vv}{\va(\vr)} \right| 
\end{align*}
is the dual norm. 
Note that the constraint of the dual \eqref{eq:dual} can be rewritten as: 
\begin{align}
\norm[\setA^\ast]{\herm{\mG} \vq} 
= \sup_{\vr \in [0,1]^2} \left| \innerprod{ \vq}{\mG \va(\vr)} \right|
= \sup_{[\tau,\nu] \in [0,1]^2} \left| \innerprod{\vq}{ \mc F_\nu \mc T_\tau  \vx } \right| 
\leq 1, 
\label{eq:constrdual}
\end{align}
where we used $\mG \va(\vr) =  \mc F_\nu \mc T_\tau  \vx$. By definition of the time and frequency shifts in \eqref{eq:deftimefreqshifts} it is seen that $\innerprod{\vq}{ \mc F_\nu \mc T_\tau  \vx }$ is a 2D trigonometric polynomial (in $\tau,\nu$, see   \eqref{eq:dualpolyinprop} for its specific form). 
The constraint in the dual is therefore equivalent to the requirement that the absolute value of a specific 2D trigonometric polynomial is bounded by one. 
A sufficient condition for the success of atomic norm minimization is given by the existence of a certain dual certificate of the form $\innerprod{\vq}{ \mc F_\nu \mc T_\tau  \vx }$. This is formalized by Proposition \ref{prop:dualmin} below and is a consequence of strong duality. Strong duality is implied by Slater's conditions being satisfied \cite[Sec.~5.2.3]{boyd_convex_2004} (the primal problem only has equality constraints). 


\begin{proposition} 
Let $\vy=\mG\vz$ with $\vz =   \sum_{\vr_j \in \T}  b_j  \va(\vr_j)$. If there exists a dual polynomial
$
Q(\vr) =  \innerprod{\vq}{\mc F_{\nu} \mc T_{\tau} \vx } 
$
with complex coefficients $\vq = \transp{[q_{-\N}, ..., q_{\N}]}$ such that 
\begin{align}
Q(\vr_j) = \sign(b_j), \text{ for all } \vr_j \in \T, \text{ and } |Q(\vr)| < 1 \text{ for all } \vr \in [0,1]^2 \setminus \T
\label{eq:dualpolyinatmincon}
\end{align}
then $\vz$ is the unique minimizer of $\AN(\vy)$. 
Moreover, $\vq$ is a dual optimal solution. 
\label{prop:dualmin}
\end{proposition}
The proof of Proposition \ref{prop:dualmin} is standard, see e.g., \cite[Proof of Prop.~2.4]{tang_compressed_2013}, and is provided in Appendix \ref{sec:proofprop:dualmin} for convenience of the reader. 
The proof of Theorem \ref{thm:mainres} consists of  constructing a dual polynomial satisfying the conditions of Proposition \ref{prop:dualmin}, see Section \ref{app:proofmainres}.

\section{
Relationship between the continuous time and discrete time models
\label{sec:probform}
}

In this section we discuss in more detail how the discrete time model \eqref{eq:periorel} follows from the continuous time model \eqref{eq:iorelintro} through band- and time-limitations. This section is geared towards readers with interest in radar and wireless communication applications. Readers not interested in the detailed justification of \eqref{eq:periorel} may wish to skip this section on a first reading.
We start by 
explaining the effect of band- and time-limitation in continuous time. We show that \eqref{eq:periorel} holds \emph{exactly} when the probing signal is $T$-periodic, and holds \emph{approximately} when the probing signal is essentially time-limited on an interval of length $3T$, as discussed in the introduction. Finally, we explicitly quantify the corresponding approximation error. 

As mentioned previously, radar systems and wireless communication channels are typically modeled as linear systems whose response is a weighted superposition of delayed and Doppler-shifted versions of the probing signal. In general, the response $y=Hx$ of the system $H$ to the probing signal $\prsig$ is given by
\begin{equation}
y(t) = \iint   \sfunc (\tau,\nu) \prsig(t-\tau)  e^{i2\pi \nu t}  d\nu d\tau,
\label{eq:ltvsys}
\end{equation}
where $\sfunc(\tau,\nu)$ denotes the spreading function associated with the system. 
In the channel identification and radar problems, 
the probing signal $x$ can be controlled by the system engineer and is known. The spreading function depends on the scene and is unknown. 
We assume that the spreading function consists of $\S$ point scatterers. In radar, these point scatterers correspond to moving targets. Mathematically, this means that the spreading function specializes to 
\begin{align}
\sfunc(\tau,\nu) = \sum_{j=1}^{\S} b_j \delta(\tau-\tauc_j) \delta(\nu-\nuc_j).
\label{eq:origradarspreadfunc}
\end{align}
Here, $b_j$, $j=1,...,\S$, are (complex-valued) attenuation factors associated with the delay-Doppler pair $(\tauc_j,\nuc_j)$. 
With \eqref{eq:origradarspreadfunc}, \eqref{eq:ltvsys} reduces to the input-output relation \eqref{eq:iorelintro} stated in the introduction, i.e., to
\[
y(t) = \sum_{j=1}^\S  b_j \prsig(t-\tauc_j)  e^{i2\pi \nuc_j t}  .
\]

\subsection{The effect of band- and time-limitations}

In practice, the probing signal $x$ has finite bandwidth $B$ and the received signal $y$ can only be observed over a finite time interval of length $T$. 
We refer to this as time-limitation, even though $y$ is non-zero outside of the time interval of length $T$.  
As shown next, this band- and time-limitations lead to a discretization of the input-output relation \eqref{eq:ltvsys} and determine the ``natural'' resolution of the system of $1/\Bint$ and $1/\Tint$ in $\tau$- and $\nu$-directions, respectively. 
 Specifically, using the fact that $\prsig$ is band-limited to $[-\Bint/2,\Bint/2]$, \eqref{eq:ltvsys} can be rewritten in the form
\begin{equation}
y(t) = 
\sum_{k \in \mb Z}\sum_{\ell \in \mb Z} \overline{\sfunc}\! \left(\frac{\ell}{B},\frac{k}{\Tint} \right) \prsig \!\left( t-\frac{\ell}{B}  \right) e^{i2\pi \frac{k}{\Tint} t},
\label{eq:sfunc_discrete}
\end{equation}
with $t \in [-\Tint/2,\Tint/2]$. Here,
\begin{align}
\overline{\sfunc}(\tau,\nu) \label{eq:btlimiorel} 
\defeq 
\iint \! \sfunc(\tau',\nu')    \sinc( (\tau \!- \!\tau') B ) \sinc( (\nu \!-\! \nu')  \Tint  )  d\tau' d\nu'
\end{align}
is a smeared version of the original spreading function. 
The relation \eqref{eq:sfunc_discrete} appears in  \cite{bello_characterization_1963}; for the sake of completeness, we detail the  derivations leading to \eqref{eq:sfunc_discrete} in Appendix \ref{app:iorelproof}. 
For point scatterers, i.e., for the spreading function in \eqref{eq:origradarspreadfunc}, \eqref{eq:btlimiorel} specializes to
\begin{align}
\overline{\sfunc}(\tau,\nu)  = 
\sum_{j=1}^{\S} b_j    \sinc( (\tau - \tauc_j) B ) \sinc( (\nu - \nuc_j)  \Tint  ).
\label{eq:specialradar}
\end{align}
%
%
\begin{figure}
\centering
\includegraphics[width=\textwidth]{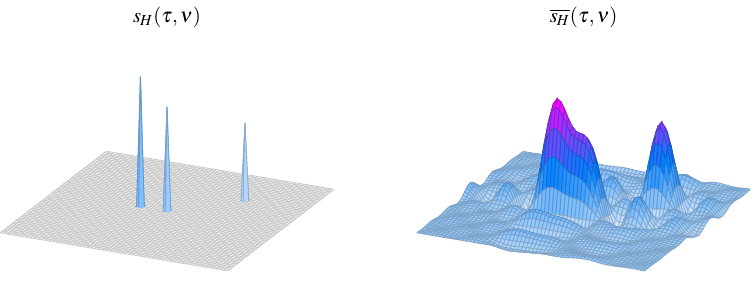}
\caption{\label{fig:intersect} Illustration of the spreading function $\sfunc(\tau,\nu)$ and the corresponding smeared spreading function $\overline{\sfunc}(\tau,\nu)$.}
\end{figure}%
%
%
Imagine for a moment that we could measure $\overline{\sfunc}(\tau,\nu)$ directly. We see that $\overline{\sfunc}(\tau,\nu)$ is the 2D low-pass-filtered version of the signal $\sfunc(\tau,\nu)$ in~\eqref{eq:origradarspreadfunc}, where the filter has resolution $1/B$ in $\tau$ direction and resolution $1/T$ in $\nu$ direction; see Figure \ref{fig:intersect} for an illustration. 
Estimation of the triplets $(b_j,\tauc_j,\nuc_j),\ j=1,\ldots,S$,  from $\overline{\sfunc}(\tau,\nu)$ is the classical 2D line spectral estimation problem. 
In our setup,  the situation is further complicated by the fact that we can not measure $\overline{\sfunc}(\tau,\nu)$ directly. We only get access to $\overline{\sfunc}(\tau,\nu)$ after the application of the Gabor linear operator in~\eqref{eq:sfunc_discrete}. 
 
%
%
%

\subsection{Choice of probing signal}

Next, we consider a probing signal that is band-limited and essentially time-limited to an interval of length $3T$. To be concrete, we consider the signal 
\begin{align}
\TL{x}(t) = \sum_{\ell=-\L-\N}^{\L+\N} x_\ell \sinc(tB - \ell)
\label{eq:truncsignal}
\end{align}
where the coefficients $x_\ell$ are $L$-periodic, with the $x_\ell, \ell = -\N,...,\N$, i.i.d.~$\mc N(0,1/\L)$. A realization of the random signal $\TL{x}(t)$ along with its Fourier transform is depicted in Figure \ref{fig:prsig}. 
Since the sinc-kernel above is band-limited to $[-\frac{B}{2},\frac{B}{2}]$, $\TL{x}$ is band-limited to $[-\frac{B}{2},\frac{B}{2}]$ as well. 
As the sinc-kernel decays relatively fast, $\TL{x}$ is essentially supported on the interval $[-\frac{3T}{2},\frac{3T}{2}]$. 
We hasten to add that there is nothing fundamental about using the sinc-kernel to construct $\TL{x}$ here;  
we choose it out of mathematical convenience, and could as well use a kernel that decays faster in time. 
For the readers familiar with wireless communication, we point out that the partial periodization of $\TL{x}$ serves a similar purpose as the cyclic prefix used in OFDM systems. 

\begin{figure}
\centering
\includegraphics{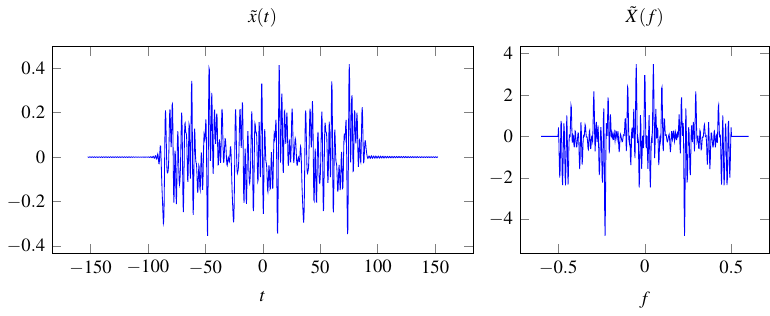}
%
%
%
%
\caption{\label{fig:prsig} 
The probing signal $\TL{x}(t)$ and the real part of its Fourier transform $\TL{X}(f)$ for $B=1, T = 61, \L = 61$: $\TL{x}(t)$ is essentially time-limited on an interval of length $3T$ and band-limited to $[-B/2,B/2]$.  
}
\end{figure}

\subsection{Sampling the output}
Recall that delay-Doppler pairs satisfy, by assumption, $(\tauc_j, \nuc_j) \in [-\frac{T}{2},\frac{T}{2}]\times [\frac{B}{2},\frac{B}{2}]$. Thus, $H\TL{x}$ is band-limited to $[-B,B]$ and approximately time-limited to $[-2T,2T]$. According to the $2WT$-Theorem \cite{slepian_bandwidth_1976,durisi_sensitivity_2012}, $H \TL{x}$ 
has on the order of $BT$ degrees of freedom so that one can think of $H \TL{x}$ as having effective dimensionality on the order of $BT$. 
Therefore, by sampling $H\TL{x}$ at rate $1/B$ in the interval $[-T/2,T/2]$, we collect the number of samples that matches the dimensionality of $H\TL{x}$ up to a constant.  The samples $\TL{y}_p \defeq (H\TL{x})(p/B)$ are given by
\begin{equation}
\TL{y}_p
=
\sum_{k,\ell \in \mb Z}  
 \overline{\sfunc} \left(\frac{\ell}{B},\frac{k}{\Tint} \right) 
 \TL{\prsig} \left(\frac{p-\ell}{B}  \right) 
 e^{i2\pi \frac{k p}{\Bint \Tint}},\quad  p = -\N,...,\N,
\label{eq:sfunc_discretesample}
\end{equation}
where $\N = (\L-1)/2$ and $\L = BT$. 
Substituting \eqref{eq:specialradar} into \eqref{eq:sfunc_discretesample} 
and defining $\TL{x}_{\ell} := \TL{x}(\ell/B)$ yields
\begin{align}
\TL{y}_p 
= 
\sum_{j=1}^{\S} b_j
\sum_{k,\ell \in \mb Z} 
 \sinc(\ell - \tauc_j B) \sinc(k - \nuc_j T) \, \TL{x}_{p-\ell} e^{i2\pi \frac{k p}{\Bint \Tint}}.  
\label{eq:sfunc_discretesample2}
\end{align}
Next, we rewrite \eqref{eq:sfunc_discretesample2} in terms of the following equivalent expression of the Dirichlet kernel defined in \eqref{eq:defDirichlet}
\begin{align}
D_\N( t ) = 
\sum_{k\in \mb Z} 
\sinc\left(  \L \left( t  - k \right) \right),
\label{eq:dirichletsum}
\end{align}
and a truncated version 
\begin{align}
\quad 
\tilde D_\N( t ) 
\defeq
\sum_{k=-1}^1 
\sinc\left(  \L \left( t  - k \right) \right),
\label{eq:dirichletsum1}
\end{align}
as 
\begin{align}
\TL{y}_p =
\sum_{k,\ell = -\N}^\N 
\sum_{j=1}^{\S} b_j
\tilde D_{\N} \! \left( \frac{p-\ell}{\L} - \tau_j \right) D_{\N} \! \left( \frac{k }{\L} - \nu_j \right)  x_{\ell}  e^{i2\pi \frac{k p}{\L}}, \quad p = -\N,...,\N,
\label{eq:iowithdirchkernelapprox}
\end{align}
where $\tau_j = \tauc_j/T$, $\nu_j =\nuc_j/B$, as before (see Appendix \ref{eq:proofeq:iowithdirchkernel} for details). 
For $t \in [-1.5, 1.5]$, $\TL{D}_\N( t )$ is well approximated by $D_\N( t )$, therefore 
\eqref{eq:iowithdirchkernelapprox} is well approximated by 
\begin{align}
y_p =
\sum_{k,\ell = -\N}^\N 
\sum_{j=1}^{\S} b_j
D_{\N} \! \left( \frac{\ell}{\L} - \tau_j \right) D_{\N} \! \left( \frac{k }{\L} - \nu_j \right)  x_{p- \ell}  e^{i2\pi \frac{k p}{\L}}, \quad p = -\N,...,\N.
\label{eq:iowithdirchkernel}
\end{align}
This is equivalent to \eqref{eq:periorel}, as already mentioned in Appendix \ref{app:discspfunc}. 
Note that \eqref{eq:iowithdirchkernel} can be viewed as the periodic equivalent of \eqref{eq:sfunc_discrete} (with $\overline{\sfunc}(\tau,\nu)$ given by \eqref{eq:specialradar}). 
We show in~Appendix \ref{eq:proofeq:iowithdirchkernel} that the discretization~\eqref{eq:iowithdirchkernel} can be obtained from \eqref{eq:sfunc_discrete} without any approximation if the probing signal $x(t)$ is chosen to be $T$-periodic with its samples selected as $x(\ell/B) = x_\ell$, for all $\ell$. Recall that the samples of the quasi-periodic probing signal satisfy $\TL{x}(\ell/B) = x_\ell$ for $\ell \in [-\N -\L, \L +\N]$ and $\TL{x}(\ell/B) = 0$ otherwise. Clearly, the periodic probing signal is not time-limited and, hence, can not be used in practice. 
The error we make by approximating $\TL{y}_p$ with $y_p$ is very small, as shown by the following proposition, proven in Appendix \ref{app:errbound}. 
We also show numerically in Section \ref{sec:numres} that in practice, the error we make by approximating $\TL{y}_p$ with $y_p$ is negligible. 

\begin{proposition}
Let the $x_\ell$ be i.i.d.~$\mc N(0,1/\L)$ and let the sign of $b_j$ be i.i.d.~uniform on $\{-1,1\}$. For all $\alpha>0$, the difference between $\TL{y}_p$ in \eqref{eq:iowithdirchkernelapprox} and $y_p$ in \eqref{eq:iowithdirchkernel} satisfies 
\[
\PR{|y_p - \TL{y}_p| 
\geq
c \frac{\alpha}{\L}
\norm[2]{\vb} 
}
\leq (4 + 2 \L^2) e^{-\frac{\alpha}{2}},
\]
where $\vb = \transp{[b_1,...,b_\S]}$ and $c$ is a numerical constant. 
\label{prop:errbound}
\end{proposition}

Proposition \ref{prop:errbound} ensures that with high probability, the  $\TL{y}_p$ are very close to the $y_p$. Note that, under the conditions of Theorem \ref{thm:mainres}, 
$\norm[2]{\vy} \approx \norm[2]{\vb}, \vy = \transp{[y_{-\N},...,y_{\N}]}$,  with high probability (not shown here). Since it follows from Proposition \ref{prop:errbound} that 
\[
\norm[2]{\vy - \TL{\vy} } \leq \frac{c_1 }{\sqrt{\L}} \norm[2]{\vb}
\]
holds with high probability, we can also conclude that $\frac{\norm[2]{\vy - \TL{\vy}} }{\norm[2]{\vy} } \leq \frac{c_2}{\sqrt{\L}},  \TL{\vy} = \transp{[\TL{y}_{-\N},...,\TL{y}_{\N}]}$ holds with high probability. That is, the relative error we make in approximating the $\TL{y}_p$ with the $y_p$ tends to zero in $\ell_2$ norm as $\L$ grows. 

\newcommand{\dT}{\mathcal S} 

\section{Super-resolution radar \label{sec:discretesuperes} on a grid}

One approach to estimate the triplets $(b_j,\tau_j, \nu_j)$ from the samples $y_p$ in \eqref{eq:periorel} is to suppose the time-frequency shifts lie on a \emph{fine} grid, and solve the problem on that grid. 
When the time-frequency shifts do not exactly lie on a grid this leads to an error that becomes small as the grid becomes finer. In Section \ref{sec:numres} below we study numerically the error incurred by this approach; for an analysis of the corresponding error in related problems see \cite{tang_sparse_2013}. 
Our results have immediate consequences for the corresponding (discrete) sparse signal recovery problem, as we discuss next.

Suppose we want to recover a sparse discrete signal $s_{m,n} \in \complexset,\; m,n = 0,...,\K-1$, $\K \geq \L=2\N+1$, from samples of the form 
\begin{align}
y_p 
&= 
\sum_{m,n=0}^{\K-1} s_{m,n}  
[\mc F_{m/K}
\mc T_{n/K}
\vx]_p, \quad p=-\N,...,\N. \\ 
 &= \sum_{m,n=0}^{\K-1} \left(
 e^{i2\pi p \frac{m}{\K}  }
 \frac{1}{L} \sum_{\ell,k=-\N}^{\N}  e^{-i2\pi k  \frac{n}{\K} }   e^{i2\pi (p-\ell) \frac{k}{\L}  }  x_{\ell} \right) s_{m,n}, \quad p=-\N,...,\N.
\label{eq:perioreldisc}
\end{align}
To see the connection to the gridless setup in the previous sections, note that the recovery of the $\S$-sparse (discrete) signal $s_{m,n}$ is equivalent to the recovery of the triplets $(b_j, \tau_j,\nu_j)$ from the samples $y_p$ in \eqref{eq:periorel} if we assume that the $(\tau_j,\nu_j)$ lie on a $(1/K,1/K)$ grid
(the non-zeros of $s_{m,n}$ correspond to the $b_j$). 
Writing the relation  \eqref{eq:perioreldisc} in matrix-vector form yields 
\[
\vy = \mR \vs
\]
where $[\vy]_{p} \defeq y_p$, $[\vs]_{(m,n)} \defeq s_{m,n}$, and $\mR \in \complexset^{\L \times \K^2}, \K \geq \L$, is the matrix with $(m,n)$-th column given by $\mc F_{m/K}
\mc T_{n/K} \vx$. 
The matrix $\mR$ contains as columns ``fractional'' time-frequency shifts of the sequence $x_\ell$. If $\K = \L$, $\mR$ contains as columns only ``whole'' time-frequency shifts of the sequence $x_\ell$ and $\mR$ is equal to the Gabor matrix  $\mG$ defined by \eqref{eq:defgabormtx}. 
In this sense, $\K = \L$ is the natural grid (cf.~Section \ref{sec:ongrid}) and the ratio $\SRF \defeq\K/\L$ can be interpreted as a super-resolution factor. The super-resolution factor determines by how much the $(1/\K,1/\K)$ grid is finer than the original $(1/\L,1/\L)$ grid.

A standard approach to the recovery of the sparse signal $\vs$ from the underdetermined linear system of equations $\vy = \mR \vs$ is to solve the following convex program:
\begin{align}
	\mathrm{L1}(\vy) \colon \;\;
\underset{\minlet{\vs}}{\text{minimize}} \; \norm[1]{\minlet{\vs}} \text{ subject to } \vy = \mR \minlet{\vs}.
\label{eq:l1minmG}
\end{align}
The following theorem is our main result for recovery on the fine grid.   
\begin{theorem} Assume that the samples of the probing signal $x_\ell = -\N,...,\N$, in \eqref{eq:perioreldisc} are i.i.d. $\mathcal N(0,1/\L)$ random variables, $L=2N+1$. Let $\mathbf{y}\in\mathbb{C}^\L$, with $L\geq 1024$ be the samples of the output signal obeying the input-output relationship \eqref{eq:perioreldisc}, i.e., $\vy = \mathbf{R}\mathbf{s}$.   
Let $\mathcal{S} \subseteq \{0,...,\K-1\}^2$ be the support of the vector $[\vs]_{(m,n)}$, $m,n=0,...,\K-1$, and suppose that it satisfies the minimum separation condition
\[
\min_{(m,n), (m', n') \in \mathcal{S} \colon (m,n) \neq ( m', n')} 
\frac{1}{\K} \max(|m- m'|, |n - n'|) \geq \frac{2.38}{\N}.
\]
Moreover, suppose that the non-zeros of $\vs$ have random sign, i.e., $\sign([\vs]_{(m,n)}), (m,n) \in \dT$ are i.i.d.~uniform on $\{-1,1\}$. 
Choose $\delta>0$ and assume
\begin{align*}
S\le c\frac{L}{(\log(L^6/\delta))^3}, 
\end{align*}
where $c$ is a numerical constant. Then, with probability at least $1-\delta$, $\mathbf{s}$ is the unique minimizer of $\mathrm{L1}(\vy)$ in \eqref{eq:l1minmG}. 
\label{cor:discretesuperres}
\end{theorem}


The proof of Theorem \ref{cor:discretesuperres}, presented in Appendix \ref{sec:proofdiscrete}, is closely linked to that of Theorem \ref{thm:mainres}. 
As reviewed in Appendix \ref{sec:proofdiscrete}, the existence of a certain dual certificate guarantees that $\vs$ is the unique minimizer of $\mathrm{L1}(\vy)$ in \eqref{eq:l1minmG}. 
The dual certificate is obtained directly from the dual polynomial for the continuous case (i.e., from Proposition \ref{prop:dualpolynomial} in Section \ref{app:proofmainres}). 

\subsection{Implementation details}

The matrix $\mR$ has dimension $\L \times \K^2$, thus as the grid becomes finer (i.e., $\K$ becomes larger) the complexity of solving \eqref{eq:l1minmG} increases. 
The complexity of solving \eqref{eq:l1minmG} can be managed as follows. First, the complexity of first-order convex optimization algorithms (such as TFOCS \cite{becker_templates_2011}) for solving \eqref{eq:l1minmG} is dominated by multiplications with the matrices $\mR$ and $\herm\mR$. Due to the structure of $\mR$, those multiplications can be done very efficiently utilizing the fast Fourier transform. Second, in practice we have 
$
(\tauc_j, \nuc_j) \in [0,\taum]\times [0, \num]
$, 
which means that 
\begin{align}
(\tau_j, \nu_j) \in 
\left[0, \frac{\taum  }{ T }  \right] 
\times 
\left[0, \frac{\num }{ B}  \right] .
\label{eq:restaunun}
\end{align}
It is therefore sufficient to consider the restriction of $\mR$ to the 
$\frac{\taum \num K^2}{BT} = \taum \num \L \cdot\SRF^2$ many columns corresponding to the $(\tau_j, \nu_j)$ satisfying  \eqref{eq:restaunun}. Since typically $\taum \num \ll BT=\L$, this results in a significant reduction of the problem size.

\subsection{Numerical results}
\label{sec:numres}

We next evaluate numerically the resolution obtained by  our approach. 
We consider identification of the time-frequency shifts from the response to the (essentially) time-limited signal in \eqref{eq:truncsignal}, which corresponds to identification from the samples  $\TL{y}_p, p=-\N,...,\N$ given by \eqref{eq:iowithdirchkernelapprox}, without and with additive Gaussian noise. 
We also consider identification from the response to a signal with $\L$-periodic samples $x(\ell/B) = x_\ell$, which corresponds to identification from the samples $y_p, p=-\N,...,\N$ in \eqref{eq:iowithdirchkernel}. 
To account for additive noise, we solve the following modification of $\mathrm{L1}(\vy)$ in \eqref{eq:l1minmG}
\begin{align}
\text{L1-ERR}\colon \underset{\minlet{\vs}}{\text{minimize}} \; \norm[1]{\minlet{\vs}} \text{ subject to } 
\norm[2]{\vy - \mR \minlet{\vs} }^2 \leq \delta,
\label{eq:BDDN}
\end{align}
with $\delta$ chosen on the order of the noise variance.  
We choose $\L = 201$, and for each problem instance, we draw $S=10$ 
time-frequency shifts $(\tau_j,\nu_j)$ uniformly at random from $[0,2/\sqrt{201}]^2$, which amounts to drawing the corresponding delay-Doppler pairs $(\tauc_j,\nuc_j)$ from $[0,2]\times [0,2]$. 
We use SPGL1 \cite{BergFriedlander:2008} to solve \text{L1-ERR}. 
The attenuation factors $b_j$ corresponding to the time-frequency shifts $(\tau_j,\nu_j)$ are drawn uniformly at random from the complex unit disc, independently across $j$. 
In Figure \ref{fig:realrecov} we plot the average resolution error versus the super-resolution factor $\SRF = K/L$. The resolution error is defined as 
the average over $j=1,...,S$ of $\L \sqrt{ (\hat \tau_j - \tau_j)^2 + (\hat \nu_j - \nu_j)^2}$, where the $(\hat \tau_j, \hat \nu_j)$ are the time-frequency shifts obtained by solving \eqref{eq:BDDN}. 
There are three error sources incurred by this approach. The first is the gridding error obtained by assuming the points lie on a fine grid with grid constant $(1/\K,1/\K)$. 
The second is the model error from approximating the $\TL{y}_p$ in \eqref{eq:iowithdirchkernelapprox}, obtained by sending an essentially time-limited probing signal (cf.~\eqref{eq:truncsignal}), with the $y_p$ in \eqref{eq:iowithdirchkernel}, obtained by sending a truly periodic input signal $x(t)$. The third is the additive noise error. 
Note that the resolution attained at $\SRF=1$ corresponds to the resolution attained by matched filtering and by the compressive sensing radar architecture \cite{herman_high-resolution_2009} discussed in Section \ref{sec:ongrid}. 
We see that for all $\SRF>1$, the resolution is significantly improved using our super-resolution radar approach. 
We furthermore observe that for low values of $\SRF$, the gridding error dominates, while for large values of $\SRF$, the additive noise error dominates. 
By looking at the noiseless case, it is seen that the gridding error decays as $1/\SRF$, e.g., at $\SRF = 20$, the error is about $0.4/20$. 
This demonstrates empirically that in practice solving the super-resolution radar problem on a fine grid is essentially as good as solving it on the continuum--provided the super-resolution factor is chosen sufficiently large. 
Finally, we observe that the model error is negligible, even for large signal-to-noise ratios (the curves in Figure \ref{fig:realrecov} corresponding to the noiseless and the periodic case are indistinguishable).

\begin{figure}
\begin{center}
\includegraphics{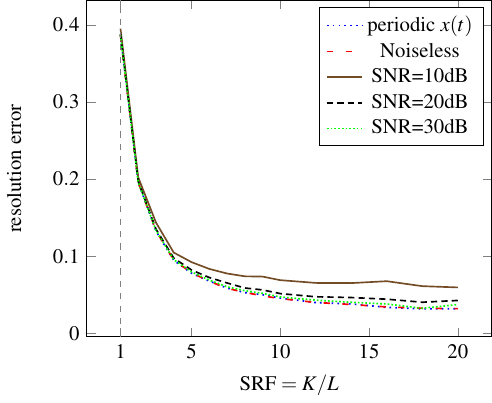}
 \end{center}
 \caption{\label{fig:realrecov} Resolution error $\L \sqrt{ (\hat \tau_j - \tau_j)^2 + (\hat \nu_j - \nu_j)^2}$ for the recovery of $\S = 10$ time-frequency shifts
 from the samples $y_p, p=-\N,...,\N$ in \eqref{eq:iowithdirchkernel} (periodic input signal $x(t)$), and identification from the samples $\TL{y}_p, p =-\N, ...,\N$ in \eqref{eq:iowithdirchkernelapprox} (essentially time-limited input signal $\TL{x}(t)$) with and without additive Gaussian noise $n_p$ of a certain signal-to-noise ratio $\text{SNR} \defeq \norm[2]{[\TL{y}_{-N},...,\TL{y}_{N}]}^2/\norm[2]{[n_{-N},...,n_{N}]}^2$, by solving \text{L1-ERR}.   
}
\end{figure}


\section{Identification of time-frequency shifts}
\label{sec:sdprel}

In this section we show that the time-frequency shifts can be obtained from a solution to the dual problem, and present a semidefinite programming relaxation to obtain a solution of the dual efficiently. 

\subsection{Semidefinite programming relaxation of the dual}
Our relaxation is similar in spirit to related convex programs in \cite[Sec.~3.1]{bhaskar_atomic_2012}, \cite[Sec.~4]{candes_towards_2014}, and \cite[Sec.~2.2]{tang_compressed_2013}. 
We show that the constraint in the dual is equivalent to the requirement that the absolute value of a specific 2D trigonometric polynomial is bounded by one, and, therefore, this constraint can be formulated as a linear matrix inequality. 
The dimensions of the corresponding matrices are, however, unspecified. 
Choosing a certain relaxation degree for those matrices, and substituting the constraint in the dual with the corresponding matrix inequality leads to a semidefinite programming relaxation of the dual. 

Recall that the constraint of the dual program  \eqref{eq:dual} is 
\[
\norm[\setA^\ast]{\herm{\mG} \vq} 
= \sup_{\vr \in [0,1]^2} \left| \innerprod{\herm{\mG} \vq}{\va(\vr)} \right|
\leq 1. 
\]
By definition of the Dirichlet kernel in \eqref{eq:defDirichlet}, the vector $\va(\vr)$ defined in \eqref{eq:defatoms} can be written as
\[
\va(\vr) =  \herm{\mF} \vf(\vr)
\]
where $\herm{\mF}$ is the (inverse) 2D discrete Fourier transform matrix with the entry in the $(k,\ell)$-th row and $(r,q)$-th column given by $[\herm{\mF}]_{(k,\ell), (r,q)} \defeq \frac{1}{\L^2} e^{i2\pi \frac{qk + r\ell}{\L}}$ and the entries of the  vector $\vf$ are given by $[\vf(\vr)]_{(r,q)} \defeq e^{-i2\pi (r\tau + q \nu)}$, $k, \ell, q, r = -\N, ..., \N$, $\vr=[\tau, \nu]^T$. With these definitions, 
\begin{align}
\innerprod{\herm{\mG}\vq}{\va(\vr)} 
= \innerprod{\herm{\mG}\vq}{ \herm{\mF} \vf(\vr) } 
= \innerprod{\mF\herm{\mG}\vq}{\vf(\vr) } 
= \sum_{r,q=-\N}^\N [\mF\herm{\mG}\vq]_{(r,q)} e^{i2\pi (r\tau + q \nu)}.
\label{eq:Q2Dtripolyresp}
\end{align}
Thus, the constraint in the dual \eqref{eq:dual} says that the 2D trigonometric polynomial in \eqref{eq:Q2Dtripolyresp} is bounded in magnitude by $1$ for all $\vr \in [0,1]^2$. 
The following form of the bounded real lemma allows us to approximate this constraint by a linear matrix inequality.

\begin{proposition}[{\cite[Cor.~4.25, p.~127]{dumitrescu_positive_2007}}]
Let $P$ be a bivariate trigonometric polynomial in $\vr=[\tau, \nu]^T$
\begin{equation*}
 P(\vr) = \sum_{k, \ell=-\N}^{\N} p_{(k,\ell)} e^{i2\pi (k \tau + \ell \nu)}.
\end{equation*}
If
\[
\sup_{\vr \in [0,1]^2} |P(\vr)| < 1
\]
then there exists a matrix $\mQ\succeq 0$ such that 
\begin{equation}
\begin{bmatrix} 
\mQ & \vp \\ \herm{\vp} & 1 
\end{bmatrix} 
\succeq \vect{0} \qquad \text{and} \qquad \forall k,\ell = -\N,...,\N, \quad  \mathrm{trace}( (\boldsymbol{\Theta}_{k} \otimes \boldsymbol{\Theta}_{\ell})  \mQ) = 
\begin{cases}
1, & (k,\ell) = (0,0) \\
0, & \text{otherwise}
\end{cases}
\label{eq:adfol}
\end{equation}
where $\boldsymbol{\Theta}_{k}$ designates the elementary Toeplitz matrix with ones on the $k$-th diagonal and zeros elsewhere. 
The vector $\vp$ contains the coefficients $p_{(k,\ell)}$ of $P$ as entries, and is padded with zeros to match the dimension of $\mQ$. 
	
Reciprocally, if there exists a matrix $\mQ \succeq 0$ satisfying \eqref{eq:adfol}, then 
\[
\sup_{\vr \in [0,1]^2} |P(\vr)| \leq 1.
\]
\label{prop:BRL}
\end{proposition}

Contrary to the corresponding matrix inequality for the 1D case in \cite{bhaskar_atomic_2012,candes_towards_2014,tang_compressed_2013}, where the size of the matrix $\mQ$ is fixed to $\L \times \L$, here, the size of the matrix $\mQ$ is not determined and may, in principle, be significantly larger than the minimum size $\L^2\times \L^2$. 
This stems from the fact that the sum-of-squares representation of a positive trigonometric polynomial of degree $(\L, \L)$ possibly involves factors of degree larger than $(\L, \L)$ (see, e.g., \cite[Sec.~3.1]{dumitrescu_positive_2007}). 
%
%
%
%
Therefore, Proposition \ref{prop:BRL} only gives a sufficient condition and can not be used to \emph{exactly} characterize the constraint of the dual program \eqref{eq:dual}. 
Fixing the degree of the matrix $\mQ$ to the minimum size of $\L^2 \times \L^2$ yields a \emph{relaxation} of the constraint of the dual in \eqref{eq:dual} that leads to the following semidefinite programming \emph{relaxation} of the dual program \eqref{eq:dual}: 


\begin{align}
\underset{\vq, \mQ \in \complexset^{\L^2\times \L^2}, \mQ \succeq 0}{\text{maximize}} \, \Re \innerprod{\vq}{\vy} 
 \text{ subject to \eqref{eq:dualsemdefconstraint}} .
\label{eq:dualsemdef}
\end{align}
\begin{align}
\begin{bmatrix} 
		\mQ & \mF\herm{\mG}\vq \\ \herm{\vq} \mG \herm{\mF} & 1 
		\end{bmatrix} 
		\succeq \vect{0}, \quad  \mathrm{trace}( (\boldsymbol{\Theta}_{k} \otimes \boldsymbol{\Theta}_{\ell})  \mQ) = 
		\begin{cases}
		1, & (k,\ell) = (0,0) \\
		0, & \text{otherwise}
		\end{cases}.
\label{eq:dualsemdefconstraint}
\end{align}
Note that we could also use a higher relaxation degree than $(\L,\L)$, which would, in general, lead to a better approximation of the original problem. 
However, the relaxation of minimal degree are known to  yield optimal solution in practice, as, e.g., observed in a related problem of 2D FIR filter design~\cite{dumitrescu_positive_2007}. 
In Section \ref{sec:estfromdual}, we report an example  that shows that the relaxation of minimal degree also yields optimal solutions for our problem in practice.


\subsection{\label{sec:estfromdual}Estimation of the time-frequency shifts from the dual solution}
%

Proposition \ref{prop:dualmin} suggests that an estimate $\hat \T$ of the set of time-frequency shifts $\T$ can be obtained from a dual solution $\vq$ by identifying the $\vr_j$ for which the dual polynomial $Q(\vr) = \innerprod{\vq}{\mc F_{\nu} \mc T_{\tau} \vx } = \innerprod{\vq}{\mG \va(\vr)}$ has magnitude $1$. In general, the solution $ \vq$ to \eqref{eq:dualsemdef} 
is not unique but we can ensure that
\[
\T \subseteq \hat \T \defeq \{\vr \colon |\innerprod{ \vq}{\mG \va(\vr)}| = 1 \}.
\] 
To see this assume that $\T \setminus \hat \T \neq \emptyset$. Then, we have 
\begin{align*}
\Re \innerprod{ \vq}{\mG \vz} 
&= \Re \innerprod{ \vq}{\mG \sum_{\vr_j \in \T}  b_j \va(\vr_j)} \\
&= \sum_{\vr_j \in \T \cap \hat \T } \Re ( \conj{b}_j \innerprod{\vq}{\mG \va(\vr_j)} )
+ \sum_{\vr_j \in \T \setminus \hat \T } \Re( \conj{b}_j \innerprod{\vq}{\mG \va(\vr_j)} ) \\
&<\sum_{\vr_j \in \hat \T \cap \T } |b_n|
+ \sum_{\vr_j \in \T \setminus \hat \T } |b_n| 
= \norm[\setA]{\vz},
\end{align*}
where the strict inequality follows from $\left|\innerprod{\vq}{\mG \va(\vr)}\right| < 1$ for $\vr \in \T \setminus \hat \T$, by definition of the set $\hat \T$. This contradicts strong duality, and thus implies that $\T \setminus \hat \T = \emptyset$, i.e., we must have $\T \subseteq \hat \T$. 

In general, we might have $\T \neq \hat \T$. However, in ``most cases'', standard semidefinite programming solvers (e.g., SDPT3) will yield a solution such that $\T  = \hat \T$, see \cite[Prop.~2.5]{tang_compressed_2013} and \cite[Sec.~4]{candes_towards_2014} for formal results on related problems. 


We next provide a numerical example where the time-frequency shifts can be recovered perfectly from a solution to the semidefinite program \eqref{eq:dualsemdef}.  We choose $N=8$, consider the case of two time-frequency shifts, specifically $\T = \{ (0.2,0.5), (0.8,0.5)\}$, and let the coefficients $x_\ell, \ell = -\N,...,\N$ and the $b_j, j=1,2$, be i.i.d.~uniform on the complex unit sphere. 
In Figure \ref{fig:exdualpoly} 
we plot the dual polynomial $Q(\vr) =  \innerprod{\vq}{\mG \va(\vr)}$ with $\vq$ obtained by solving \eqref{eq:dualsemdef} via CVX \cite{cvx_2014} (CVX calls the SDPT3 solver). It can be seen that the time-frequency shifts can be recovered perfectly, i.e., $\hat \T = \T$. 


\begin{figure}
\centering
%
%
%
\includegraphics[width=\textwidth]{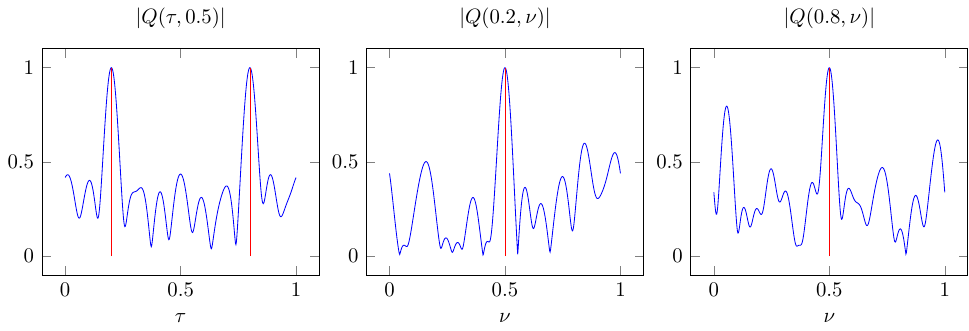}
\caption{\label{fig:exdualpoly} Localization of the time-frequency shifts via the estimated dual polynomial $Q(\tau,\nu)$ obtained by solving \eqref{eq:dualsemdef} with noiseless measurements. The red lines show the actual positions of the time-frequency shifts located at $(0.2,0.5)$ and $(0.8,0.5)$. Note that the estimated dual polynomial satisfies $|Q(\tau,\nu)| = 1$ if $(\tau,\nu) \in \{(0.2,0.5),(0.8,0.5)\}$ and $|Q(\tau,\nu)| < 1$ otherwise, thereby providing accurate identification of the time-frequency shifts.}
\end{figure}

\subsection{Recovery in the noisy case \label{sec:noisycase}}

In practice, the samples $y_p$ in \eqref{eq:periorel} are corrupted by additive noise. In that case, perfect recovery of the $(b_j,\tau_j,\nu_j)$ is in general no longer possible, and we can only hope to identify the time-frequency shifts up to an error. In the noisy case, we solve the following convex program:
\begin{align}
\underset{\minlet\vz}{\text{minimize}} \,  \norm[\setA]{\minlet{\vz}} \; \text{ subject to } \; \norm[2]{\vy - \mG \minlet{\vz}} \leq \delta.
\label{eq:lasso}
\end{align}
The semidefinite programing formulation of the dual of \eqref{eq:lasso} takes the form
\begin{align}
\underset{\vq, \mQ}{\text{maximize}} \, \Re \innerprod{\vq}{\vy} - \delta \norm[2]{\vq}
 \text{ subject to \eqref{eq:dualsemdefconstraint}} 
\label{eq:lassodualsemdef}
\end{align}
and we again estimate the time-frequency shifts $\vr_j$ as the $\vr$ for which the dual polynomial $Q(\vr) =  \innerprod{\vq}{\mG \va(\vr)}$ achieves magnitude $1$. 
We leave theoretical analysis of this approach to future work, and only provide a
 numerical example demonstrating that this approach is stable. 

We choose $N=5$, consider the case of one time-frequency shift at $(\tau_1,\nu_1) =  (0.5,0.8)$ (so that the dual polynomial can be visualized in 3D) and let the coefficients $x_\ell, \ell = -\N,...,\N$ and $b_1$ be random variables that are i.i.d.~uniform on the complex unit sphere but normalized to have variance $1/\L$. 
We add i.i.d.~complex Gaussian noise $n_p$ to the samples samples $y_p$ in \eqref{eq:periorel}, such that the signal-to-noise ratio (SNR), defined as $\text{SNR} = \norm[2]{[y_{-\N},...,y_{N}]}^2/\norm[2]{[n_{-\N},...,n_{N}]}^2$, is 10dB (we express the SNR in decibels computed as $10\log_{10}(\text{SNR})$).  
In Figure \ref{fig:noisypoly} we plot the dual polynomial $Q(\vr) =  \innerprod{\vq}{\mG \va(\vr)}$ with $\vq$ obtained by solving \eqref{eq:lassodualsemdef} (with $\delta=0.8$) using CVX. 
The time-frequency shift for which the dual polynomial achieves magnitude $1$ is $(0.4942,0.7986)$; it is very close to the original time-frequency shift $(0.5,0.8)$. 


\begin{figure}
\centering
%
%
%
%
%
%
%
%
\includegraphics[width=\textwidth]{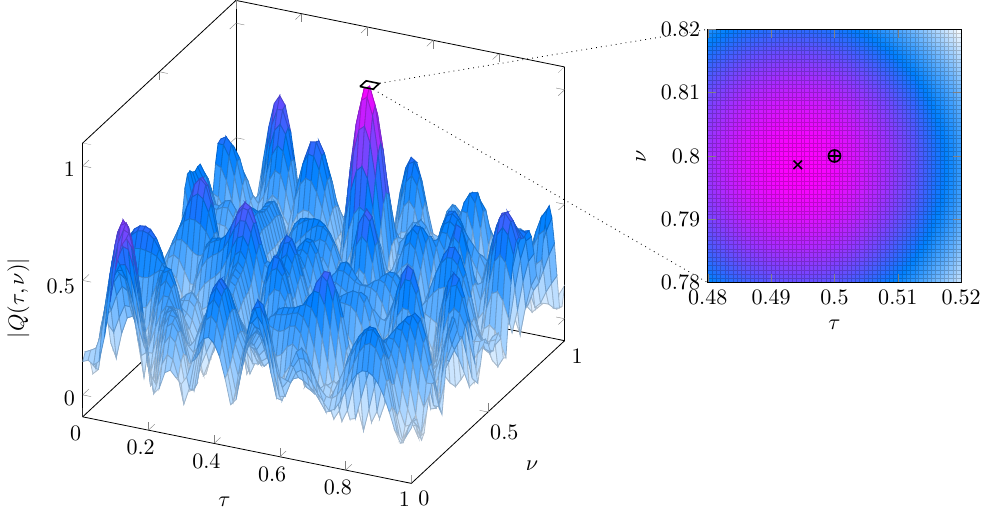}
\caption{\label{fig:noisypoly} 
Localization of the time-frequency shifts via the estimated dual polynomial $Q(\tau,\nu)$ obtained by solving \eqref{eq:lassodualsemdef} using noisy measurements (10dB noise). The estimated dual polynomial satisfies $|Q(\tau,\nu)| = 1$ for $(\tau,\nu)  = (0.4942,0.7986)$ (marked by $\times$); this is very close to the original time-frequency shift $(0.5,0.8)$ (marked by $\oplus$). }
\end{figure}



\section{Relation to previous work \label{sec:priorwork}}

%


The general problem of extracting the spreading function $\sfunc (\tau,\nu)$ of a linear time varying system of the form \eqref{eq:ltvsys} is known as system identification. It has been shown that 
 LTV systems with spreading function compactly supported on a known region of area $\Delta$ in the time-frequency plane are identifiable if and only if $\Delta \leq 1$ as shown by Kailath, Bello, Kozek, and Pfander in \cite{kailath_measurements_1962,bello_measurement_1969,kozek_identification_2005,pfander_measurement_2006}. If the spreading function's support region is unknown, a necessary and sufficient condition for identifiability is $\Delta\leq 1/2$ \cite{heckel_identification_2013}. 
In contrast to our work, the input (probing) signal in \cite{kailath_measurements_1962,bello_measurement_1969,kozek_identification_2005,pfander_measurement_2006,heckel_identification_2013} is not band-limited, and the response to the probing signal is not time-limited. In fact, the probing signal in those works is a (weighted) train of Dirac impulses, a signal that decays neither in time nor in frequency. 
If it is necessary to use a band-limited probing signal and observe the response to the probing signal for a finite amount time (as is the case in practice), it follows from the theory developed in~\cite{kailath_measurements_1962,bello_measurement_1969,kozek_identification_2005,pfander_measurement_2006,heckel_identification_2013} that $\sfunc (\tau,\nu)$ can be identified with precision of $1/B$ and $1/T$ in $\tau$ and $\nu$ directions, respectively. This is good enough if the function $\sfunc (\tau,\nu)$ is smooth in the sense that it does not vary much at the scale $1/B$ and $1/T$ in $\tau$ and $\nu$ directions, respectively, and its support contains only a few disjoint regions. The assumptions we make about  $\sfunc (\tau,\nu)$ in this paper are very far from smoothness: Our $\sfunc (\tau,\nu)$ consists of Dirac delta functions, and hence is peaky, and all the elements of its support are disjoint. In this sense, the results in this paper are complementary to those in~\cite{kailath_measurements_1962,bello_measurement_1969,kozek_identification_2005,pfander_measurement_2006,heckel_identification_2013}.

Taub\"{o}ck et al.~\cite{taubock_compressive_2010} and Bajwa et al.~\cite{bajwa_learning_2008} 
considered the identification of LTV systems with 
spreading function compactly supported  in a rectangle of area $\Delta \leq 1$. 
In \cite{taubock_compressive_2010,bajwa_learning_2008}, the time-frequency shifts lie on a \emph{coarse} grid (i.e., the grid in Section \ref{sec:ongrid}). In our setup, the time-frequency shifts need not lie on a grid and may in principle lie in a rectangle of area $\L = BT$ that is considerably larger than $1$. 
In a related direction, Herman and Strohmer \cite{herman_high-resolution_2009} (in the context of compressed sensing radar) and Pfander et al.~\cite{pfander_identification_2008} considered the case where the time-frequency shifts lie on a coarse grid (i.e., the grid in Section \ref{sec:ongrid}).  

Bajwa et al.~\cite{bajwa_identification_2011} also considered the identification of an LTV of the form \eqref{eq:iorelintro}. The approach in \cite{bajwa_identification_2011} requires 
that the time-frequency shifts $(\tauc_j,\nuc_j)$  lie in a rectangle of area much smaller than $1$ (this is required for a certain approximation in \cite[Eq.~5]{bajwa_identification_2011} to be valid) and 
in the worst case $(BT)^2\geq c \S$. Both assumption are not required here. 
  
In \cite{candes_towards_2014}, Cand\`es and Fernandez-Granda study the recovery of the frequency components of a mixture of $\S$ complex sinusoids from $\L$ equally spaced samples as in \eqref{eq:supres}. As mentioned previously, this corresponds to the case of only time or only frequency shifts. 
Recently, Fernandez-Granda \cite{fernandez-granda_super-resolution_2015} improved the main result of \cite{candes_towards_2014} by providing an tighter constant in the minimum separation condition. 
Tang et al.~\cite{tang_compressed_2013} study a related problem, namely the recovery of the frequency components from a random subset of the $\L$ equally spaced samples. Both \cite{candes_towards_2014,tang_compressed_2013} study convex algorithms analogous to the algorithm studied here, and the proof techniques developed in these papers inspired the analysis presented here. In \cite{MSthesis} the author improved the results of \cite{candes_towards_2014} with simpler proofs by building approximate dual polynomials. We believe that one can utilize this result to simplify our proofs and/or remove the random sign assumption. We leave this to future work. 
Finally, we would like to mention a few recent papers that use algorithms not based on convex optimization for the super-resolution problem \cite{demanet2014recoverability, demanet2013super, fannjiang2011music, liao2014music, moitra2014threshold}. We should note that some of these approaches can handle smaller separation compared to convex optimization based approaches but the stability of these approaches to noise is not well understood, they do not as straightforwardly generalize to higher dimensions, and often need the model order (e.g., the number of frequencies) as input parameter. 
To the best of our knowledge, the approaches in \cite{demanet2014recoverability, demanet2013super, fannjiang2011music, liao2014music, moitra2014threshold} have never been generalized to the more general radar problem as in \eqref{eq:periorel}.  



\section{Construction of the dual polynomial  \label{app:proofmainres}}

In this section we prove Theorem \ref{thm:mainres} 
by constructing a dual polynomial that satisfies the conditions of Proposition \ref{prop:dualmin}. Existence of the dual polynomial is guaranteed by the following proposition, which is the main technical result of this paper. 
\begin{proposition}
Assume that the samples of the probing signal $x_\ell, \ell =-N,...,N$, are i.i.d.~$\mathcal N(0,1/\L)$ random variables and $\L = 2\N+1 \geq 1024$. Let $\T = \{\vr_1, \vr_2,...,\vr_\S\} \subset [0,1]^2$ be an arbitrary set of points obeying the minimum separation condition 
\begin{align}
\max(|\tau_j - \tau_{j'}|, |\nu_j - \nu_{j'}| ) \geq \frac{2.38}{\N}
\text{ for all } [\tau_j, \nu_j], [\tau_{j'}, \nu_{j'}] \in \T  \text{ with } j \neq j'.
\label{eq:mindistcond}
\end{align}
Let $\vu \in \{-1,1\}^\S$ be a random vector, whose entries are i.i.d. and uniform on $\{-1,1\}$. Choose $\delta>0$ and assume
\[
S \le c\frac{L}{\log^3\left(\frac{L^6}{\delta} \right)}
\]
where $c$ is a numerical constant.
Then, with probability at least $1-\delta$, there exists a trigonometric polynomial $Q(\vr)$, $\vr = \transp{[\x,\y]}$, of the form
\begin{align}
Q(\vr) =  
 \innerprod{\vq}{\mc F_{\nu} \mc T_{\tau} \vx } 
=   \sum_{p=-\N}^{\N} \conj{[\mc F_{\nu} \mc T_{\tau} \vx]}_p q_p
 = \sum_{k,\ell = -\N}^{\N}
  \underbrace{
  \left(
  \frac{1}{\L}
  \conj{x}_\ell
  \sum_{p=-\N}^{\N}
    e^{i2\pi (p-\ell) \frac{k}{L} } 
   q_p
   \right)
   }_{q_{k,\ell} \defeq }
   e^{-i2\pi (k \x + p \y)} 
\label{eq:dualpolyinprop}
\end{align}
%
%
with complex coefficients $\vq = \transp{[ q_{-\N}, ..., q_{\N}]}$ such that 
\begin{align}
Q(\vr_j) = u_j, \text{ for all } \vr_j \in \T, \text{ and } |Q(\vr)| < 1 \text{ for all } \vr \in [0,1]^2 \setminus \T.
\label{eq:intboundcondpro}
\end{align}
\label{prop:dualpolynomial}
\end{proposition}
We provide a proof of Proposition \ref{prop:dualpolynomial} by constructing $Q(\vr)$ explicitly. 
Our construction of the polynomial $Q(\vr)$ is inspired by that in \cite{candes_towards_2014,tang_compressed_2013}. 
To built $Q(\vr)$ we need to construct a 2D-trigonometric polynomial that satisfies \eqref{eq:intboundcondpro}, and whose coefficients $q_{k,\ell}$ are constraint to be of the form  \eqref{eq:dualpolyinprop}. 
It is instructive to first consider the construction of a 2D trigonometric polynomial
\begin{align}
\bar Q(\vr) = \sum_{k,\ell = -\N}^\N \bar q_{k,\ell} e^{-i2\pi (k \x + \ell \y)}
\label{eq:Qtrigonpoly}
\end{align}
satisfying \eqref{eq:intboundcondpro} without any constraint on the coefficients $\bar q_{k,\ell}$. 
To this end, we next review the construction in \cite{candes_towards_2014} establishing that there exists a 2D trigonometric polynomial $\bar Q$ satisfying simultaneously the interpolation condition 
$\bar Q(\vr_j) = u_j, \text{ for all } \vr_j \in \T$, and the boundedness condition $\abs{\bar Q(\vr)} < 1$, for all $\vr \notin \T$, provided the minimum separation condition \eqref{eq:mindistcond} holds. 
In order to construct $\bar Q$, Cand\`es and Fernandez-Granda \cite{candes_towards_2014} interpolate the points $u_j$ with a fast-decaying kernel $\bar G$ and its partial derivatives according to  
\begin{align}
\bar Q(\vr) = \sum_{k=1}^\S \bar \alpha_k \bar G(\vr-\vr_k) + \bar \beta_{1k} \bar G^{(1,0)}(\vr - \vr_k) +  \bar\beta_{2k} \bar G^{(0,1)}(\vr - \vr_k).
\label{eq:detintpolC}
\end{align}
Here, $\bar G^{(m,n)}(\vr) \defeq \frac{\derd^m }{ \derd \x^m} \frac{\derd^n }{ \derd \y^n}  \bar G(\vr)$ and 
$  
\bar G(\vr)
\defeq \FK(\tau) \FK(\nu)
$
where $\FK(t)$ is the squared Fej\'er kernel defined as
\[
 \FK(t) \defeq \left( \frac{\sin\left( M \pi t\right)}{M \sin(\pi t)} \right)^4, \quad M \defeq \frac{\N}{2}+1.
\]
For $\N$ even, the Fej\'er kernel is a trigonometric polynomial of degree $\N/2$. It follows that $\FK(t)$ is  a trigonometric polynomial of degree $\N$ and can be written as a trigonometric polynomial with coefficients $g_k$ according to 
\begin{align}
 \FK(t) = \frac{1}{M} \sum_{k=-\N}^{\N}  g_k e^{i2\pi t k}, \quad M \defeq \frac{\N}{2}+1
\label{eq:def:FejK}
\end{align}
Since shifted versions of $F(t)$ and the derivatives of $\FK(t)$ are also 1D trigonometric polynomials of degree $\N$, it follows that the kernel $\bar G$, its partial derivatives, and shifted versions thereof, are 2D trigonometric polynomials of the form \eqref{eq:Qtrigonpoly}. 
Since $\FK(t)$ decays rapidly around the origin $t=0$ ($\FK(0)=1$), $\bar G(\vr)$ decays rapidly around the origin $\vr = \vect{0}$ as well. 
To ensure that $\bar Q(\vr)$ reaches local maxima, which is necessary for the interpolation and boundedness conditions to hold simultaneously, the coefficients $\bar \alpha_k, \bar \beta_{1k}$ and $\bar \beta_{2k}$ are chosen in a specific way guaranteeing that 
\begin{align}
\bar Q(\vr_k) = u_k,  \quad
 \bar Q^{(1,0)}(\vr_k) = 0, \text{ and }  \bar Q^{(0,1)}(\vr_k) = 0, \text{ for all } \vr_k \in \T,
\label{eq:condQinterpolv}
\end{align} 
where $\bar Q^{(m,n)}(\vr) \defeq \frac{\derd^m }{ \derd \x^m} \frac{\derd^n }{ \derd \y^n}  \bar Q(\vr)$. 
The idea of this construction is to interpolate the $u_j$ with the functions $\bar G(\vr - \vr_j)$ (the $\alpha_k$ are close to the $u_j$) and to slightly adopt this interpolation near the $\vr_j$ with the functions $\bar G^{(1,0)}(\vr - \vr_j)$ , $\bar G^{(0,1)}(\vr - \vr_j)$ to ensure that local maxima are achieved at the $\vr_j$ (the $\beta_{1k}, \beta_{2k}$ are in fact very small). 
The key properties of the interpolating functions used in this construction are that $G(\vr - \vr_j)$ decays fast around $\vr_j$, as this enables a good constant in the minimum separation condition \eqref{eq:mindistcond}, and that the ``correction'' functions $\bar G^{(1,0)}(\vr - \vr_j)$, $\bar G^{(0,1)}(\vr - \vr_j)$ are small at $\vr_j$, but sufficiently large in a small region relatively close to $\vr_j$, and decay fast far from $\vr_j$ as well. 

A first difficulty with generalizing this idea to the case where the coefficients $q_{k,\ell}$ of the 2D trigonometric polynomial have the special form  \eqref{eq:dualpolyinprop} is this: Since $\vx$ is a random vector our interpolation and correction functions are naturally non-deterministic and thus showing the equivalent to \eqref{eq:condQinterpolv}, namely \eqref{eq:interpcondQ}, 
requires a probabilistic analysis. 
Specifically, we will use concentration of measure results.  
A second difficulty is that interpolating the points $u_j$ with shifted versions of a \emph{single} function will not work, as shifted versions of a function of the special form  \eqref{eq:dualpolyinprop}, playing the role of $\bar G$, are in general not of the form \eqref{eq:dualpolyinprop} (the time and frequency shift operators $\mc T_\tau$ and $\mc F_\nu$ do not commute). 
As a result we have to work with different interpolating functions for different $\vr_j$'s. These function reach their maxima at or close to the $\vr_j$'s. 
A third difficulty is that we cannot simply use the derivatives of our interpolation functions as  ``correction'' functions, because the derivatives of a polynomial of the form \eqref{eq:dualpolyinprop} are in general not of the form \eqref{eq:dualpolyinprop}. 

We will construct the polynomial $Q(\vr)$ by interpolating the points $(\vr_k, u_k)$ with functions $G_{(m,n)}(\vr,\vr_k),\allowbreak m,n =0,1$ that have the form \eqref{eq:dualpolyinprop}:   
\begin{align}
Q(\vr) = \sum_{k=1}^\S \alpha_k G_{(0,0)}(\vr,\vr_k) + \beta_{1k} G_{(1,0)}(\vr,\vr_k) +  \beta_{2k} G_{(0,1)}(\vr,\vr_k).
\label{eq:dualpolyorig}
\end{align}  
Choosing the functions $G_{(m,n)}(\vr,\vr_k)$ in \eqref{eq:dualpolyorig} to be of the form \eqref{eq:dualpolyinprop} ensures that $Q(\vr)$ itself is of the form \eqref{eq:dualpolyinprop}. 
We will show that, with high probability, there exists a choice of coefficients $\alpha_k, \beta_{1k}, \beta_{2k}$ such that  
\begin{align}
Q(\vr_k) = u_k,  \quad
 Q^{(1,0)}(\vr_k) = 0, \text{ and }  Q^{(0,1)}(\vr_k) = 0, \text{ for all } \vr_k \in \T, \quad    
\label{eq:interpcondQ}
\end{align}
where $Q^{(m,n)}(\vr) \defeq \frac{\derd^m }{ \derd \x^m} \frac{\derd^n }{ \derd \y^n}   Q(\vr)$. 
This ensures that $Q(\vr)$ reaches local maxima at the $\vr_k$. We will then show that with this choice of coefficients, the resulting polynomial satisfies $|Q(\vr)| < 1$ for all $\vr \notin \T$. 

Now that we have stated our general strategy, we turn our attention to the construction of the interpolating and correction functions $G_{(m,n)}(\vr,\vr_k)$. We will chose $G_{(0,0)}(\vr,\vr_k)$ such that its peak is close to $\vr_k$ and it decays fast around $\vr_k$. 
For $G_{(0,0)}(\vr,\vr_k)$ to have the form \eqref{eq:dualpolyinprop} it must be random (due to $\vx$); we will show that the properties just mentioned are satisfied with high probability. 
We start by recalling (cf.~Section \ref{sec:sdprel})
\[
\mc F_\nu \mc T_\tau \vx = \mG \herm{\mF} \vf(\vr),
\]
where $\mG$ is the Gabor matrix (cf.~\ref{eq:defgabormtx}). Here,  $\herm{\mF}$ is the inverse 2D discrete Fourier transform matrix with the entry in the $(k,\ell)$-th row and $(r,q)$-th column given by $[\herm{\mF}]_{(k,\ell), (r,q)} \defeq \frac{1}{\L^2} e^{i2\pi \frac{qk + r\ell}{\L}}$ and $[\vf(\vr)]_{(r,q)} = e^{-i2\pi (r\tau + q \nu)}$ with $k, \ell, q, r = -\N, ..., \N$.  
Next, define the vector $\vg_{(m,n)}(\vr_j)   \in \complexset^{\L^2}$  as 
\[
[\vg_{(m,n)}(\vr_j)]_{(r,q)} = g_r g_q e^{-i2\pi(\x_j r + \y_j q)}  (i2\pi r)^m (i2\pi q )^n, \quad r,q=-\N,...,\N,\quad \vr_j = \transp{[\tau_j,\nu_j]}. 
\]
Here, the $g_r$ are the coefficients of the squared Fej\'er kernel in \eqref{eq:def:FejK}. 
With this notation, we define 
\begin{align}
G_{(m,n)}(\vr, \vr_j) 
&\defeq
\frac{\L^2}{M^2}
\innerprod{\mG \herm{\mF} \vg_{(m,n)}(\vr_j)
}{\mc F_\nu \mc T_\tau \vx } 
 \nonumber \\
&=
\frac{\L^2}{M^2}
\herm{\vf}(\vr) \mF \herm{\mG} \mG \herm{\mF}  \vg_{(m,n)}(\vr_j).
\label{eq:defkernelG}
\end{align}
By identifying $\vq$ in \eqref{eq:dualpolyinprop} with 
$\mG \herm{\mF} \vg_{(m,n)}(\vr_j)$, 
we immediately see that $G_{(m,n)}(\vr, \vr_j)$ and in turn $Q(\vr)$ has the form \eqref{eq:dualpolyinprop}, as desired. 
%
The particular choice of $G_{(m,n)}(\vr, \vr_j)$ is motivated by the fact---made precise later---that $G_{(m,n)}(\vr, \vr_j)$ concentrates around the deterministic function $G^{(m,n)}(\vr - \vr_j)$ of \eqref{eq:detintpolC}. 
Thus, $G_{(0,0)}(\vr,\vr_j)$ decays rapidly around $\vr_j$ with high probability. To demonstrate this empirically, in Figure \ref{fig:GbarG} we plot $\bar G(\vr)$ and $G_{(0,0)}(\vr,\vect{0})/ G_{(0,0)}(\vect{0},\vect{0})$ for $\N = 60$ and $\N=300$. Note that close to $\vr= \vect{0}$, the random function $G_{(0,0)}(\vr,\vect{0})$ and the deterministic kernel $\bar G(\vr)$ are very close. 
A simple calculation shows that the expected value of $G_{(m,n)}(\vr, \vr_j)$ with respect to  $\mathbf{x}$ is equal to $G^{(m,n)}(\vr - \vr_j)$. 
Specifically, 
as shown later on in Section \ref{sec:concstep1}, $\EX{\herm{\mG} \mG} = \mI$. This immediately implies that 
\begin{align}
\EX{G_{(m,n)}(\vr,\vr_j)} 
&= 
\frac{1}{M^2}
\herm{\vf}(\vr)  \vg_{(m,n)}(\vr_j)  \nonumber \\
&= 
\frac{1}{M^2}
\sum_{r,q=-\N}^\N
e^{i2\pi (r\x + q \y )}
g_r g_q e^{-i2\pi(\x_j r + \y_j q)}  (i2\pi r)^m (i2\pi q)^n \nonumber \\
&= 
\frac{\derd^m }{ \derd \x^m} \frac{\derd^n }{ \derd \y^n}
\frac{1}{M^2}
\sum_{r,q=-\N}^\N
e^{i2\pi( r (\x - \x_j)   + q (\y - \y_j)  )}
g_r g_q   \nonumber \\
&=
\frac{\derd^m }{ \derd \x^m} \frac{\derd^n }{ \derd \y^n}
\FK(\x-\x_j)  \FK(\y - \y_j)  \nonumber \\ 
&=  \bar G^{(m,n)}(\vr-\vr_j). 
\label{eq:expGmn}
\end{align}
%
%
%
%
%
%
%
%
%
\begin{figure}
\centering
\includegraphics{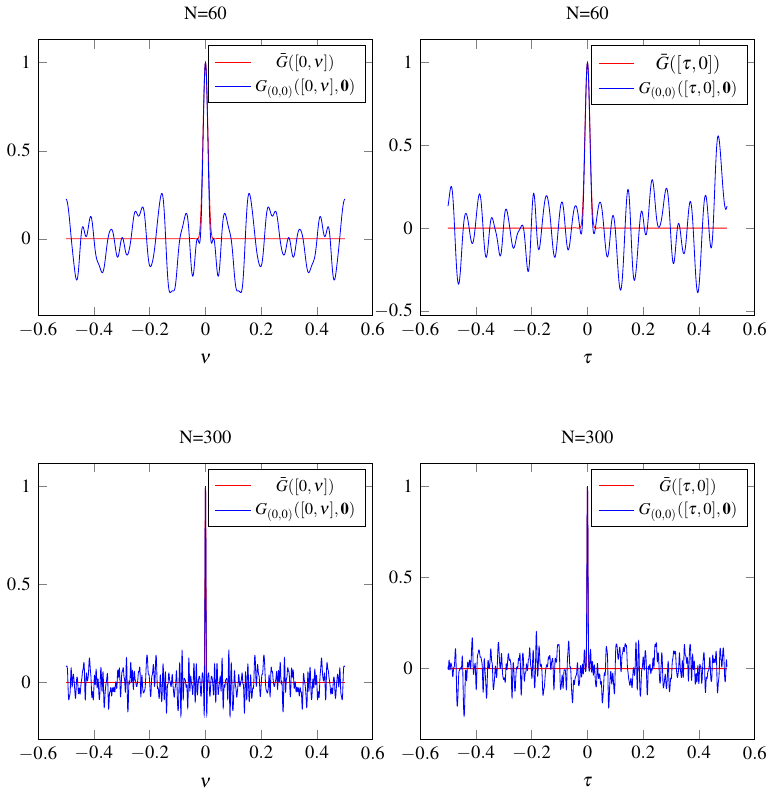}
\caption{\label{fig:GbarG} Plots of the random kernel $G_{(0,0)}(\vr,\vect{0}) / G_{(0,0)}(\vect{0},\vect{0})$ along with the deterministic kernel $\bar G(\vr)$.}
\end{figure}
%
%
%

The remainder of the proof is organized as follows. 
\begin{description}
\item[Step 1:] \label{it:step1} 
We will show that for every $\vr \in [0,1]^2$ the function $G_{(0,0)}(\vr,\vr_j)$ is close to $\bar G(\vr-\vr_j)$ with high probability, i.e., $|G_{(0,0)}(\vr,\vr_j) - \bar G(\vr-\vr_j)|$ is small. 
\item[Step 2:]\label{it:step2} 
We will then show that for a randomly chosen $\vx$,  with high probability, there exists a specific choice of coefficients $\alpha_k, \beta_{1k}, \beta_{2k}$ guaranteeing that \eqref{eq:interpcondQ} is satisfied. 
\item[Step 3:]\label{it:step3} 
We conclude the proof by showing that with the coefficients chosen as in Step 2, $\abs{Q(\vr)} < 1$ with high probability uniformly for all $\vr \notin \T$. 
This is accomplished using an $\epsilon$-net argument. 
\begin{description}
\item[Step 3a:]\label{it:step3a} Let $\Omega \subset [0,1]^2$ be a (finite) set of grid points. For every $\vr\in \Omega$, we show that $Q(\vr)$ is ``close'' to $\bar Q(\vr)$ with high probability. 
\item[Step 3b:]\label{it:step3b} We use Bernstein's polynomial inequality to conclude that this result holds with high probability uniformly for all $\vr \in [0,1]^2$.
\item[Step 3c:]\label{it:step3c} 
Finally, we combine this result with a result in  \cite{candes_towards_2014} that shows that $\abs{\bar Q(\vr)} < 1$ for all $\vr \notin \T$ to conclude that $\abs{Q(\vr)} < 1$ holds with high probability uniformly for all  $\vr \notin \T$. 
\end{description}
\end{description}


\subsection{Step 1: Concentration of $G_{(0,0)}(\vr,\vr_j)$ around $\bar G(\vr-\vr_j)$}
\label{sec:concstep1}

In this subsection we establish the following result. 

\begin{lemma}
Let 
$G_{(m',n')}^{(m,n)}(\vr, \vr_j)
=
\frac{\derd^m }{ \derd \x^m} \frac{\derd^n }{ \derd \y^n}
G_{(m',n')}(\vr, \vr_j)
$ and 
fix $\vr, \vr_j \in [0,1]^2$. For all $\alpha \geq 0$, and for all nonnegative integers $m, m',n,n'$ with $m + m'+n + n' \leq 4$, 
\begin{align}
&\PR{ \frac{1}{\kappa^{m+m'+n+n'}}  | G^{(m,n)}_{(m',n')}(\vr, \vr_j) - \bar G^{(m + m',n + n')}(\vr - \vr_j)  |  > c_1  12^{\frac{m + m'+n+ n'}{2}}   \frac{\alpha}{\sqrt{\L}}  } \nonumber \\
&\hspace{8cm}\leq 2 \exp\left( - c \min\left( \frac{ \alpha^2}{c_2^4   },  \frac{ \alpha}{ c_2^2   }\right) \right),  \label{eq:polydev}
\end{align}
where $\kappa \defeq \sqrt{|\FK''(0)|}$ and $c,c_1,c_2$ are numerical constants. 
\label{lem:polybo}
\end{lemma}

To this aim, first note that, by the definition of $G_{(m',n')}(\vr, \vr_j)$ in \eqref{eq:defkernelG} we have 
\begin{align}
\label{eq:defGmnrr}
G_{(m',n')}^{(m,n)}(\vr, \vr_j)
=
\frac{\derd^m }{ \derd \x^m} \frac{\derd^n }{ \derd \y^n}
G_{(m',n')}(\vr, \vr_j)
=
\frac{\L^2}{M^2}
\herm{(\vf^{(m,n)}(\vr))} \mF \herm{\mG} \mG \herm{\mF}  \vg_{(m',n')}(\vr_j),
\end{align}
where for $r, q= -\N, ..., \N$ we define $[\vf^{(m,n)}(\vr)]_{(r,q)} \defeq  (-i2\pi r)^m (-i2\pi r)^n   e^{-i2\pi (r\tau + q \nu)}$.  
Lemma~\ref{lem:polybo} is now proven in two steps. 
First, we show that $\EX{\herm{\mG} \mG} = \mI$. 
From this, following calculations similar to \eqref{eq:expGmn}, we obtain
\begin{align}
\EX{
G^{(m,n)}_{(m',n')}(\vr, \vr_j)} 
=   \bar G^{(m + m',n + n')}(\vr - \vr_j). 
\label{eq:expGmnGen}
\end{align}
Second, we express $G^{(m,n)}_{(m',n')}(\vr,\vr_j)$ as a quadratic form in $\mathbf{x} \defeq \transp{[x_{-\N},...,x_{\N}]}$, and utilize the Hanson-Wright inequality stated in the lemma below to show that $G^{(m,n)}_{(m',n')}(\vr,\vr_j)$ does not deviate too much from its expected value $\bar G^{(m + m',n + n')}(\vr - \vr_j)$. 

\begin{lemma}[Hanson-Wright inequality {\cite[Thm.~1.1]{rudelson_hanson-wright_2013}}]
Let $\mathbf{x} \in \mathbb R^L$ be a random vector with independent zero-mean $K$-sub-Gaussian entries (i.e., the entries obey $\sup_{p\geq 1} p^{-1} (\EX{|x_\ell|^p})^{1/p} \leq K$), and let $\mV$ be an $L\times L$ matrix. Then, for all $t\geq 0$, 
\[
\PR{ | \transp{\mathbf{x}} \mV \mathbf{x}  -  \EX{\transp{\mathbf{x}} \mV \mathbf{x}} |  > t } \leq 2 \exp\left( - c \min\left( \frac{t^2}{K^4 \norm[F]{\mV}^2  },  \frac{t}{K^2 \norm[\opnormss]{\mV}  }\right) \right)
\]
where $c$ is a numerical constant. 
\label{thm:hanswright}
\end{lemma}
We first establish $\EX{\herm{\mG} \mG} = \mI$. 
By definition of the Gabor matrix in \eqref{eq:defgabormtx}, the entry in the $(k,\ell)$-th row and $(k',\ell')$-th column of $\herm{\mG} \mG$ is given by 
\[
[\herm{\mG} \mG]_{(k,\ell), (k',\ell')}  =  \sum_{p=-\N}^\N \conj{x}_{p-\ell} x_{p-\ell'} e^{-i2\pi \frac{kp}{\L}}  e^{i2\pi \frac{k'p}{\L}}.
\]
Noting that $\EX{x_\ell} = 0$, we conclude that $\EX{[\herm{\mG} \mG]_{(k,\ell), (k',\ell')}} = 0$ for $\ell \neq \ell'$. For $\ell = \ell'$, using the fact that $\EX{\conj{x}_{p-\ell} x_{p-\ell}} = 1/\L$, we arrive at 
\[
\EX{ [\herm{\mG} \mG]_{(k,\ell), (k',\ell')} } = \frac{1}{\L} \sum_{p=-\N}^\N    e^{i2\pi \frac{(k' - k )p}{\L}}.
\]
The latter is equal to $1$ for $k = k'$ and $0$ otherwise. This concludes the proof of $\EX{\herm{\mG} \mG} = \mI$.

We now turn our attention to the concentration part of the argument, where we express $G^{(m,n)}_{(m',n')}(\vr,\vr_j)$ as a quadratic form in $\mathbf{x}$ and apply Lemma~\ref{thm:hanswright}. To this end, first note that 
\begin{align}
[L \mG \herm{\mF} \vg_{(m',n')}(\vr_j) ]_p 
&=
\frac{1}{\L}
\sum_{k,\ell=-\N}^N 
\left(
\sum_{r,q=-\N}^\N  
g_r g_q e^{-i2\pi(\x_j r + \y_j q)}  (i2\pi r)^{m'} (i2\pi q)^{n'}
e^{i2\pi \frac{q k + r\ell }{\L}}
\right)
x_{p-\ell} e^{i2\pi \frac{kp}{\L}} 
%
\nonumber \\
&=
 \sum_{\ell=-\N}^\N  x_\ell   \sum_{r=-\N}^\N e^{i2\pi \frac{r(p-\ell)}{\L}} 
 g_r g_p e^{-i2\pi(\x_j r - \y_j p)}  (i2\pi r)^{m'} (-i2\pi p)^{n'}
\label{eq:GFHgs},
\end{align}
where we used that $\frac{1}{\L} \sum_{k=-\N}^\N e^{i2\pi \frac{k (p+q)}{\L}}$ is equal to $1$ if $p= -q$ and equal to $0$ otherwise, together with the fact that $x_\ell$ is $\L$-periodic.   
We next write \eqref{eq:GFHgs} in matrix-vector form. For the ease of presentation, we define the matrix $\mA( \vg_{(m',n')}(\vr_j)) \in \complexset^{\L\times \L}$ (note that $\mA$ is a function of $\vg_{(m',n')}(\vr_j)$) by
\[
[\mA(  \vg_{(m',n')}(\vr_j) )]_{p,\ell} \defeq  \sum_{k=-\N}^{\N}   e^{i2\pi\frac{k(p-\ell)}{L}} 
g_k g_p e^{-i2\pi(\x_j k - \y_j p)}  (i2\pi k)^{m'} (-i2\pi p)^{n'}.
\]
Utilizing this shorthand, writing \eqref{eq:GFHgs} in matrix-vector form yields 
\[
L \mG \herm{\mF} \vg_{(m',n')}(\vr_j)  = 
\mA(  \vg_{(m',n')}(\vr_j) ) \vx, 
\]
where $\vx = \transp{[x_{-N}, ..., x_{\N}]}$.

Analogously as in \eqref{eq:GFHgs}, we have 
\begin{align}
[L \mG \herm{\mF} \vf^{(m,n)}(\vr)]_{p} = 
 \sum_{\ell=-\N}^{\N}x_{\ell}
\sum_{k=-\N}^{\N}    
e^{i2\pi\frac{k(p-\ell)}{L}}  e^{-i2\pi(k \x - p \y)}
(-i2\pi k)^{m}
(i2\pi p)^{n}
. 
\label{eq:GFhfrm}
\end{align}
Defining the matrix $\herm{\mA}(\vf^{(m,n)}(\vr))\in \complexset^{L\times L}$ by 
\[
[\herm{\mA}(\vf^{(m,n)}(\vr))  ]_{\tilde \ell, p} = \sum_{\tilde k=-\N}^{\N}    e^{-i2\pi\frac{\tilde k ( p - \tilde \ell)}{L}} e^{i2\pi(\tilde k \x - p \y)}
(i2\pi \tilde k)^{m}
(-i2\pi p)^{n}
\]
allows us to express \eqref{eq:GFhfrm} in matrix-vector form according to 
\[
\L\,\herm{(\vf^{(m,n)}(\vr))} \mF \herm{\mG} = \herm{\mathbf{x}} \herm{\mA}(\vf^{(m,n)}(\vr)).
\]
This allows us to represent $G^{(m,n)}_{(m',n')}(\vr, \vr_j)$ in the desired quadratic form
\begin{align}
G^{(m,n)}_{(m',n')}(\vr, \vr_j) 
= \frac{\L^2}{M^2} \herm{(\vf^{(m,n)}(\vr))} \mF \herm{\mG} \mG \herm{\mF} \vg_{(m',n')}(\vr_j)
=    \herm{\mathbf{x}}  
\underbrace{\frac{1}{M^2} \herm{\mA}(\vf^{(m,n)}(\vr))  \mA(\vg_{(m',n')}(\vr_j)) }_{\mV^{(m,n)}_{(m',n')}(\vr,\vr_j)  \defeq} \mathbf{x}, 
\label{eq:gmnqform}
\end{align}
where
\begin{align}
&[\mV^{(m,n)}_{(m',n')}(\vr,\vr_j)]_{\tilde \ell, \ell}  \label{eq:defV}\\
&\hspace{0.7cm} =\frac{1}{M^2}  \sum_{p,k,\tilde k=-\N}^{\N}    e^{i2\pi\frac{(p - \ell )k}{L}}  e^{-i2\pi\frac{(p - \tilde \ell)\tilde k}{L}} e^{i2\pi(\tilde k \x - p (\y - \y_j) - k \x_j  )}  g_p g_k (i2\pi \tilde k)^m (i2\pi k)^{m'} (-i2\pi p)^{n+n'} .
\nonumber 
\end{align}
%
%
In order to evaluate the RHS of the Hanson-Wright inequality, we will need the following upper bound on $\norm[F]{\mV^{(m,n)}_{(m',n')}(\vr,\vr_j)}$. We defer the proof to Section \ref{seclem:boundonVFnorm}.
\begin{lemma}
For all $\vr$ and $\vr_j$, and for all non-negative $m, m',n,n'$ with $m + m'+n+n' \leq 4$, 
\begin{equation}
\norm[F]{\mV^{(m,n)}_{(m',n')}(\vr,\vr_j)} \leq c_1 (2\pi \N)^{m + m' + n +n'} \sqrt{L}.
\label{eq:boundonVFnorm}
\end{equation}
\label{lem:boundonVFnorm}
\end{lemma}

We are now ready to establish Lemma \ref{lem:polybo} by applying the Hanson-Wright inequality. 
To this end note that using $\kappa=\sqrt{\FK''(0)} = \sqrt{\frac{\pi^2}{3}(\N^2+4\N)}$ and utilizing \cite[Eq.~2.23]{candes_towards_2014} we have
\[
\frac{(2\pi \N)^{m}}{\kappa^{m}} =  \frac{(2\pi \N)^{m}}{( \frac{\pi^2}{3} (\N^2 + 4 \N))^{(m)/2}}  \leq 12^{\frac{m}{2}} . 
\]
Setting $\mV \defeq \mV^{(m,n)}_{(m',n')}(\vr, \vr_j)$ for ease of presentation, we have 
\begin{align}
&\hspace{-1cm}
\PR{ \frac{1}{\kappa^{m+m'+n+n'}}  | G^{(m,n)}_{(m',n')}(\vr, \vr_j) - \bar G^{(m + m',n + n')}(\vr - \vr_j)  |  > c_1  12^{\frac{m + m'+n+ n'}{2}}   \frac{\alpha}{\sqrt{\L}}  } \nonumber \\
&\leq
\PR{  | G^{(m,n)}_{(m',n')}(\vr, \vr_j) - \bar G^{(m + m',n + n')}(\vr - \vr_j)  |  > c_1  (2\pi \N)^{m + m'+n +n'}   \frac{\alpha}{\sqrt{\L}}  } \nonumber \\
&\leq \PR{ | \transp{\mathbf{x}} \mV  \mathbf{x}  -  \EX{\transp{\mathbf{x}} \mV  \mathbf{x}} |  >  \norm[F]{ \mV} \frac{\alpha}{\L}  } \label{eq:uselemboundvfnor} \\
&\leq 2 \exp\left( - c \min\left( \frac{\norm[F]{\mV }^2 \alpha^2}{\L^2 K^4 \norm[F]{\mV  }^2  },  \frac{\norm[F]{\mV} \alpha}{ \L K^2 \norm[\opnormss]{\mV}  }\right) \right) \label{eq:usehansonwr} \\
&\leq 2 \exp\left( - c \min\left( \frac{ \alpha^2}{c_2^4   },  \frac{ \alpha}{ c_2^2   }\right) \right).  \label{eq:simphanswrer}
\end{align}
Here, \eqref{eq:uselemboundvfnor} follows from~\eqref{eq:gmnqform} and~\eqref{eq:boundonVFnorm}, together with the fact that  
$\EX{\herm{\mathbf{x}} \mV \mathbf{x}} =\EX{G^{(m,n)}_{(m',n')}(\vr, \vr_j)} = \bar G^{(m+m',n+n')}(\vr - \vr_j)$ (cf.~\eqref{eq:expGmnGen}). To obtain \eqref{eq:usehansonwr}, we used Lemma \ref{thm:hanswright} with $t=\norm[F]{\mV} \frac{\alpha}{\L}$. 
Finally, \eqref{eq:simphanswrer} holds because the sub-Gaussian parameter $K$ of the random variable $[\mathbf{x}]_\ell \sim \mathcal N(0,1/\L)$ is given by $K = c_2/\sqrt{\L}$ (e.g., \cite[Ex.~5.8]{vershynin_introduction_2012}) and $\norm[F]{\mV }/\norm[\opnormss]{\mV }  \geq 1$.


\subsubsection{Proof of Lemma \ref{lem:boundonVFnorm}:}
\label{seclem:boundonVFnorm}

We start by upper-bounding $|[\mV^{(m,n)}_{(m',n')}(\vr,\vr_j)]_{\tilde \ell, \ell}|$. 
By definition of $\FK(t)$ (cf.~\eqref{eq:def:FejK})
\begin{align*}
&[\mV^{(m,n)}_{(m',n')}(\vr,\vr_j)]_{\tilde \ell, \ell}   \nonumber \\
&= 
  \sum_{p=-\N}^{\N}
   \left(\frac{1}{M} \sum_{k=-\N}^{\N} g_k (i2\pi k)^{m'} e^{i2\pi\left( \frac{p - \ell }{L}  - \tau_j \right)k} \right)
  \left(
\frac{1}{M}  
  \sum_{\tilde k=-\N}^{\N} (i2\pi \tilde k)^{m}
     e^{-i2\pi \left( \frac{p - \tilde \ell}{L}  - \x \right)\tilde k }  \right) \cdot \nonumber \\
   &\hspace{9.5cm} \cdot g_p  (-i2\pi p)^{n+n'}   e^{-i2\pi  (\y - \y_j)p} \nonumber \\
&= 
\sum_{p=-\N}^{\N}
F^{(m')}
\left( 
\frac{p-\ell}{L} - \tau_j
\right)
  \left(
\frac{1}{M}  
  \sum_{\tilde k=-\N}^{\N} (i2\pi \tilde k)^{m}
     e^{-i2\pi \left( \frac{p - \tilde \ell}{L}  - \x \right)\tilde k }  \right) 
     g_p  (-i2\pi p)^{n+n'}   e^{-i2\pi  (\y - \y_j)p},\nonumber \\
\end{align*}
where $F^{(m)}(t) \defeq \frac{\derd^m }{ \derd \x^m}  F(t)$. 
Since $|g_p| \leq 1$ holds for all $p$, we obtain 
\begin{align}
|[\mV^{(m,n)}_{(m',n')}(\vr,\vr_j)]_{\tilde \ell, \ell}|
&\leq (2\pi \N)^{n+n'} \sum_{p=-\N}^{\N}    \left| \FK^{(m')} \left( \frac{p-\ell}{L}   - \tau_j   \right) \right|   
\left| \frac{1}{M} \sum_{\tilde k=-\N}^{\N}  (-i2\pi \tilde k)^m    e^{i2\pi\left(\frac{p - \tilde \ell}{L} - \x \right) \tilde k  }  \right| \nonumber \\
&= (2\pi \N)^{n+n'} \sum_{p=-\N}^{\N}    \left| \FK^{(m')} \left( \frac{p}{L}    +s/\L - \tau_j \right) \right|   
\left| \frac{1}{M} \sum_{\tilde k=-\N}^{\N}  (-i2\pi \tilde k)^m    e^{i2\pi\left(\frac{p + \ell + s - \tilde \ell}{L} - \x \right) \tilde k  }  \right| ,  \label{mahlabel}
\end{align}
where we choose $s$ as the integer minimizing $|s/\L - \tau_j|$ and used the fact that the absolute values in the sum above are $L$-periodic in $p$ (recall that $\FK(t)$ is $1$-periodic).  

We proceed by upper-bounding $|F^{(m)}(t)|$. 
To this aim, we use Bernstein's polynomial inequality, (cf.~Proposition \ref{prop:bernstein} below), to conclude that 
\begin{align}
\sup_{t} \left| F^{(m)}(t)  \right|  \leq (2 \pi \N)^{m} \sup_{t}  |F(t)| = (2 \pi \N)^{m}. 
\label{eq:Fmtb1}
\end{align}
Also note that, from \cite[Lem.~2.6]{candes_towards_2014} we know that for $|t| \in [1/(2N), 1/2]$, there exists a numerical constant $\tilde c$ such that 
\begin{align}
|F^{(m)}(t)| 
\leq \tilde c (2\pi \N)^m \frac{1}{(2Mt)^4}.
\label{eq:Fmtb2}
\end{align}
Combining \eqref{eq:Fmtb1} and \eqref{eq:Fmtb2} we arrive at 
\begin{align}
|F^{(m)}(t)| 
\leq 
H^{(m)}(t)
\defeq
\bar c (2\pi \N)^{m} 
 \min \left( 1, \frac{1}{(2M t)^4 } \right).
\label{eq:boundderrFej}
\end{align}
Utilizing the latter inequality we have 
\begin{align}
\left| \FK^{(m)} \left( \frac{p}{L} + s/\L - \tau_j \right) \right|
&\leq
H^{(m)}\left( \frac{p}{L} + s/\L - \tau_j \right)  \nonumber \\
&\leq 
c' H^{(m)}\left( \frac{p}{L}\right) \label{eq:usesdef} \\
&\leq c' \bar c (2\pi \N)^{m}  \min \left( 1, \frac{1}{(2Mp /L  )^4 } \right) \label{eq:uaseeq:boundderrFej}\\
&\leq  c' \bar c (2\pi \N)^{m}   \min \left( 1, \frac{16}{p^4 } \right) \leq 16 c' \bar c (2\pi \N)^{m} \min \left( 1, \frac{1}{p^4 } \right).
\label{eq:stdnqwml}
\end{align}
Here, \eqref{eq:usesdef} holds because when $s$ is the integer minimizing $|s/\L - \tau_j|$, we have that $|s/\L - \tau_j|\leq 1/(2L)$. Therefore, for all $p$ with $|p|>0$, $2M ( p/\L    - s/L -  \tau_j )$ is within a constant factor of $2M p/L$, which proves that \eqref{eq:usesdef} holds for a numerical constant $c'$. 
To obtain \eqref{eq:uaseeq:boundderrFej} we used \eqref{eq:boundderrFej}. Finally, \eqref{eq:stdnqwml} follows from $\frac{\L}{2M} = \frac{2\N+1}{\N+2} < 2$. 

Plugging \eqref{eq:stdnqwml} into \eqref{mahlabel} we obtain 
\begin{align}
&|[\mV^{(m,n)}_{(m',n')}(\vr,\vr_j)]_{\tilde \ell, \ell}|  \nonumber \\
&\hspace{0.7cm}\leq (2\pi \N)^{m+m'+n+n'}    
\underbrace{ \hat c (2\pi \N)^{-m}   \!\! \sum_{p=-\N}^{\N} \min\left(1, \frac{1}{p^4} \right)    
\left| \frac{1}{M} \sum_{\tilde k=-\N}^{\N}     (-i2\pi \tilde k)^m e^{i2\pi\left(\frac{ p + s+ \ell - \tilde \ell}{L} - \x \right) \tilde k  } \right|}_{ U\left( \x - \frac{s + \ell - \tilde \ell}{L} \right)  \defeq}   \label{mah2},
%
\end{align}
where $\hat c = 16 c' \bar c$. 
We show in Appendix \ref{sec:boundU} that $U(t)$ is $1$-periodic and satisfies $U(t) \leq c \min(1, \frac{1}{L |t|})$ for $|t|\leq 1/2$. Using this bound together with \eqref{mah2} we conclude that 
\begin{align}
\norm[F]{\mV^{(m,n)}_{(m',n')}(\vr,\vr_j)}^2 
=& \sum_{\ell, \tilde \ell=-\N}^{\N} \left|[\mV^{(m,n)}_{(m',n')}(\vr,\vr_j)]_{\tilde \ell, \ell} \right|^2\nonumber\\ 
\le&  (2\pi \N)^{2(m +m'+ n + n')}  \sum_{\ell, \tilde \ell = -\N}^\N   U^2\!\left( \x  - \frac{s + \ell - \tilde \ell }{L} \right)\nonumber\\
\le& (2\pi \N)^{2(m +m'+ n + n')} \sum_{\tilde{\ell}=-N}^N\sum_{\ell=-\N}^\N \left(  c \min\left(1, \frac{1}{L |\ell/L|}\right)  \right)^2\nonumber\\
\le&c^2(2\pi \N)^{2(m +m'+ n + n')}\sum_{\tilde{\ell}=-N}^N \left( 1 + 2 \sum_{\ell\geq 1}   \frac{1}{\ell^2}  \right) \nonumber\\
=&c^2L(2\pi \N)^{2(m +m'+ n + n')}\left( 1 + \frac{\pi^2}{3} \right) \nonumber.
\end{align}
The proof is now complete by setting $c_1=c\sqrt{1+\pi^2/3}\sqrt{L}$.

%

\subsection{Step 2: Choice of the coefficients $\alpha_k, \beta_{1k}, \beta_{2k}$}

We next show that, with high probability, it is possible to select the coefficients $\alpha_k, \beta_{1k}, \beta_{2k}$ such that $Q(\vr)$ satisfies \eqref{eq:interpcondQ}. 
To this end, we first review the result in \cite{candes_towards_2014} that ensures that there exists a set of coefficients $\bar \alpha_k, \bar \beta_{1k}, \bar \beta_{2k}$ such that  \eqref{eq:condQinterpolv} is satisfied. Specifically, writing \eqref{eq:condQinterpolv} in matrix form yields 
\begin{align}
\underbrace{
\begin{bmatrix}
\bar \mD^{(0,0)} & \frac{1}{\kappa}  \bar \mD^{(1,0)} & \frac{1}{\kappa} \bar \mD^{(0,1)} \\
-\frac{1}{\kappa} \bar \mD^{(1,0)} & -\frac{1}{\kappa^2} \bar \mD^{(2,0)} & -\frac{1}{\kappa^2} \bar \mD^{(1,1)} \\
-\frac{1}{\kappa} \bar \mD^{(0,1)} & -\frac{1}{\kappa^2} \bar \mD^{(1,1)} & -\frac{1}{\kappa^2}\bar \mD^{(0,2)} 
\end{bmatrix}
}_{\bar \mD }
\begin{bmatrix}
\bar \val \\
\kappa \bar \vbe_1 \\
\kappa \bar \vbe_2
\end{bmatrix}
&=
\begin{bmatrix}
\vu \\
\vect{0} \\
\vect{0}
\end{bmatrix}
%
%
\end{align}
where $[\bar \mD^{(m,n)}]_{j,k} \defeq \bar G^{(m,n)}(\vr_j - \vr_k)$, $[\bar \val]_k \defeq \bar \alpha_k$, $[\bar \vbe_1]_k \defeq \bar \beta_{1k}$ and $[\bar \vbe_2]_k \defeq \bar \beta_{2k}$. 
Here, we have scaled the entries of $\bar \mD$ such that its diagonal entries are $1$ ($\FK(0)=1$, $\kappa^2 = |\FK''(0)|$, and $\FK''(0)$ is negative). 
Since $\bar \mD^{(0,0)},\bar \mD^{(1,1)},\bar \mD^{(2,0)},\bar \mD^{(0,2)}$ are symmetric and $\bar \mD^{(1,0)},\bar \mD^{(0,1)}$ are antisymmetric, $\bar \mD$ is symmetric. 

The following result, which directly follows from  \cite[Eq.~C6, C7, C8, C9]{candes_towards_2014}, ensures that $\bar \mD$ is invertible and thus the coefficients $\bar \alpha_k, \bar \beta_{1k}, \bar \beta_{2k}$ can be obtained according to 
\begin{align}
\begin{bmatrix}
\bar \val \\
\kappa \bar \vbe_1 \\
\kappa \bar \vbe_2
\end{bmatrix}
=
\inv{\bar \mD}
\begin{bmatrix}
\vu \\
\vect{0} 
\end{bmatrix}
=
\bar \mL \vu,
\label{eq:barLu}
\end{align}
where $\bar \mL$ is the $3\S \times \S$ submatrix of $\inv{\bar \mD}$ corresponding to the first $\S$ columns of $\inv{\bar \mD}$.

\begin{proposition}
$\bar \mD$ is invertible and 
\begin{align}
\norm[\opnormss]{\mI  - \bar \mD}  &\leq 0.19808 \\
\norm[\opnormss]{\bar \mD} &\leq 1.19808 \\
\norm[\opnormss]{ \inv{\bar \mD}} &\leq 1.24700.  \label{eq:boundinvbard}
\end{align}
\end{proposition}

\begin{proof}
The proof of this proposition is an immediate consequence of \cite[Eq.~C6, C7, C8, C9]{candes_towards_2014}. Since $\bar \mD$ is real and symmetric, it is normal, and thus its singular values are equal to the absolute values of its eigenvalues. Using that the diagonal entries of $\bar \mD$ are $1$, by Gershgorin's circle theorem \cite[Thm.~6.1.1]{horn_matrix_2012}, the eigenvalues of $\bar \mD$ are in the interval $[1-\norm[\infty]{\mI  - \bar \mD}, 1+ \norm[\infty]{\mI  - \bar \mD}]$, where $\norm[\infty]{\mA} \defeq \max_i \sum_j |[\mA]_{i,j}|$. Using that $\norm[\infty]{\mI  - \bar \mD} \leq 0.19808$ (shown below), it follows that $\bar \mD$ is invertible and 
\begin{align*}
\norm[\opnormss]{\bar \mD} &\leq 1+ \norm[\infty]{\mI  - \bar \mD}  \leq 1.19808  
\nonumber \\
\norm[\opnormss]{ \inv{\bar \mD}} &\leq \frac{1}{1- \norm[\infty]{\mI  - \bar \mD}} \leq 1.2470.
\end{align*}
The proof is concluded by noting that
\begin{align}
\norm[\infty]{\mI - \bar \mD} 
&= 
\max \left\{ \norm[\infty]{\mI - \bar \mD^{(0,0)}} \!+\! 2\norm[\infty]{ \frac{1}{\kappa}  \bar \mD^{(1,0)}},  \norm[\infty]{ \frac{1}{\kappa}  \bar \mD^{(1,0)}} \!+\! \norm[\infty]{\mI - \frac{1}{\kappa^2} \bar \mD^{(2,0)}} \!+\! \norm[\infty]{ \frac{1}{\kappa^2}  \bar \mD^{(1,1)}}  
  \right\} \nonumber \\
&\leq 0.19808, \nonumber
\end{align} 
where we used \cite[Eq.~C6, C7, C8, C9]{candes_towards_2014}: 
\begin{align*}
\norm[\infty]{\mI - \bar \mD^{(0,0)}} &\leq 0.04854 \\
\norm[\infty]{ \frac{1}{\kappa}  \bar \mD^{(1,0)}} = \norm[\infty]{ \frac{1}{\kappa} \bar \mD^{(0,1)}} &\leq 0.04258 \\
\norm[\infty]{ \frac{1}{\kappa^2}  \bar \mD^{(1,1)}}  &\leq 0.04791 \\
\norm[\infty]{\mI - \frac{1}{\kappa^2}\bar \mD^{(0,2)}} = \norm[\infty]{\mI - \frac{1}{\kappa^2}\bar \mD^{(2,0)}}  &\leq 0.1076.
\end{align*}

\end{proof}

We next select the coefficients $\alpha_k, \beta_{1k}, \beta_{2k}$ such that  
$Q(\vr)$ satisfies the interpolation conditions \eqref{eq:interpcondQ}. To this end, we write \eqref{eq:interpcondQ} in matrix form: 
\begin{align}
\underbrace{
\begin{bmatrix}
\mD_{(0,0)}^{(0,0)} & \frac{1}{\kappa}  \mD_{(1,0)}^{(0,0)} & \frac{1}{\kappa} \mD_{(0,1)}^{(0,0)} \\
-\frac{1}{\kappa} \mD^{(1,0)}_{(0,0)} & -\frac{1}{\kappa^2} \mD^{(1,0)}_{(1,0)} & -\frac{1}{\kappa^2} \mD^{(1,0)}_{(0,1)} \\
-\frac{1}{\kappa} \mD^{(0,1)}_{(0,0)} & -\frac{1}{\kappa^2} \mD^{(0,1)}_{(1,0)} & -\frac{1}{\kappa^2} \mD^{(0,1)}_{(0,1)} 
\end{bmatrix}
}_{\mD}
\begin{bmatrix}
\val \\
\kappa \vbe_1 \\
\kappa \vbe_2
\end{bmatrix}
=
\begin{bmatrix}
\vu \\
\vect{0} \\
\vect{0}
\end{bmatrix}, 
%
\label{eq:syseqorig}
\end{align}
where $[\mD^{(m,n)}_{(m',n')}]_{j,k} \defeq G^{(m,n)}_{(m',n')}(\vr_j, \vr_k)$ (cf.~\eqref{eq:defGmnrr}), 
$[\val]_k \defeq \alpha_k$, $[\vbe_1]_k \defeq \beta_{1k}$, and $[\vbe_2]_k \defeq \beta_{2k}$. To show that the system of equations \eqref{eq:syseqorig} has a solution, and in turn \eqref{eq:interpcondQ} can be satisfied, we will show that, with high probability, $\mD$ is invertible. To this end, we show that the probability of the event 
\[
\mc E_\xi = \{ \norm[\opnormss]{ \mD - \bar \mD }  \leq \xi\}
\]
is high, and $\mD$ is invertible on $\mathcal E_\xi$ for all $\xi \in (0,1/4]$. The fact that $\mD$ is invertible on $\mathcal E_\xi$ for all $\xi \in (0,1/4]$ follows from the following set of inequalities:
\[
\norm[\opnormss]{\mI - \mD} \leq \norm[\opnormss]{\mD - \bar \mD} + \norm[\opnormss]{\bar \mD - \mI} \leq \xi + 0.1908 \leq 0.4408. 
\]
Since $\mD$ is invertible, the coefficients $\alpha_k, \beta_{1k}, \beta_{2k}$ can be selected as 
\begin{align}
\begin{bmatrix}
\val \\
\kappa \vbe_1 \\
\kappa \vbe_2
\end{bmatrix}
= \inv{\mD} \begin{bmatrix}
\vu \\
\vect{0} \\
\vect{0}
\end{bmatrix}
= \mL \vu ,
\label{eq:alphabeta}
\end{align}
where $\mL$ is the $3\S \times \S$ submatrix of $\inv{\mD}$ corresponding to the first $\S$ columns of $\inv{\mD}$. 
We record two useful inequalities about $\mathbf{L}$ and its deviation from $\bar{\mathbf{L}}$ on the event $\mc E_\xi$ in the lemma below. 
\begin{lemma}
On the event $\mc E_\xi$ with $\xi \in (0,1/4]$ the following identities hold
\begin{align}
\norm[\opnormss]{\mL} \leq&  2.5 \label{eq:normmLb},\\
\norm[\opnormss]{\mL-\bar{\mL}}\leq& 2.5 \xi. \label{eq:normLmbL} 
\end{align}
\end{lemma}
\begin{proof}
We will make use of the following lemma. 
\begin{lemma}[{\cite[Proof of Cor.~4.5]{tang_compressed_2013}}]
Suppose that $\mC$ is invertible and 
$
\norm{\mB -\mC} \norm{\inv{\mC}} \leq 1/2
$. Then i) $\norm{\inv{\mB}} \leq 2 \norm{\inv{\mC}}$ and ii) $\norm{\inv{\mB} - \inv{\mC}} \leq 2 \norm{\inv{\mC}}^2 \norm{\mB - \mC}$. 
\label{lem:breadqy}
\end{lemma}

First note that since $\norm[]{\mD - \bar \mD} \leq 1/4$ and $\norm[]{\inv{\bar \mD}} \leq 1.247$ (cf.~\eqref{eq:boundinvbard}) the conditions of Lemma \ref{lem:breadqy} with $\mB = \mD$ and $\mC = \bar \mD$ are satisfied. 

Equation \eqref{eq:normmLb} now readily follows:
\begin{align}
\norm[\opnormss]{\mL} \leq \norm[\opnormss]{\inv{\mD}} \leq  2 \norm[\opnormss]{\inv{\bar \mD}} \leq 2.5 \label{eq:normmLb2}.
\end{align}
Specifically, for the first inequality, we used that $\mL$ is a submatrix of $\inv{\mD}$, the second inequality follows by application of part $i)$ of Lemma \ref{lem:breadqy}, and the last inequality follows from \eqref{eq:boundinvbard}. 

To prove \eqref{eq:normLmbL}, we use part $ii)$ of the lemma above together with \eqref{eq:boundinvbard}, to conclude that
\begin{align}
\norm[\opnormss]{\mL - \bar \mL}
\leq \norm[\opnormss]{\inv{\mD} - \inv{\bar \mD}}  
\leq  2 \norm[\opnormss]{\inv{\bar \mD}}^2 \norm[\opnormss]{\mD - \bar \mD}  \leq 2.5 \xi. \label{eq:normLmbL2} 
\end{align}
holds on $\mc E_\xi$  with $\xi \in (0,1/4]$.
\end{proof}

It only remains to prove that the event $\mc E_\xi$ does indeed have high probability.


\begin{lemma} For all $\xi>0$ 
\[
\PR{  \mc E_\xi } 
=
\PR{  \norm[\opnormss]{ \mD - \bar \mD }  \leq \xi }
\geq 1 - \delta
\]
provided that 
\begin{align}
\L \geq \S \frac{c_4}{\xi^2} \log^2(\L\S/\delta),
\label{eq:condlemaaf}
\end{align}
where $c_4$ is a numerical constant. 
\label{lem:preptaubound}
\end{lemma}

%
%


\begin{proof}
We first note that we can write
\begin{align}
\mD 
=
\sum_{\ell, \ell' = -\N}^\N
\mB_{\ell, \ell'} 
\prsig_\ell \prsig_{\ell'},
\end{align}
with
\begin{align}
\mB_{\ell,\ell'}
=
\begin{bmatrix}
\mB_{(0,0)}^{(0,0)} & \frac{1}{\kappa}  \mB_{(1,0)}^{(0,0)} & \frac{1}{\kappa} \mB_{(0,1)}^{(0,0)} \\
-\frac{1}{\kappa} \mB^{(1,0)}_{(0,0)} & -\frac{1}{\kappa^2} \mB^{(1,0)}_{(1,0)} & -\frac{1}{\kappa^2} \mB^{(1,0)}_{(0,1)} \\
-\frac{1}{\kappa} \mB^{(0,1)}_{(0,0)} & -\frac{1}{\kappa^2} \mB^{(0,1)}_{(1,0)} & -\frac{1}{\kappa^2} \mB^{(0,1)}_{(0,1)} 
\end{bmatrix}
\in \reals^{3\S \times 3 \S}
\end{align}
where we defined
\begin{align}
[ \mB_{(m,n)}^{(m',n')} ]_{j,k}
=
[\mV^{(m,n)}_{(m',n')}(\vr_j,\vr_k)]_{\ell, \ell'}
\end{align}
with $\mV$ defined in~\eqref{eq:defV}. 
From $\bar \mD = \EX{\mD}$, it follows that
\begin{align}
\mD - \bar \mD
&=
\sum_{\ell, \ell' = -\N}^\N
\mB_{\ell, \ell'} 
\left(
\prsig_\ell \prsig_{\ell'} 
-
\EX{ \prsig_\ell \prsig_{\ell'} }
\right).
\end{align}
We now treat each of the $3S \times 3S$-sub-matrices of this difference separately for simplicity. The maximal singular value of the matrix $\mD - \bar \mD$ is at most $3$-times the maximal singular value of any of the submatrices; therefore controlling the maximal singular values of the submatrices is sufficient. 

We start with $\mB_{(0,0), \ell, \ell'}^{(0,0)}$, and drop the indices by writing $\mB_{\ell, \ell'} = \mB_{(0,0), \ell, \ell'}^{(0,0)}$ with a slight abuse of notation. 
We apply the following matrix Hanson-Wright inequality, proven in the appendix:

\begin{theorem}
\label{thm:mtxHW}
Let $g_1,\ldots,g_{L}$ be iid zero-mean and unit variance Gaussian random variables, and let $\mB_{\ell,\ell'} \in \reals^{d_1\times d_2}$ be fixed matrices. 
Then, for $t \geq 0$,
\begin{align*}
\PR{
\norm{\sum_{\ell, \ell'=1}^L
\mtx{B}_{\ell,\ell'} (g_\ell g_{\ell'} - \EX{g_\ell g_{\ell'} } ) }\ge t
}
\leq
2(d_1+d_2)e^{-c\min\left(\frac{t}{\max_\ell\norm{\mtx{B}_\ell}},\frac{t^2}{\max\left(\norm{\sum_{\ell=1}^L \mtx{B}_\ell^T\mtx{B}_\ell},\norm{\sum_{\ell=1}^L \mtx{B}_\ell\mtx{B}_\ell^T}\right)}\right)}\\
+
4(d_1+d_2)(L+1)e^{-\frac{t}{8\sqrt{e L \max_{\ell'}\max\left(\norm{\sum_{\ell\neq\ell'}\mtx{B}_{\ell,\ell'}^T\mtx{B}_{\ell,\ell'}},\norm{\sum_{\ell\neq\ell'}\mtx{B}_{\ell,\ell'}\mtx{B}_{\ell,\ell'}^T}\right)}}}
\end{align*}
\end{theorem}


To apply the matrix Hanson-Wright inequality, we first note that as shown below,
\begin{align}
\label{eq:boundonsumBellell'}
\norm{\sum_{\ell}\mtx{B}_{\ell,\ell}\mtx{B}_{\ell,\ell}^T},
\norm{\sum_{\ell\neq\ell'}\mtx{B}_{\ell,\ell'}\mtx{B}_{\ell,\ell'}^T}
\leq c S.
\end{align}
Moreover, we have the corse bound 
$\norm{\mB_{\ell,\ell'}} \leq cS$, which follows by Geshgorin's disk theorem and the fact that the entries of $\mB_{\ell,\ell'}$ are bounded by a constant.
Using those two inequalities along with $x_\ell \sim \mc N(0,1/L)$, it follows from the matrix Hanson-Wright inequality that
\begin{align*}
\PR{
\norm{\sum_{\ell, \ell'=1}^L
\mtx{B}_{\ell,\ell'} (x_\ell x_{\ell'} - \EX{x_\ell x_{\ell'} } ) }\ge t / L
}
\leq
4S
e^{-c\min\left( t/S , t^2/S \right)}
+
8S(L+1)
e^{-\frac{c t}{\sqrt{ LS }} }. 
\end{align*}
With $t = \xi \L/3$, this gives
\begin{align*}
\PR{
\norm{\sum_{\ell, \ell'=1}^L
\mtx{B}_{\ell,\ell'}
(x_\ell x_{\ell'} - \EX{x_\ell x_{\ell'} } ) }\ge \xi/3
}
\leq
4S
e^{-c\min\left( \xi L/S , \xi^2 L^2/S \right)}
+
8S(L+1)
e^{-c \xi \sqrt{L/S} }  \leq \delta/9,
\end{align*}
where the last inequality follows from the assumption~\eqref{eq:condlemaaf} stating that $\L \geq \S \frac{c_4}{\xi^2} \log^2(\L\S/\delta)$. 

The proof for all other $S\times S$ sub matrices of the $3S \times 3S$ matrix $\mD - \bar \mD$ is analogous, and yields the same bound. Combining them via the union bound gives the desired bound
\begin{align*}
\PR{ \norm{\mD - \bar \mD} \geq \xi } \leq \delta,
\end{align*}
which concludes the proof.

\paragraph{Proof of inequality~\ref{eq:boundonsumBellell'}:}
In order to prove this inequality, first note that the $i,j$-th entry of the symmetric matrix $\mB_{(0,0),\ell,\ell'}^{(0,0)}$ is given by:
\begin{align}
[\mV^{(0,0)}_{(0,0)}(\vr_i,\vr_j)]_{\tilde \ell, \ell}
=
\frac{1}{M^2}  \sum_{p,k,\tilde k=-\N}^{\N}    e^{i2\pi\frac{(p - \ell )k}{L}}  e^{-i2\pi\frac{(p - \tilde \ell)\tilde k}{L}} e^{i2\pi(\tilde k \x_i - p (\y_i - \y_j) - k \x_j  )}  g_p g_k .
\nonumber 
\end{align}
For simplicity, first consider the special case of $\nu_i = \tilde \nu_i = 0$, and assume that the $\tau_i,\tau_j$ lie on a $1/L$ grid and are $\tau_i = t_i/L$, where the $t_i$'s are integers. Note that for this setup, the $t_i$'s are separated by multiples of at least $2$ by the minimum separation condition. 
For this special case,
\begin{align}
[\mV^{(0,0)}_{(0,0)}(\vr_i,\vr_j)]_{\tilde \ell, \ell}
&=
\frac{1}{M^2} \sum_{p,k,\tilde k=-\N}^{\N}    e^{i2\pi\frac{(p - \ell - t_j )k}{L}}  e^{-i2\pi\frac{(p - \tilde \ell - t_i)\tilde k}{L}}   g_p g_k. \\
\nonumber
&=
\frac{1}{M} \sum_{p,k=-\N}^{\N}    e^{i2\pi\frac{(p - \ell - t_j )k}{L}}  \ind{ 0 = p -  \tilde \ell - t_i}  g_p g_k.
\nonumber  \\ 
&=
\frac{L}{M}
\frac{1}{M} \sum_{k=-\N}^{\N} e^{i2\pi\frac{(\tilde \ell - \ell + t_i - t_j )k}{L}}   g_{\tilde \ell + t_i} g_k.
\nonumber  \\ 
&=
\frac{L}{M}
\FK \left( 
\frac{\tilde \ell - \ell + t_i - t_j }{L} 
\right)
g_{\tilde \ell + t_i} 
\nonumber.
\end{align}
As a consequence,
\begin{align}
\label{eq:ineqmBcoeff}
[ \mB_{\ell,\ell'} \herm{\mB}_{\ell,\ell'} ]_{i,j}
&= 
\sum_{k=1}^\S
\frac{L}{M}
\FK \left(
\frac{\tilde \ell - \ell + t_i - t_k }{L} 
\right)
\FK \left(
\frac{\tilde \ell - \ell + t_j - t_k }{L} 
\right)
g_{\tilde \ell + t_j}
g_{\tilde \ell + t_j}.
\end{align}
By Gershgorin's circle theorem, we have that
\begin{align*}
\norm{ \sum_{\ell \neq \ell'} \mB_{\ell,\ell'} \herm{\mB}_{\ell,\ell'} }
&\leq
\max_i 
\sum_{j=1}^S
\left| \sum_{\ell \neq \ell'} [ \mB_{\ell,\ell'} \herm{\mB}_{\ell,\ell'} ]_{i,j} \right| \\
&\stackrel{(i)}{\leq}
\max_i
c
\sum_{j=1}^S
\sum_{\ell \neq \ell'}
\sum_{k=1}^\S
\left| 
\FK \left(
\frac{\tilde \ell - \ell + t_i - t_k }{L} 
\right)
\FK \left(
\frac{\tilde \ell - \ell + t_j - t_k }{L} 
\right)
\right| \\
&\stackrel{(ii)}{\leq}
\max_i
c S 
\sum_{j=1}^S
\sum_{\ell \neq \ell'}
\left|
\FK \left(
\frac{\tilde \ell - \ell + t_i }{L} 
\right)
\FK \left(
\frac{\tilde \ell - \ell + t_j }{L} 
\right)
\right| \\
&\stackrel{(iii)}{\leq}
\max_i
c S 
\sum_{j=1}^S
\sum_{\ell \neq \ell'}
\min\left(1, \frac{1}{(\tilde \ell - \ell + t_i)^4 } \right)
\min\left(1, \frac{1}{(\tilde \ell - \ell + t_j)^4 } \right) \\
&\stackrel{(iv)}{\leq} 
\tilde c S,
\end{align*}
where inequality (i) holds by inequality~\eqref{eq:ineqmBcoeff} and by using that the coefficients $g_\ell$ are bounded by one and that $L/M$ is bounded by a constant;  
inequality (ii) holds because the the kernel $F$ is 1-periodic and the sum over $\ell$ is over $\{-\N,\ldots,\N\} \setminus \{\ell'\}$, $L = 2\N +1$, therefore regardless of $t_k$, the sum is the same;
inequality (iii) holds by the bound~\eqref{eq:boundderrFej}, and inequality 
(iv) holds by noting that the double sum evaluates to a constant using that the $t_i,t_j$'s are separated by at least a multiple of $2$ ($1$ is sufficient as well) by the minimum separation condition. 

The proof for the general case (i.e., where the $\nu_i$'s are non-zero) is more technical but gives the same result using techniques as in Lemma \ref{lem:boundonVFnorm}.

\end{proof}

\subsection{
Step 3a: $Q(\vr)$ and $\bar Q(\vr)$ are close on a grid
}

The goal of this section is to prove Lemma \ref{lem:diffqmbarqmongrid} below that shows that $Q(\vr)$ and $\bar Q(\vr)$ and their partial derivatives are close on a set of (grid) points. 

\begin{lemma}
\label{lem6}
Let $\Omega \subset [0,1]^2$ be a finite set of points. Fix $0<\epsilon \leq 1$ and $\delta > 0$. Suppose that
\[
L \geq \frac{\S}{\epsilon^2}  \max \left( 
c_5 \log^2\left(\frac{12 \S |\Omega|}{\delta}\right) \log\left(\frac{8 |\Omega|}{\delta}\right), 
c \log \left( \frac{4 |\Omega|}{\delta} \right) \log\left( \frac{18\S^2}{\delta}\right) \right).
\]
Then, 
\[
\PR{
\max_{\vr \in \Omega} \frac{1}{\kappa^{n+m}} \left| Q^{(m,n)}(\vr) - \bar Q^{(m,n)}(\vr)  \right| \leq \epsilon 
} \geq 1 - 4 \delta.
\]
\label{lem:diffqmbarqmongrid}
\end{lemma}

In order to prove Lemma \ref{lem:diffqmbarqmongrid}, first note that 
the $(m,n)$-th partial derivative of $Q(\vr)$ (defined by \eqref{eq:dualpolyorig}) after normalization with $1/\kappa^{m+n}$ is  given by
\begin{align}
\frac{1}{\kappa^{m+n}} Q^{(m,n)}(\vr) 
&= \frac{1}{\kappa^{m+n}} \sum_{k=1}^\S  \Big( \alpha_k  G^{(m,n)}_{(0,0)}(\vr, \vr_k) 
+ \kappa \beta_{1k}  \frac{1}{\kappa} G^{(m,n)}_{(1,0)}(\vr, \vr_k) +  \kappa \beta_{2k} \frac{1}{\kappa} G^{(m,n)}_{(0,1)}(\vr, \vr_k)  \Big)
\nonumber \\
&= \herm{(\vv^{(m,n)}(\vr) )} \mL \vu.  \label{eq:Qmninnprodform}
\end{align}
Here, we used \eqref{eq:alphabeta} and the shorthand $\vv^{(m,n)} (\vr)$ defined by 
\begin{align*}
\herm{(\vv^{(m,n)}) } (\vr)
\defeq \frac{1}{\kappa^{m+n}} 
\bigg[
&G^{(m,n)}_{(0,0)}(\vr,\vr_1), ..., G^{(m,n)}_{(0,0)}(\vr,\vr_\S),
\; \frac{1}{\kappa} G^{(m,n)}_{(1,0)}(\vr,\vr_1), ..., \frac{1}{\kappa} G^{(m,n)}_{(1,0)}(\vr, \vr_\S), \\
& \frac{1}{\kappa} G^{(m,n)}_{(0,1)}(\vr, \vr_1),..., \frac{1}{\kappa} G^{(m,n)}_{(0,1)}(\vr, \vr_\S)
\bigg].
\end{align*}
Since 
$
\EX{
G^{(m,n)}_{(m',n')}(\vr, \vr_j)} 
=   \bar G^{(m + m',n + n')}(\vr - \vr_j)
$ (cf.~\eqref{eq:expGmnGen}), we have 
\[
\EX{\vv^{(m,n)} (\vr) } = \bar \vv^{(m,n)}(\vr),
\]
where 
\begin{align*}
\herm{(\bar\vv^{(m,n)})}(\vr)
\defeq
 \frac{1}{\kappa^{m+n}} 
\bigg[
&\bar G^{(m,n)}(\vr-\vr_1), ..., \bar G^{(m,n)}(\vr-\vr_\S),
\; \frac{1}{\kappa} \bar G^{(m+1,n)}(\vr - \vr_1) ,..., \frac{1}{\kappa} \bar G^{(m+1,n)}(\vr - \vr_\S), \\
& \frac{1}{\kappa} \bar G^{(m,n+1)}(\vr - \vr_1) ,..., \frac{1}{\kappa} \bar G^{(m,n+1)}(\vr - \vr_\S)
\bigg].
\end{align*}
Next, we decompose the derivatives of $Q(\vr)$ according to 
\begin{align}
\frac{1}{\kappa^{m+n}} &Q^{(m,n)}(\vr)
=
\innerprod{\vu}{\herm{\mL}  \vv^{(m,n)}(\vr) } \nonumber \\
&= \innerprod{\vu}{  \herm{\bar \mL}  \bar \vv^{(m,n)}(\vr) }
+ \underbrace{\innerprod{\vu}{\herm{\mL} (\vv^{(m,n)}(\vr)  - \bar \vv^{(m,n)}(\vr)) } }_{I^{(m,n)}_1(\vr)}
+ \underbrace{\innerprod{\vu}{\herm{(\mL - \bar \mL)} \bar \vv^{(m,n)}(\vr) } }_{I^{(m,n)}_2(\vr)} \nonumber \\
&= \frac{1}{\kappa^{m+n}} \bar Q^{(m,n)}(\vr) + I_1^{(m,n)}(\vr) + I_2^{(m,n)}(\vr),
\label{eq:pertI1I2}
\end{align}
where $\bar \mL$ was defined below \eqref{eq:barLu}. 
The following two results establish that the perturbations $I_1^{(m,n)}(\vr)$ and $I_2^{(m,n)}(\vr)$ are small on a set of (grid) points $\Omega$ with high probability.

\begin{lemma} 
Let $\Omega \subset [0,1]^2$ be a finite set of points and suppose that $m+n\leq 2$. Then, for all $0<\epsilon\le 1$ and for all $\delta > 0$, 
\[
\PR{\max_{\vr \in \Omega}  |I^{(m,n)}_1(\vr)| \geq \epsilon} \leq \delta +\PR{\comp{ \mc E}_{1/4}}
\]
provided that 
\[
\L \geq \frac{c_5}{\epsilon^2} \S \log^2\left(\frac{12 \S |\Omega|}{\delta}\right) \log\left(\frac{8 |\Omega|}{\delta}\right) . 
\]
\label{lem:uboundI1}
\end{lemma}
\begin{lemma}
Let $\Omega \subset [0,1]^2$ be a finite set of points. 
Suppose that $m+n\leq 2$. Then, for all $\epsilon,\delta > 0$, and for all $\xi>0$ with  
\begin{align}
\xi \leq \frac{\epsilon c_6}{\sqrt{\log\left( \frac{4|\Omega|}{\delta} \right)}},
\label{eq:ubontaune}
\end{align}
where $c_6\leq 1/4$ is a numerical constant, it follows
\[
\PR{\max_{\vr \in \Omega}  |I^{(m,n)}_2(\vr)| \geq \epsilon \Big| \mc E_\xi } \leq \delta.
\]
\label{lem:uboundI2}
\end{lemma}

Lemma \ref{lem:uboundI1} and \ref{lem:uboundI2} are proven in Section \ref{sec:proflemuboundl1} and \ref{sec:prooflemuboundl2}, respectively.

We are now ready to complete the proof of Lemma \ref{lem6}. From \eqref{eq:pertI1I2}, we obtain for all $\xi>0$ satisfying~\eqref{eq:ubontaune} 
\begin{align}
\PR{\max_{\vr \in \Omega} \frac{1}{\kappa^{n+m}} \left| Q^{(m,n)}(\vr) - \bar Q^{(m,n)}(\vr)  \right| \geq 2 \epsilon } 
&=
\PR{ \max_{\vr \in \Omega} \left|  I_1^{(m,n)}(\vr) + I_2^{(m,n)}(\vr) \right| \geq 2 \epsilon } \nonumber \\
&\hspace{-4cm}\leq
\PR{ \max_{\vr \in \Omega} \left|  I_1^{(m,n)}(\vr) \right| \geq \epsilon }
+
\PR{  \comp{\mc E}_\xi }
+
\PR{ \max_{\vr \in \Omega} \left| I_2^{(m,n)}(\vr) \right| \geq \epsilon | \mc E_\xi } \label{eq:uppbpri1I2}\\
&\hspace{-4cm}\leq 4 \delta,  \nonumber 
\end{align}
where \eqref{eq:uppbpri1I2} follows from the union bound and the fact that $\PR{A} =  \PR{A \cap \comp{B}} + \PR{A \cap B} \leq \PR{\comp{B}} + \PR{A | B}$ by setting  
$B = \mc E_\xi$ and $A = \left\{\max_{\vr \in \Omega} \left| I_2^{(m,n)}(\vr) \right| \geq \epsilon \right\}$, and the last inequality follows from Lemmas \ref{lem:preptaubound}, \ref{lem:uboundI1} and \ref{lem:uboundI2}, respectively.  In more details, we choose $\xi = \epsilon c_6 \log^{-1/2}\left( \frac{4|\Omega|}{\delta} \right)$. It then follows from Lemma \ref{lem:uboundI2} that the third probability in \eqref{eq:uppbpri1I2} is smaller than $\delta$. With this choice of $\xi$, the condition in Lemma \ref{lem:preptaubound} becomes $L \geq \S \frac{c_4}{\epsilon^2 c_6^2} \log\left(\frac{4 |\Omega|}{\delta} \right) \log \left( \frac{18\S^2}{\delta} \right)$, which is satisfied by choosing $c = \frac{c_4}{c_6^2}$. Moreover, $\xi \leq 1/4$ since $\epsilon\leq 1$ and $c_6\leq 1/4$. 
Thus, Lemma \ref{lem:preptaubound} yields $\PR{  \comp{\mc E}_\xi } \leq \delta$ and $\PR{  \comp{\mc E}_{1/4} } \leq \delta$. Finally, observe that the conditions of Lemma \ref{lem:uboundI1} are satisfied, thus the first probability in \eqref{eq:uppbpri1I2} can be upper-bounded by
\[
\PR{ \max_{\vr \in \Omega} \left|  I_1^{(m,n)}(\vr) \right| \geq \epsilon } \leq \delta +\PR{\comp{ \mc E}_{1/4}} \leq 2 \delta.
\] 
This concludes the proof.

\subsubsection{Proof of Lemma \ref{lem:uboundI1} \label{sec:proflemuboundl1}}
Set $\Delta \vv^{(m,n)} \defeq \vv^{(m,n)} (\vr) - \bar \vv^{(m,n)} (\vr)$  for notational convenience. By the union bound, we have for all $a,b\geq 0$,
\begin{align}
\PR{\max_{\vr \in \Omega}  |I^{(m,n)}_1(\vr)| \geq 2.5 a b} 
\hspace{-3cm}&\hspace{3cm}= \PR{\max_{\vr \in \Omega} \left|\innerprod{\vu}{\herm{\mL} \Delta \vv^{(m,n)} }\right| \geq 2.5 a b} \nonumber \\
&\leq \PR{\bigcup_{\vr \in \Omega } 
\left\{ \left|\innerprod{\vu}{\herm{\mL} \Delta \vv^{(m,n)} }\right| \geq \norm[2]{\herm{\mL} \Delta \vv^{(m,n)}} b \right\} 
\cup \left\{  \norm[2]{\herm{\mL} \Delta \vv^{(m,n)}}\geq  2.5 a \right\}  } \nonumber \\
&\leq \PR{\bigcup_{\vr \in \Omega } 
\left\{ \left|\innerprod{\vu}{\herm{\mL} \Delta \vv^{(m,n)} }\right| \geq \norm[2]{\herm{\mL} \Delta \vv^{(m,n)}} b \right\} 
\cup \left\{  \norm[2]{\Delta \vv^{(m,n)}}\geq a \right\} \cup \left\{ \norm[\opnormss]{\mL} \geq 2.5 \right\}  }  \nonumber  \\
&\leq \PR{\norm[\opnormss]{\mL} \geq 2.5 } 
+\sum_{\vr \in \Omega} \left( \PR{\left|\innerprod{\vu}{\herm{\mL} \Delta \vv^{(m,n)} }\right| \geq \norm[2]{\herm{\mL} \Delta \vv^{(m,n)}} b}
 +\PR{\norm[2]{\Delta \vv^{(m,n)}}\geq  a} \right)  \nonumber \\
&\leq \PR{\comp{ \mc E}_{1/4}} +  |\Omega|  4 e^{- \frac{b^2}{4}}
 + \sum_{\vr \in \Omega}  \PR{\norm[2]{\Delta \vv^{(m,n)}}\geq  a} \label{eq:useHoeffanddf} \\
 &\leq \PR{\comp{ \mc E}_{1/4}} + \frac{\delta}{2} + \sum_{\vr \in \Omega} \PR{\norm[2]{\Delta \vv^{(m,n)}}\geq  a},  \label{eq:inposltdel}
\end{align}
where \eqref{eq:useHoeffanddf} follows from application of Hoeffding's inequality (stated below) and from $\{\norm[\opnormss]{\mL} \geq 2.5\} \subseteq \comp{ \mc E}_{1/4}$ according to \eqref{eq:normmLb}.  For \eqref{eq:inposltdel}, we used $|\Omega|  4 e^{- \frac{b^2}{4}} \leq \frac{\delta}{2}$ ensured by choosing $b = 2 \sqrt{\log(8|\Omega|/\delta)}$. 
\begin{lemma}[Hoeffding's inequality]
Suppose the entries of $\vu \in \mathbb R^{\S}$ are i.i.d.~with $\PR{u_i = -1} = \PR{u_i =1} = 1/2$. Then, for all $t\geq 0$, and for all $\vv \in \complexset^\S$
\[
\PR{ \left|\innerprod{\vu}{\vv}\right|  \geq \norm[2]{\vv}  t } \leq 4 e^{- \frac{t^2}{4}}. 
\]
\label{thm:hoeff}
\end{lemma}
We next upper-bound $\PR{\norm[2]{\Delta \vv^{(m,n)}}\geq  a}$ in \eqref{eq:inposltdel}. For all $\alpha \geq 0$, using that $12^{\frac{n+m+1}{2}} \leq 12^{\frac{3}{2}}$, we have
\begin{align}
\PR{\norm[2]{\Delta \vv^{(m,n)} } \geq \frac{\sqrt{3\S}}{\sqrt{L}} 12^{\frac{3}{2}} c_1 \alpha  }
&\leq 
\PR{\norm[2]{\Delta \vv^{(m,n)} } \geq \frac{\sqrt{3\S}}{\sqrt{L}} 12^{\frac{n+m+1}{2}} c_1 \alpha  } \nonumber \\
&=
\PR{\norm[2]{\Delta \vv^{(m,n)}}^2 \geq \frac{3\S}{L} 12^{n+m+1} c_1^2 \alpha^2  } \nonumber \\
&\leq \sum_{k=1}^{3\S} \PR{|[ \Delta \vv^{(m,n)} ]_k|^2\geq  \frac{1}{L} 12^{n+m+1} c_1^2 \alpha^2 } \label{eq:ubagao} \\
&= \sum_{k=1}^{3\S} \PR{|[\Delta \vv^{(m,n)} ]_k| \geq  \frac{1}{\sqrt{L}} 12^{\frac{n+m+1}{2}} c_1 \alpha}  \nonumber \\
&\leq 3\S \cdot 2 \exp\left( - c \min\left( \frac{ \alpha^2}{c_2^4   },  \frac{ \alpha}{ c_2^2   }\right) \right) \label{eq:lobDvgead} \\
&\leq \frac{\delta}{2 |\Omega|} ,\label{eq:choosealphac2sq}
\end{align}
where \eqref{eq:ubagao} follows from the union bound, \eqref{eq:lobDvgead} follows from Lemma \ref{lem:polybo}. Finally, \eqref{eq:choosealphac2sq} follows by choosing $\alpha = \frac{c_2^2}{ c} \log\left( \frac{12 \S |\Omega| }{\delta} \right)$ and using the fact that for $\alpha \geq c_2^2$ (since $c\ge 1$)
we have $\min\left( \frac{ \alpha^2}{c_2^4   },  \frac{ \alpha}{ c_2^2   }\right) = \frac{ \alpha}{ c_2^2   }$. We have established that $\PR{\norm[2]{\Delta \vv^{(m,n)} } \geq a  } \leq \frac{\delta}{2|\Omega|}$ with $a = \frac{\sqrt{3\S}}{\sqrt{L}} 12^{\frac{3}{2}} c_1 \frac{c_2^2}{ c} \log\left( \frac{12\S |\Omega| }{\delta} \right)$. 

Substituting \eqref{eq:choosealphac2sq} into \eqref{eq:inposltdel} we get
\[
\PR{\max_{\vr \in \Omega}  |I^{(m,n)}_1(\vr)| \geq \sqrt{c_5} \frac{\sqrt{\S}}{\sqrt{L}}  \log\left( \frac{12\S |\Omega| }{\delta} \right) \sqrt{\log\left( \frac{8|\Omega|}{\delta} \right)} 
} \leq \delta + \PR{\comp{ \mc E}_{1/4}},
\]
where $c_5 = (5 \sqrt{3} \,12^{\frac{3}{2}} c_1 \frac{c_2^2}{\sqrt{c}})^2$ is a numerical constant. This concludes the proof.
\subsubsection{\label{sec:prooflemuboundl2}Proof of Lemma \ref{lem:uboundI2}}
%
By the union bound 
\begin{align}
\PR{\max_{\vr \in \Omega}  |I^{(m,n)}_2(\vr)| \geq \epsilon \Big| \mc E_\xi } 
&\leq \sum_{\vr \in \Omega}
\PR{ \left|\innerprod{ \vu }{\herm{(\mL - \bar \mL)} \bar \vv^{(m,n)}(\vr) }\right| \geq \epsilon  \Big| \mc E_\xi } \nonumber \\
&\leq \sum_{\vr \in \Omega}  \PR{\left|\innerprod{\vu}{\herm{(\mL - \bar \mL)} \bar \vv^{(m,n)}(\vr) }\right| \geq  \norm[2]{\herm{(\mL - \bar \mL)} \bar \vv^{(m,n)}(\vr) }  \frac{\epsilon}{c_5\xi}  }
\label{eq:useeq:ubltlvr}\\
&\leq |\Omega|  4 e^{- \frac{(\epsilon/(c_5\xi))^2}{4}}
\label{eq:useHoeffanddf2} \\
&\leq \delta \label{eq:iqledeluc}
\end{align}
where \eqref{eq:useeq:ubltlvr} follows from \eqref{eq:ubltlvr} below,  \eqref{eq:useHoeffanddf2} follows by Hoeffding's inequality (cf.~Lemma \ref{thm:hoeff}), and to obtain \eqref{eq:iqledeluc} we used the assumption  \eqref{eq:ubontaune} with $c_6=1/(2c_5)$. 

To complete the proof, note that by \eqref{eq:normLmbL} we have $\norm[\opnormss]{\mL - \bar \mL} \leq 2.5 \xi$ on $\mc E_\xi$. Thus, conditioned on $\mc E_\xi$,
\begin{align}
\norm[2]{\herm{(\mL - \bar \mL)} \bar \vv^{(m,n)}(\vr) } 
\leq \norm[\opnormss]{\mL - \bar \mL} \norm[2]{ \bar \vv^{(m,n)}(\vr) } \leq 2.5 \xi \norm[1]{ \bar \vv^{(m,n)}(\vr) }
\leq c_5 \xi,
\label{eq:ubltlvr}
\end{align}
where we used $\norm[2]{\cdot} \leq \norm[1]{\cdot}$, and the last inequality follows from the fact that, for all $\vr$, 
\begin{align*}
\norm[1]{ \bar \vv^{(m,n)}(\vr) }
&= \frac{1}{\kappa^{m+n}} \sum_{k=1}^\S \left( 
\left|\bar G^{(m,n)}(\vr-\vr_k)\right|
+ \left| \frac{1}{\kappa} \bar G^{(m+1,n)}(\vr - \vr_k) \right|
+ \left|\frac{1}{\kappa} \bar G^{(m,n+1)}(\vr - \vr_k) \right|  \right)
\leq \frac{c_5}{2.5}.
\end{align*}
Here, $c_5$ is a numerical constant, and we used \cite[C.12, Table 6]{candes_towards_2014} and $\N/ \kappa \leq 0.5514$. 
%


\subsection{Step 3b: $Q(\vr)$ and $\bar Q(\vr)$ are close for all $\vr$}

We next use an $\epsilon$-net argument together with Lemma \ref{lem:diffqmbarqmongrid} to establish that $Q^{(m,n)}(\vr)$ is close to $\bar Q^{(m,n)}(\vr)$ with high probability uniformly for all $\vr \in [0,1]^2$. 
\begin{lemma}
Let $\epsilon, \delta > 0$. If
\begin{align}
L \geq \S \frac{c}{\epsilon^2} \log^3\left(\frac{c' L^6}{\delta \epsilon^2} \right)
\label{eq:condLleex}
\end{align}
then, with probability at least $1-\delta$, 
\begin{align}
\max_{\vr \in [0,1]^2, (m,n)\colon m+n \leq 2}
\frac{1}{\kappa^{n+m}} \left| Q^{(m,n)}(\vr) - \bar Q^{(m,n)}(\vr) \right| 
\leq \epsilon.
\label{eq:ledffbaere}
\end{align}
\label{lem:lemdffqbarqevry}
\end{lemma}
\begin{proof}
We start by choosing a set of points $\Omega$ (i.e., the $\epsilon$-net) that is sufficiently dense in the $\infty$-norm. Specifically, we choose the points in $\Omega$ on a rectangular grid such that 
\begin{align}
\max_{\vr \in [0,1]^2} \min_{\vr_g \in \Omega} \infdist{ \vr - \vr_g} \leq \frac{\epsilon}{3 \tilde c L^{5/2}}.
\label{eq:gridmaxdist}
\end{align}
The cardinality of the set $\Omega$ is 
\begin{align}
|\Omega| = \left(\frac{3\tilde c L^{5/2}}{\epsilon}\right)^2 = c' L^5/\epsilon^2. 
\end{align}
First, we use Lemma \ref{lem:diffqmbarqmongrid} to show that 
$\left| Q^{(m,n)}(\vr_g) - \bar Q^{(m,n)}(\vr_g) \right|$ is small for all points $\vr_g \in \Omega$. 
Note that the condition of Lemma \ref{lem:diffqmbarqmongrid} is satisfied by assumption \eqref{eq:condLleex}. Using the union bound over all $6$ pairs $(m,n)$ obeying $m+n\leq 2$, it now follows from Lemma \ref{lem:diffqmbarqmongrid}, that
\begin{align}
\left\{
\max_{\vr_g \in \Omega, m+n\leq 2}  \frac{1}{\kappa^{m+n}} \left| Q^{(m,n)}(\vr_g) - \bar Q^{(m,n)}(\vr_g) \right| 
\leq \frac{\epsilon}{3}
\right\} 
\label{eq:QrmgbarQrg}
\end{align}
holds with probability at least $1-  6\delta' = 1- \frac{\delta}{2}$. Here, $\delta'$ is the original $\delta$ in Lemma \ref{lem:diffqmbarqmongrid}. 
Next, we will prove that this result continues to hold uniformly for all $\vr \in [0,1]^2$. 
To this end, we will also need that the event
%
\begin{align}
\left\{
\max_{\vr \in [0,1]^2, m+n\leq 2}\frac{1}{\kappa^{m+n}} \left | Q^{(m,n)}(\vr) \right| \leq  \frac{\tilde c}{2} L^{3/2}  
\right\}
\label{eq:QrmQrg}
\end{align}
holds with probability at least $1-\frac{\delta}{2}$. This is shown in Section \ref{sec:techres1} below. 
By the union bound, the events in \eqref{eq:QrmgbarQrg} and \eqref{eq:QrmQrg} hold simultaneously with probability at least $1-\delta$. The proof is concluded by noting that   \eqref{eq:QrmgbarQrg} and \eqref{eq:QrmQrg} imply \eqref{eq:ledffbaere} (see Section \ref{sec:techres2}). 

\subsubsection{\label{sec:techres1} Proof of the fact that \eqref{eq:QrmQrg} holds with probability at least $1-\frac{\delta}{2}$:} 
In order to show that \eqref{eq:QrmQrg} holds with probability at least $1-\frac{\delta}{2}$, we first upper-bound $|Q^{(m,n)}(\vr)|$. By \eqref{eq:Qmninnprodform},
\begin{align}
\frac{1}{\kappa^{m+n}} \left | Q^{(m,n)}(\vr) \right|   
&= \left| \innerprod{ \mL \vu }{ \vv^{(m,n)}(\vr)}  \right| \nonumber \\
&\leq  \norm[\opnormss]{{\mL}} \norm[2]{ \vu }  \norm[2]{ \vv^{(m,n)} (\vr)} \nonumber \\
&\leq \norm[\opnormss]{{\mL}}  \sqrt{\S}  \norm[2]{ \vv^{(m,n)} (\vr)} \nonumber \\
&\leq \norm[\opnormss]{{\mL}} \sqrt{\S}   \sqrt{3\S}\norm[\infty]{ \vv^{(m,n)} (\vr)} \nonumber \\
&= \norm[\opnormss]{{\mL}} \sqrt{3} \, \S \max_{j, (m', n') \in \{ (0,0), (1,0), (0,1) \}} 
\frac{1}{\kappa^{m+m'+n+n'}} \left | G^{(m,n)}_{(m',n')}(\vr,\vr_j) \right|,
\label{eq:ubqmnarr}
\end{align}
where we used $\norm[2]{\vu} = \sqrt{\S}$, since the entries of $\vu$ are $\pm 1$. 
Next, note that, for all $\vr$ and all $\vr_j$, we have, by~\eqref{eq:gmnqform} 
\begin{align}
\frac{1}{\kappa^{m+m'+n+n'}} \left | G^{(m,n)}_{(m',n')}(\vr,\vr_j) \right| 
&=  
\frac{1}{\kappa^{m+m'+n+n'}} \herm{\vx} \mV^{(m,n)}_{(m',n')}(\vr,\vr_j) \vx 
\leq 
\frac{1}{\kappa^{m+m'+n+n'}} \norm[2]{\vx}^2 \norm[\opnormss]{ \mV^{(m,n)}_{(m',n')}(\vr,\vr_j) }  
\nonumber \\
&\leq c_1\frac{(2\pi \N)^{m + m' + n + n'} }{\kappa^{m + m'+n+n'}}  \sqrt{\L}  \norm[2]{\vx}^2 
\leq c_1 12^{\frac{m + m'+n+n'}{2}} \sqrt{\L} \norm[2]{\vx}^2 \nonumber \\
&\leq
c_1 12^{\frac{3}{2}} \sqrt{\L} \norm[2]{\vx}^2,
\label{eq:fbongmnra}
\end{align}
where we used Lemma \ref{lem:boundonVFnorm} to conclude $\norm[\opnormss]{ \mV^{(m,n)}_{(m',n')}(\vr,\vr_j) } \leq \norm[F]{ \mV^{(m,n)}_{(m',n')}(\vr,\vr_j) } \leq c_1 (2\pi \N)^{m + m'+ n + n'} \sqrt{L}$ and \eqref{eq:fbongmnra} follows from $m+m'+n+n'\leq 3$ (recall that $m+n \leq 2$). 
Substituting \eqref{eq:fbongmnra} into \eqref{eq:ubqmnarr} and using that $\S \leq \L$  (by assumption \eqref{eq:condLleex}) yields 
\[
\frac{1}{\kappa^{m+n}} \left | Q^{(m,n)}(\vr) \right|  
\leq 
\sqrt{3} \, 12^{\frac{3}{2}}  c_1  L^{3/2}
\norm[\opnormss]{{\mL}}  \norm[2]{\vx}^2. 
\]
It follows that (with $\frac{\tilde c}{2} = 2.5 \cdot 3 \cdot \sqrt{3} \, 12^{\frac{3}{2}}  c_1$)
\begin{align}
\PR{\max_{\vr \in [0,1]^2, m+n\leq 2}\frac{1}{\kappa^{m+n}} \left | Q^{(m,n)}(\vr) \right| 
\geq \frac{\tilde c}{2} L^{3/2}
} 
&\leq
\PR{\norm[\opnormss]{{\mL}}  \norm[2]{\vx}^2 \geq 2.5 \cdot 3} \nonumber \\
&\leq 
\PR{\norm[\opnormss]{{\mL}} \geq 2.5 }
+
\PR{ \norm[2]{\vx}^2 \geq 3} \label{eq:ublapr} \\
&\leq \frac{\delta}{2} \label{eq:funbmaxqrgl}
\end{align}
as desired. 
Here, \eqref{eq:ublapr} follows from the union bound and \eqref{eq:funbmaxqrgl} follows from 
$
\PR{\norm[\opnormss]{{\mL}} \geq 2.5} \leq \PR{\comp{\mathcal E}_{1/4} } \leq \frac{\delta}{4}
$ (by \eqref{eq:normmLb} and application of Lemma \ref{lem:preptaubound}; note that the condition of Lemma \ref{lem:preptaubound} is satisfied by \eqref{eq:condLleex}) and $\PR{ \norm[2]{\vx}^2 \geq 3} \leq \frac{\delta}{4}$, shown below. 
Using that $4\log(4/\delta) \leq \L$ (by \eqref{eq:condLleex}), we obtain
\begin{align}
\PR{\norm[2]{\vx}^2 \geq 3   } 
&\leq \PR{\norm[2]{\vx}^2 \geq 2\left(1 + \frac{2 \log(4/\delta) }{L}  \right) } \nonumber \\
&\leq \PR{\norm[2]{\vx} \geq \left(1 + \frac{\sqrt{2 \log(4/\delta)}}{\sqrt{L}}  \right) }
\leq e^{- \frac{2 \log(4/\delta) }{2}} = \frac{\delta}{4},
\end{align}
where we used $\sqrt{2(1+\beta^2)}\geq (1+\beta)$, for all $\beta$, and a standard concentration inequality for the norm of a Gaussian random vector, e.g., \cite[Eq. 1.6]{ledoux_probability_1991}. 
This concludes the proof of \eqref{eq:QrmQrg} holding with probability at least $1-\frac{\delta}{2}$. 

\subsubsection{\label{sec:techres2} Proof of the fact that \eqref{eq:QrmgbarQrg} and \eqref{eq:QrmQrg} imply \eqref{eq:ledffbaere}:}
Consider a point $\vr \in [0,1]^2$ and let $\vr_g$ be the point in $\Omega$ closest to $\vr$ in $\infty$-distance. By the triangle inequality,
\begin{align}
&\frac{1}{\kappa^{n+m}} \left| Q^{(m,n)}(\vr) - \bar Q^{(m,n)}(\vr) \right| 
\leq \nonumber \\
&\hspace{0.5cm}\frac{1}{\kappa^{n+m}} \left[ 
\left| Q^{(m,n)}(\vr) - Q^{(m,n)}(\vr_g) \right| 
+ \left| Q^{(m,n)}(\vr_g) - \bar Q^{(m,n)}(\vr_g) \right| 
+ \left| \bar Q^{(m,n)}(\vr_g) - \bar Q^{(m,n)}(\vr) \right|
\right].
\label{eq:Qmqgrid}
\end{align}
We next upper-bound the terms in \eqref{eq:Qmqgrid} separately. With a slight abuse of notation, we write $Q^{(m,n)}(\x,\y) = Q^{(m,n)}(\transp{[\x,\y]}) \allowbreak = Q^{(m,n)}(\vr)$. The first absolute value in \eqref{eq:Qmqgrid} can be upper-bounded according to
\begin{align}
\left| Q^{(m,n)}(\vr) - Q^{(m,n)}(\vr_g) \right| 
&= \left| Q^{(m,n)}(\x,\y) - Q^{(m,n)}(\x,\y_g) + Q^{(m,n)}(\x,\y_g)    - Q^{(m,n)}(\x_g,\y_g) \right| \nonumber \\
&\leq \left| Q^{(m,n)}(\x,\y) - Q^{(m,n)}(\x,\y_g)\right| + \left| Q^{(m,n)}(\x,\y_g)    - Q^{(m,n)}(\x_g,\y_g) \right| \nonumber \\
&\leq |\y - \y_g| \sup_{z} \left|Q^{(m,n+1)}(\x,z)\right|
+ |\x - \x_g| \sup_{z}  \left|Q^{(m+1,n)}(z,\y_g)\right|  \nonumber \\
&\leq  |\y- \y_g| 2 \pi \N \sup_{z}  \left|Q^{(m,n)}(\x,z)\right|
+ |\x - \x_g| 2 \pi \N \sup_{z}  \left|Q^{(m,n)}(z,\y_g)\right|, 
\label{eq:ubqrmqrg}
\end{align}
where \eqref{eq:ubqrmqrg} follows from Bernstein's polynomial inequality, stated below (note that $Q^{(m,n)}(\x,\y)$ is a trigonometric polynomial of degree $\N$ in both $\x$ and $\y$). 
%
%
%
%
\begin{proposition}[Bernstein's polynomial inequality {\cite[Cor.~8]{harris_bernstein_1996}
}]
Let $p(\theta)$ be a trigonometric polynomial of degree $\N$ with complex coefficients $p_k$, i.e., $p(\theta) = \sum_{k=-\N}^{\N}  p_k e^{i2\pi \theta k}$. Then 
\[
\sup_{\theta} \left| \frac{d}{d\theta} p(\theta)  \right|  \leq 2 \pi \N \sup_{\theta} |p(\theta)|.
\]
\label{prop:bernstein}
\end{proposition}
%
%
Substituting \eqref{eq:QrmQrg} into \eqref{eq:ubqrmqrg} yields
\begin{align}
\frac{1}{\kappa^{m+n}} \left| Q^{(m,n)}(\vr) - Q^{(m,n)}(\vr_g) \right| 
\leq  \frac{\tilde c}{2} L^{5/2} 
( |\x- \x_g| + |\y- \y_g|)
\leq
\tilde c  L^{5/2}
\infdist{\vr - \vr_g}
\leq \frac{\epsilon}{3}, 
\label{eq:diffQrQrg}
\end{align}
where the last inequality follows from \eqref{eq:gridmaxdist}. 

We next upper-bound the third absolute value in \eqref{eq:Qmqgrid}. 
Using steps analogous to those leading to \eqref{eq:diffQrQrg}, 
we obtain  
\begin{align}
\frac{1}{\kappa^{m+n}} \left| \bar Q^{(m,n)}(\vr_g) - \bar Q^{(m,n)}(\vr) \right| 
\leq \frac{\epsilon}{3}. 
\label{eq:barQrmbarQrg}
\end{align}
%
Substituting \eqref{eq:QrmgbarQrg}, \eqref{eq:diffQrQrg}, and \eqref{eq:barQrmbarQrg} into \eqref{eq:Qmqgrid} yields that
\[
\frac{1}{\kappa^{n+m}} \left| Q^{(m,n)}(\vr) - \bar Q^{(m,n)}(\vr) \right|  \leq \epsilon,  \text{ for all } (m,n)\colon m+n \leq 2 \text{ and for all } \vr \in [0,1]^2.
\]
This concludes the proof of Lemma \ref{lem:lemdffqbarqevry}.  
\end{proof}

\subsection{Step 3c: Ensuring that $\abs{Q(\vr)} < 1$ for all $\vr \notin \T$}

\begin{lemma}
Suppose that 
\[
L \geq \S c  \log^3\left(\frac{c' L^6}{\delta} \right).
\]
Then with probability at least $1 - \delta$ the following statements hold:
\begin{enumerate}
\item \label{it:stat1} For all $\vr$, that satisfy $\min_{\vr_j \in \T} \infdist{\vr - \vr_j } \geq 0.2447/\N$ 
we have that 
$
\abs{Q(\vr)} < 0.9963.
$
\item \label{it:stat2} For all $\vr \notin \T$ that satisfy $0 < \infdist{\vr - \vr_j} \leq 0.2447/\N$ for some $\vr_j \in \T$, we have that $\abs{Q(\vr)} < 1$. 
\end{enumerate}
\end{lemma}
\begin{proof}
Choose $\epsilon =0.0005$. It follows from Lemma \ref{lem:lemdffqbarqevry} that 
\begin{align}
\frac{1}{\kappa^{n+m}} \left| Q^{(m,n)}(\vr) - \bar Q^{(m,n)}(\vr) \right| 
\leq 0.0005
\label{eq:conddiffQbarQinfipr}
\end{align}
 for all $(m,n)\colon m+n \leq 2$, and for all $\vr$  with probability at least $1 - \delta$. To prove the lemma we will show that statements \ref{it:stat1} and \ref{it:stat2} follow from \eqref{eq:conddiffQbarQinfipr} and certain properties of $\bar Q^{(m,n)}(\vr)$ established in \cite{candes_towards_2014}. 

Statement \ref{it:stat1} follows directly by combining \eqref{eq:conddiffQbarQinfipr} with the following result via the triangle inequality. 
\begin{proposition}[{\cite[Lem.~C.4]{candes_towards_2014}}]
For all $\vr$, that satisfy $\min_{\vr_j \in \T} |\vr - \vr_j | \geq 0.2447/\N$ 
we have that 
$
\abs{Q(\vr)} < 0.9958.
$ 
\end{proposition}


In order to prove statement \ref{it:stat2}, assume without loss of generality that $\vect{0} \in \T$, and consider $\vr$ with $|\vr| \leq 0.2447/\N$. Statement \ref{it:stat2} is established by showing that the Hessian matrix of $\tilde Q(\vr) \defeq |Q(\vr)|$, i.e.,
\[
\mH = 
\begin{bmatrix}
\tilde Q^{(2,0)}(\vr) & \tilde Q^{(1,1)}(\vr) \\
\tilde Q^{(1,1)}(\vr) & \tilde Q^{(0,2)}(\vr) 
\end{bmatrix},
\quad \tilde Q^{(m,n)}(\vr) \defeq \frac{\derd^m }{ \derd \x^m} \frac{\derd^n }{ \derd \y^n}  \tilde Q(\vr)
\]
is negative definite. This is done by showing that 
\begin{align}
\mathrm{trace}(\mH) = \tilde Q^{(2,0)} + \tilde Q^{(0,2)} < 0 \label{eq:traceHleq0} \\ 
\mathrm{det}(\mH)     = \tilde Q^{(2,0)} \tilde Q^{(0,2)} - (\tilde Q^{(1,1)})^2   > 0,  \label{eq:detHgeq0}
\end{align}
which implies that both eigenvalues of $\mH$ are strictly negative. To prove \eqref{eq:traceHleq0} and \eqref{eq:detHgeq0}, we will need the following result.

\begin{proposition}[{\cite[Sec.~C.2]{candes_towards_2014}}]
For $|\vr| \leq 0.2447/\N$ and for $\N \geq 512$,
\begin{align}
1\geq \bar Q(\vr) \geq 0.6447 \\
\frac{1}{\kappa^2}\bar Q^{(2,0)}(\vr) \leq -0.3550 \\
\frac{1}{\kappa^2}|\bar Q^{(1,1)}(\vr)| \leq 0.3251 \\
\frac{1}{\kappa^2}|\bar Q^{(1,0)}(\vr)| \leq 0.3344.
\end{align}
\label{propcanc2}
\end{proposition}

Define $Q_R^{(m,n)} = \frac{1}{\kappa^{m+n}} \mathrm{Re}(Q^{(m,n)})$ and $Q_I^{(m,n)} = \frac{1}{\kappa^{m+n}} \mathrm{Im}(Q^{(m,n)})$. We have that 
\[
\frac{1}{\kappa}\tilde Q^{(1,0)} = \frac{Q_R^{(1,0)}Q_R + Q_I^{(1,0)}Q_I  }{|Q|}
\]
therefore 
\begin{align}
\frac{1}{\kappa^2}
\tilde Q^{(2,0)} 
&= -\frac{(Q_R Q_R^{(1,0)} + Q_I Q_I^{(1,0)})^2 }{|Q|^3} 
+ \frac{|Q^{(1,0)}|^2 + Q_R Q_R^{(2,0)}  + Q_I Q_I^{(2,0)} }{|Q|} \nonumber \\
&=-\frac{Q_R^2 {Q_R^{(1,0)}}^2 + 2Q_R Q_R^{(1,0)}  Q_I Q_I^{(1,0)} + Q_I^2 {Q_I^{(1,0)}}^2  }{|Q|^3} 
+ \frac{ {Q_R^{(1,0)}}^2 + {Q_I^{(1,0)}}^2 +  Q_R Q_R^{(2,0)}  + Q_I Q_I^{(2,0)}  }{|Q|} \nonumber \\
&=
\left(1-\frac{Q_R^2}{|Q|^2}\right) \frac{{Q_R^{(1,0)}}^2}{|Q|} -\frac{2Q_R Q_R^{(1,0)}  Q_I Q_I^{(1,0)} + Q_I^2 {Q_I^{(1,0)}}^2  }{|Q|^3} 
 + \frac{{Q_I^{(1,0)}}^2   + Q_I Q_I^{(2,0)} }{|Q|}+ \frac{Q_R}{|Q|} Q_R^{(2,0)}.
\label{eq:lhsoftildq20}
\end{align}
By Proposition \ref{propcanc2}, using the triangle inequality, and using the fact that $\bar Q^{(m,n)}(\vr)$ is real, the following bounds are in force: 
\begin{align*}
  Q_R(\vr) &\leq \bar Q(\vr) + \epsilon \leq 1+\epsilon  \\
  Q_R(\vr) &\geq \bar Q(\vr) - \epsilon \geq 0.6447 - \epsilon \\
Q_I^{(m,n)} &\leq \epsilon \\
Q_R^{(2,0)}(\vr)  &\leq  \frac{1}{\kappa^2}\bar Q^{(2,0)}(\vr) + \epsilon \leq -0.3550 +\epsilon \\
|Q_R^{(1,1)}| &\leq \frac{1}{\kappa^2}|\bar Q^{(1,1)}(\vr)| + \epsilon \leq 0.3251 + \epsilon \\
|Q^{(1,0)}_R(\vr)| &\leq \frac{1}{\kappa^2} |\bar Q^{(1,0)}(\vr)| + \epsilon \leq 0.3344 +\epsilon.
\end{align*}
Using these bounds in \eqref{eq:lhsoftildq20} 
with $\epsilon = 0.0005$ we obtain $\tilde Q^{(2,0)} < -0.3539$, which implies that \eqref{eq:traceHleq0} is satisfied. 
 
It remains to verify \eqref{eq:detHgeq0}. First note that 
\begin{align}
&\frac{1}{\kappa^2}
\tilde Q^{(1,1)} \nonumber \\
&= 
\frac{Q_R^{(1,1)}Q_R + Q_R^{(1,0)}Q_R^{(0,1)} + Q_I^{(1,1)}Q_I + Q_I^{(1,0)}Q_I^{(0,1)} }{|Q|} 
- \frac{ (Q_R^{(0,1)}Q_R + Q_I^{(0,1)}Q_I ) (Q_R^{(1,0)}Q_R + Q_I^{(1,0)}Q_I)  }{|Q|^3} \nonumber \\
&= Q_R^{(1,1)} \frac{Q_R}{|Q|} + \frac{Q_R^{(1,0)}Q_R^{(0,1)}}{|Q|} \left(1- \frac{Q_R^2}{|Q|^2} \right) 
+ \frac{ Q_I^{(1,1)}Q_I + Q_I^{(1,0)}Q_I^{(0,1)} }{|Q|} \nonumber \\
&- \frac{Q_R^{(0,1)}Q_R Q_I^{(1,0)}Q_I +  Q_I^{(0,1)}Q_I (Q_R^{(1,0)}Q_R + Q_I^{(1,0)}Q_I)  }{|Q|^3}. \label{eq:tilQ11}
\end{align}
Using the bounds above in \eqref{eq:tilQ11} yields,
with $\epsilon = 0.0005$, that $\frac{1}{\kappa^2}|\tilde Q^{(1,1)}| \leq 0.3267$. With $\frac{1}{\kappa^2}\tilde Q^{(2,0)} < -0.3539$, it follows that the RHS of \eqref{eq:traceHleq0} can be lower-bounded by   
\[
\frac{1}{\kappa^2}(0.3539^2 - 0.3267^2)  = \frac{1}{\kappa^2}0.01855 > 0,
\]
i.e., \eqref{eq:detHgeq0} holds. This concludes the proof of Statement \ref{it:stat2}. 
\end{proof}

\section*{Funding}
VM was supported by the Swiss National Science Foundation fellowship for advanced researchers under grant PA00P2\_139678.

\section*{Acknowledgments}
RH would like to thank C\'{e}line Aubel, Helmut B\"{o}lcskei, Emmanuel Cand\`{e}s, and Nora Loose for helpful discussions.
RH would also like to thank Emmanuel Cand\`{e}s 
for his hospitality during a visit to the Statistics Department at  Stanford, and Helmut B\"{o}lcskei for his support and for initiating this visit.  
We would also like to thank the referees for helpful comments and suggestions, which greatly improved the manuscript. 
The authors would like to thank Arash Amini, Mohammad Mahdi Kamjoo, Saeed Razavikia, and Sajad Daei for making us aware of an error in an earlier proof of Lemma~\ref{lem:preptaubound}, which we have now fixed. 
The authors would also like to thank Radoslaw Adamczak on helpful pointers to the literature on matrix Hanson-Wright inequalities.


\bibliographystyle{plain}
\bibliography{bibliography.bib}

\appendices

\section{Equivalence of 
\eqref{eq:periorel} 
and
\eqref{eq:iowithdirchkernel1} 
\label{app:discspfunc}
}

Note that with
\begin{align}
D_{\N}(t) 
= \frac{1}{L} \sum_{k=-\N}^{\N} e^{i2\pi t k}
\label{eq:usingDirichTrigon}
\end{align}
we obtain
\begin{align}
\sum_{r = -\N}^\N 
e^{i2\pi \frac{rp}{\L}}
D_{\N} \! \left( \frac{r }{\L} - \nu_j \right)
&=
\sum_{r = -\N}^\N 
e^{i2\pi \frac{rp}{\L}}
\frac{1}{L} \sum_{k=-\N}^{\N} e^{i2\pi k \left( \frac{r }{\L} - \nu_j \right)}  \nonumber \\
&=
\sum_{k=-\N}^{\N} 
e^{-i2\pi k \nu_j}  
\frac{1}{L}
\sum_{r = -\N}^\N 
e^{i2\pi \frac{r}{\L}  (p+k)}
\nonumber \\
&=
e^{i2\pi p \nu_j}. 
\label{eq:sumdirfre}
\end{align}
Using \eqref{eq:sumdirfre} in \eqref{eq:iowithdirchkernel1} and using again \eqref{eq:usingDirichTrigon}, we get 

\begin{align}
y_p
&= \sum_{j=1}^{\S}  b_j
e^{i2\pi p \nu_j }
\sum_{\ell=-\N}^{\N}   
\frac{1}{L} \sum_{k=-\N}^{\N} e^{i2\pi k \left( \frac{\ell}{L}  -  \tau_j \right) }  
x_{p-\ell}  
\nonumber \\
&= \sum_{j=1}^{\S} b_j
e^{i2\pi p  \nu_j }  
\frac{1}{L} \sum_{\ell,k=-\N}^{\N}  e^{-i2\pi k \tau_j }   e^{i2\pi (p-\ell) \frac{k}{L}  }  x_{\ell},  \nonumber
\end{align}
where we used that $x_\ell$ is $\L$-periodic. This is \eqref{eq:periorel}, as desired. 



\section{\label{app:iorelproof}Proof of \eqref{eq:sfunc_discrete}}

First, write \eqref{eq:ltvsys} in its equivalent form
\begin{equation}
y(t) = \int  \tfunc(t,f) \prsigF(f) e^{i2\pi f t} df,
\label{eq:rep_tvtf}
\end{equation}
where $\prsigF(f) = \int x(t) e^{i2\pi f t} dt$ is the Fourier transform of $\prsig(t)$, and $\tfunc(t,f)$ is the time-varying transfer function given by 
\begin{equation}
\tfunc(t,f) \defeq \iint \sfunc(\tau,\nu) e^{i2\pi (\nu t - \tau f)} d\tau d \nu.
\label{eq:rel_sf_tvtf}
\end{equation}
Since  $\prsig(t)$ is band-limited to $[-\Bint/2,\Bint/2]$, we may write 
\[
\prsigF(f) = \prsigF(f) H_I(f), \quad H_I(f) \defeq \begin{cases} 1, & |f| \leq B/2 \\ 0,  &\text{else}.   \end{cases}
\]
For $y(t)$ on $[-\Tint/2,\Tint/2]$ we may write  
 \[
y(t) =  y(t) h_O(t), \quad h_O(t) \defeq \begin{cases} 1, & |t| \leq  \Tint/2 \\ 0,  &\text{else}.   \end{cases}
\]
With the input band-limitation and the output time-limitation, 
\eqref{eq:rep_tvtf} becomes 
\begin{equation}
y(t) = \int \btlim{\tfunc} (t,f)    \prsigF(f) e^{i2\pi f t} df,
\label{eq:inoutbtlim}
\end{equation}
where 
\begin{equation}
\btlim{\tfunc} (t,f) \defeq \tfunc(t,f)  h_O(t) H_I(f)
\label{eq:TVTF}
\end{equation}
i.e., the effect of input band-limitation and output time-limitation is accounted for by passing the probing signal through a system with time varying transfer function given by $\btlim{\tfunc}$. 
The spreading function $\btlim{\sfunc}$ of the system \eqref{eq:inoutbtlim} and the time-varying transfer function $\btlim{\tfunc}$ are related by the two-dimensional Fourier transform in \eqref{eq:rel_sf_tvtf}. 
We see that $\btlim{\tfunc} (t,f)$ is ``band-limited'' with respect to $t$ and $f$, and hence, by the sampling theorem, $\btlim{\sfunc}$ can be expressed in terms of its samples as 
\begin{align}
\btlim{ \sfunc}(\tau,\nu) =   \sum_{m, \ell \in \mb Z}  \btlim{\sfunc} \left(\frac{m}{B},\frac{l}{\Tint}\right) \sinc \left(  \left( \tau  - \frac{m}{B}\right)B\right) \sinc\left( \left(\nu - \frac{\ell}{\Tint}\right)\Tint  \right). 
\label{eq:samplthmsbar}
\end{align}
In terms of $\btlim{ \sfunc}(\tau,\nu)$ \eqref{eq:inoutbtlim} can be written as 
\begin{equation}
y(t) = \iint  \overline{\sfunc} (\tau,\nu) \prsig(t-\tau)  e^{i2\pi \nu t}  d\nu d\tau
\label{eq:ltvsys2}
\end{equation}
and with \eqref{eq:samplthmsbar}
\begin{align*}
y(t) 
&= 
  \sum_{m, \ell \in \mb Z}  \btlim{\sfunc}  \left(\frac{m}{B},\frac{\ell}{\Tint}\right) 
  \int  \sinc\left( \left(\tau  - \frac{m}{B}\right)B\right) \prsig(t-\tau) d\tau \int   \sinc\left( \left(\nu - \frac{\ell}{\Tint}\right)\Tint  \right)  e^{i2\pi \nu t}  d\nu 
\\%
&= 
 \sum_{m,\ell \in \mb Z} \btlim{\sfunc} \left(\frac{m}{B},\frac{\ell}{\Tint} \right) \prsig \left( t-\frac{m}{B}  \right) e^{j 2 \pi \frac{\ell}{\Tint} t}. 
\end{align*}
According to \eqref{eq:rel_sf_tvtf} and \eqref{eq:TVTF},  $\btlim{\sfunc }(\tau,\nu)$ and $\sfunc(\tau,\nu)$ are related as
\begin{align}
\overline{\sfunc}(\tau,\nu) \label{eq:btlimiorel2} = 
 \int \!\int \! \sfunc(\tau',\nu')    \sinc( (\tau \!- \!\tau') B ) \sinc( (\nu \!-\! \nu')  \Tint  )  d\tau' d\nu'.
\end{align}
This concludes the proof of \eqref{eq:sfunc_discrete}.

\section{Proof of \eqref{eq:iowithdirchkernelapprox} and \eqref{eq:iowithdirchkernel} \label{eq:proofeq:iowithdirchkernel} }

Starting with \eqref{eq:sfunc_discretesample2}, we first use that $BT = \L$ and  $\tau_j = \tauc_j \frac{B}{L}$ and $\nu_j = \nuc_j \frac{T}{L}$ to obtain 
\begin{align}
\TL{y}_p
= 
\sum_{j=1}^{\S} b_j
\left(
\sum_{\ell \in \mb Z} 
\sinc(\ell - \tau_j \L) \TL{x}_{p-\ell}  
\right)
\left(
\sum_{k \in \mb Z} 
\sinc(k - \nu_j \L) 
e^{i2\pi \frac{k p}{\Bint \Tint}}
\right).
\label{eq:firststepyp}
\end{align}
Next, changing the order of summation according to $k = r + \L q$ with $r = \N,...,\N$ and $q \in \mb Z$ yields
\begin{align}
\sum_{k \in \mb Z} 
\sinc(k - \nu_j \L) 
e^{i2\pi \frac{k p}{\Bint \Tint}}
&=
\sum_{r = -\N}^\N 
\sum_{q \in \mb Z} 
\sinc\!\left( \left( \frac{r }{\L} - \nu_j +q \right) \L \right) 
e^{i2\pi \frac{(r + \L q) p}{\L}}    \nonumber \\
&= 
\sum_{r = -\N}^\N 
e^{i2\pi \frac{rp}{\L}}
\sum_{q \in \mb Z} 
\sinc\!\left( \left( \frac{r }{\L} - \nu_j +q \right) \L \right) 
    \nonumber \\
&=
\sum_{r = -\N}^\N 
e^{i2\pi \frac{rp}{\L}}
D_{\N} \! \left( \frac{r }{\L} - \nu_j \right).
\label{eq:iowithdirchkernel2}
\end{align}
Here, \eqref{eq:iowithdirchkernel2} follows from the definition of the Dirichlet kernel in \eqref{eq:dirichletsum}. 
Next, note that 
\begin{align}
\sum_{\ell \in \mb Z}
\TL{x}_{p-\ell}
\sinc\left( \ell  -  \L \tau_j  \right)
&=
\sum_{\ell \in \mb Z}
\TL{x}_{\ell}
\sinc\left( p - \ell - \L \tau_j \right) \nonumber \\
&=
\sum_{\tilde \ell = -\N}^\N
x_{\tilde \ell}
\sum_{k=-1}^1
\sinc\left( p - \tilde \ell  - \L k- \L \tau_j \right) \label{eq:usetruncsam} \\
&=
\sum_{\ell = -\N}^\N
x_{\ell}
\TL{D}_\N \left( \frac{p-\ell}{\L}  - \tau_j \right),
 \label{eq:bydefDN} 
\end{align}
where \eqref{eq:usetruncsam} follows from $\TL{x}(\ell/B) = \TL{x}_\ell =  x_\ell$ for $\ell \in [-\N - \L, \L + \N]$ and $\TL{x}(\ell/B) = 0$ for all other $\ell$, and from $\L$-periodicity of $x_\ell$; 
 \eqref{eq:bydefDN} follows by the definition of $\TL{D}_\N(t)$ in~\eqref{eq:dirichletsum1}. 
Substituting \eqref{eq:iowithdirchkernel2} and \eqref{eq:bydefDN} into \eqref{eq:firststepyp} yields \eqref{eq:iowithdirchkernelapprox}, as desired.

Next, consider the case where $x(t)$ is $T$-periodic. In this case, the samples $x_\ell = x(\ell/B)$ are $L$-periodic. 
Similarly to the derivation above, changing the order of summation of the sum in the second brackets in \eqref{eq:firststepyp} according to $\ell = \tilde \ell + \L k$ yields 
\begin{align}
\sum_{\ell \in \mb Z} 
\sinc(\ell - \tau_j \L) \prsig_{p-\ell} 
= \sum_{\tilde \ell = -\N}^\N 
x_{p-\tilde \ell}
\sum_{k \in \mb Z}
\sinc\left(  \left( \frac{\tilde \ell}{\L} - \tau_j + k \right) \L   \right) 
= \sum_{\ell = -\N}^\N 
D_{\N} \! \left( \frac{\ell}{\L} - \tau_j \right)
x_{p- \ell},
\label{eq:sincuxper}
\end{align}
where we used that $x_\ell$ is $\L$-periodic in the first equality, and the last equality follows by definition of the Dirichlet kernel in \eqref{eq:dirichletsum}. Using  \eqref{eq:iowithdirchkernel2} and \eqref{eq:sincuxper} in \eqref{eq:firststepyp} yields \eqref{eq:iowithdirchkernel}. 

\section{Proof of Proposition \ref{prop:errbound} \label{app:errbound}}

By \eqref{eq:sumdirfre}, \eqref{eq:iowithdirchkernel2} is equal to $e^{i2\pi p \nu_j}$. Using this in \eqref{eq:iowithdirchkernelapprox} and \eqref{eq:iowithdirchkernel} yields 
\[
\TL{y}_p  - y_p = 
\sum_{j=1}^{\S} b_j
e^{i2\pi  \nu_j p }
\underbrace{
\sum_{\ell=-\N}^\N
\epsilon_{j,p,\ell}  x_\ell 
}_{\epsilon_{j,p}},
\]
where we defined the error 
\[
\epsilon_{j,p,\ell} 
\defeq 
\tilde D_\N\left(  \frac{p - \ell}{\L} - \tau_j \right) - D_\N\left( \frac{p - \ell}{\L} - \tau_j  \right). 
\]
For $t\in [-1.5,..,1.5]$, we have 
\[
\tilde D_\N(t) - D_\N(t)  \leq c/L.
\]
Using that $\frac{p - \ell}{\L} - \tau_j \in [-1.5,1.5]$ it follows that $\epsilon_{j,p,\ell} \leq c/L$.

Recall that the $b_j$ have random sign and the $x_\ell$ are i.i.d.~$\mc N(0,1/\L)$ distributed. By the union bound, we have, for all $\beta>0$,
\begin{align}
\PR{|\TL{y}_p - y_p| 
\geq
\beta
 \norm[2]{\vb} 
c \frac{\beta}{\L}
}
&\leq 
\PR{
\left| 
\sum_{j=1}^{\S} b_j e^{i2\pi p  \nu_j  } \epsilon_{j,p}
\right| 
\geq 
\beta \norm[2]{\vb} 
\max_{j,p} |\epsilon_{j,p}|
}  
+
\PR{
\max_{j,p} |\epsilon_{j,p}| \geq  c \frac{\beta}{\L}
} \label{eq:ubstebou}   \\
&\leq (4+2\L^2) e^{-\frac{\beta^2}{L}},
\label{eq:errboundf}
\end{align}
where \eqref{eq:errboundf} is established immediately below. Setting $\alpha=\beta^2$ completes the proof of Proposition~\ref{prop:errbound}. 

It is left to establish \eqref{eq:errboundf}. 
With 
\[
\left(
\sum_{j=1}^\S 
\left(
|b_j| e^{i2\pi p  \nu_j  }
\epsilon_{j,p}
\right)^2
\right)^{1/2}
\leq \beta \norm[2]{\vb} \max_{j,p} |\epsilon_{j,p}|
\]
the first probability in \eqref{eq:ubstebou} can be upper-bounded by 
\[ 
\PR{
\left| 
\sum_{j=1}^{\S} \sign(b_j) |b_j| e^{i2\pi p  \nu_j  } \epsilon_{j,p}
\right| 
\geq 
\beta 
\left(
\sum_{j=1}^\S 
\left(
|b_j| e^{i2\pi p  \nu_j  }
\epsilon_{j,p}
\right)^2
\right)^{1/2} 
}
\leq 4 e^{- \frac{\beta^2}{2}}
\]
where we applied Hoeffding's inequality, i.e., Lemma \ref{thm:hoeff} (recall that the $b_j$ have random sign). 
By the union bound, the second probability in \eqref{eq:ubstebou} is upper-bounded by 
\begin{align}
\PR{
\max_{j,p}
\left| 
\epsilon_{j,p}
\right|
\geq  
c
\frac{\beta}{\L}
}
\leq 
\sum_{j,p} 
\PR{
\left| 
\epsilon_{j,p}
\right|
\geq  c
\frac{\beta}{\L}
}
\leq
\S \L  2 e^{- \frac{\alpha^2}{2}},
\label{eq:treatgausse}
\end{align}
where the last inequality is proven as follows. 
Since the $x_\ell$ are i.i.d.~$\mc N(0,1/\L)$, we have that 
$\epsilon_{j,p} = \sum_{\ell=-\N}^\N
\epsilon_{\ell, p,j} x_\ell$ is Gaussian with variance $\frac{1}{\L} \sum_{\ell=-\N}^\N
\epsilon_{\ell, p,j}^2  \leq \frac{c^2}{\L^2}$, where we used that $\epsilon_{\ell, p,j} \leq c/\L$. 
Eq.~\eqref{eq:treatgausse} now follows from a standard bound on the tail probability of a Gaussian \cite[Prop.~19.4.2]{lapidoth_foundation_2009} random variable.


%
%
%
%
%
%


\section{\label{sec:proofdiscrete}Proof of Theorem \ref{cor:discretesuperres}}

The following proposition---standard in the theory of compressed sensing (see e.g.,~\cite{candes_robust_2006})---shows that the existence of a certain dual polynomial guarantees that $\mathrm{L1}(\vy)$ in \eqref{eq:l1minmG} succeeds in reconstructing $\vs$. 

\begin{proposition}
Let $\vy=\mR\vs$ and let $\dT$ denote the support of $\vs$. Assume that $\mR_{\dT}$ has full column rank. If there exists a vector $\vv$ in the row space of $\mR$ with 
\begin{align}
\vv_{\dT} = \sign(\vs_{\dT}) \quad \text{and} \quad \norm[\infty]{\vv_{\comp{\dT}}} < 1
\label{eq:dualcerl1}
\end{align}
then $\vs$ is the unique minimizer of $\mathrm{L1}(\vy)$ in \eqref{eq:l1minmG}. 
\end{proposition}
The proof now follows directly from Proposition \ref{prop:dualpolynomial}. To see this, set $\vu = \sign(\vs_{\dT})$ in Proposition \ref{prop:dualpolynomial} and consider the polynomial $Q(\vr)$ from Proposition \ref{prop:dualpolynomial}. Define $\vv$ as
 $[\vv]_{(m,n)} = Q([m/\K,n/\K])$ and note that $\vv$ satisfies \eqref{eq:dualcerl1} since $Q([m/\K,n/\K]) = \sign(\vs_{(m,n)})$ for $(m,n) \in \dT$ and $\abs{Q([m/\K,n/\K])} < 1$ for $(m,n) \notin \dT$.


\section{\label{sec:boundU}Bound on $U$}

We have that 
\begin{align}
U( t ) 
&= 
\hat c (2\pi N)^{-m} \sum_{p=-\N}^{\N} \min\left(1, \frac{1}{p^4} \right)    
P^{(m)}(p/\L - t)
\label{eq:Pmbuform} 
\end{align}
with 
\[
P^{(m)}(t) 
\defeq  \frac{1}{M} \sum_{k=-\N}^{\N}  (-i2\pi k)^m e^{i2\pi t k  }
= \frac{\derd^m }{ \derd t^m}   \frac{\sin(L \pi t)}{M\sin(\pi t)}.
\]
We start by upper-bounding $|P^{(m)}(t)|$. First note that $|P^{(m)}(t)|$ is a $1$-periodic and symmetric function, thus in order to upper-bound $|P^{(m)}(t)|$, we only need to consider $t \in [0,1/2]$. 
 
For $m=0$, we have that 
\[
|P^{(0)}(t)| \leq  \min\left(4, \frac{1}{M |\sin(\pi t)|} \right).
\]
Next, consider the case $m=1$, and assume that $t\geq 1/\L$. We have
\[
P^{(1)}(t) = \frac{\cos(\L \pi t) \L \pi}{M \sin(\pi t)} - \frac{\pi \sin(\L \pi t) \cos(\pi t)  }{M \sin^2(\pi t)}.
\]
Using that $\sin(\pi t) \geq 2 t \geq 2/\L$ for $1/\L\leq t \leq 1/2$ we get 
\[
|P^{(1)}(t)| \leq  \frac{ 1.5 \L \pi}{M | \sin(\pi t)|}. 
\]
Next, consider the case $m=2$. We have 
\[
P^{(2)}(t)
= 
\frac{\pi^2 (1-\L^2) \sin(\L \pi t)   }{M \sin(\pi t)} 
- \frac{2 L \pi^2 \cos(\L \pi t) \cos(\pi t)}{M \sin^2(\pi t)}
+ \frac{2 \pi^2  \sin(L \pi t) \cos^2(\pi t) }{M \sin^3(\pi t)}.
\]
Using again that that $\sin(\pi t) \geq 2 t \geq 2/\L$ for $1/\L\leq t \leq 1/2$ we get
\[
|P^{(2)}(t)|
\leq
\frac{2.5 \L^2  \pi^2}{M |\sin(\pi t)|}.
\]
Analogously, we can obtain bounds for $m=3,4$. We therefore obtain, for $1/\L \leq |t| \leq 1/2$, 
\begin{align}
|P^{(m)}(t)| \leq (\L\pi)^m  \frac{c_1}{M |\sin(\pi t)|} \leq \underbrace{1.0039c_1}_{c_2 \defeq}   (2\pi\N)^m  \frac{1}{M |\sin(\pi t)|},
\label{eq:Pmb2}
\end{align}
where $c_1$ is a numerical constant and where we 
used that $(\L/(2\N))^m \leq 1.0039$ for $\N\geq 512$ and $m\leq 4$. Regarding the range $0 \leq |t| \leq 1/\L$, simply note that by Bernstein's polynomial inequality (cf.~Proposition \ref{prop:bernstein}) we have, for all $t$, from $|P^{(0)}(t)| \leq 4$, that 
\begin{align}
|P^{(m)}(t)|  \leq 4 (2\pi \N )^m. 
\label{eq:Pmb1}
\end{align}
%
%
%
With $c_2\geq1$, application of \eqref{eq:Pmb2} and \eqref{eq:Pmb1} on \eqref{eq:Pmbuform} yields
\begin{align*}
U( t ) 
&\leq 
\hat c \sum_{p=-\N}^{\N} \min\left(1, \frac{1}{p^4} \right)    
c_2
\begin{cases}
4, & |p/\L - t + n| \leq 1/\L, n \in \mathbb Z\\
\frac{1}{M |\sin(\pi (p/\L - t))|}, & \text{else}.
\end{cases}
\end{align*}
The RHS above is $1$-periodic in $t$ and symmetric around the origin. Thus, it suffices to consider $t \in [0,1/2]$. Assume furthermore that $Lt$ is an even integer, the proof for general $t$ is similar. 
For $p\geq 0$, we have that $|p/L - t| \leq 1/2$ and thus  
$M |\sin(\pi (p/L - t) )| \geq M |2 (p/L - t)| = 2M/L |p- Lt| \geq 1/2 |p-Lt|$. It follows that  
\begin{align}
U(t) 
&\leq 
\hat c  c_2\sum_{p=0}^{\N} \min\left(1, \frac{1}{p^4} \right)    \min\left(4, \frac{2}{| p - L t|} \right)
\nonumber \\
&\leq 
\hat c c_2\sum_{p=0}^{Lt/2} \min\left(1, \frac{1}{p^4} \right)    
 \frac{2}{L t - p} 
+
\sum_{p=Lt/2+1}^{Lt-1}  \frac{1}{p^4}  \frac{2}{L t-p} 
+
\sum_{p=Lt}^{\N} \frac{1}{p^4}  4   
\nonumber \\
&= 
\frac{\hat c c_2}{Lt}
\left(
\sum_{p=0}^{Lt/2} \min\left(1, \frac{1}{p^4} \right)    
 \frac{Lt}{L t - p} 
+
\sum_{p=Lt/2+1}^{Lt-1}  \frac{Lt}{p^4} \frac{2}{L t-p} 
+
\sum_{p=Lt}^{\N} \frac{4 Lt}{p^4}     
\right) \nonumber \\
&\leq 
\frac{\hat c c_2}{Lt}
\left(
\sum_{p=0}^{Lt/2} 2 \min\left(1, \frac{1}{p^4} \right)  
+
\sum_{p=Lt/2+1}^{Lt-1} 2\frac{2}{p^3} 
+
\sum_{p=Lt}^{\N} \frac{4}{p^3} 
\right) \nonumber \\
&\leq \frac{\tilde c}{Lt}. \nonumber
\end{align}
Analogously we can upper-bound the sum over $p=-\N,...,-1$, which yields $U(t) \leq \frac{c}{L|t|}$, as desired. 

\section{Proof of Proposition \ref{prop:dualmin} \label{sec:proofprop:dualmin}}

The argument is standard, see e.g., \cite[Prop.~2.4]{tang_compressed_2013}. 
By definition, $\vq$ is dual feasible. To see this, note that 
\begin{align}
\norm[\setA^\ast]{\herm{\mG} \vq} 
= \sup_{\vr \in [0,1]^2} \left| \innerprod{\herm{\mG} \vq}{\va(\vr)} \right|
= \sup_{\vr \in [0,1]^2} \left| \innerprod{\vq}{\mG\va(\vr)} \right| 
= \sup_{\vr \in [0,1]^2} |Q(\vr)| \leq 1
\label{eq:AdnormGqles}
\end{align}
where the last inequality holds by assumption. 
By \eqref{eq:dualpolyinatmincon}, we obtain
\begin{align}
\innerprod{\vq}{\vy} 
&= \innerprod{\vq}{ \mG \sum_{\vr_n \in \T} b_n \va(\vr_n) } 
= \sum_{\vr_n \in \T}  \conj{b}_n \innerprod{\vq}{ \mG \va(\vr_n) } 
= \sum_{\vr_n \in \T}  \conj{b}_n  \sign(b_n) 
= \sum_{\vr_n \in \T}  |b_n| 
\geq \norm[\setA]{\vz},
\end{align}
where the last inequality holds by definition of the atomic norm. By H\"older's inequality we have that 
\[
\Re \innerprod{\vq}{\vy} 
= \Re \innerprod{\vq}{\mG \vz} 
= \Re \innerprod{\herm{\mG}\vq}{\vz}
\leq \norm[\setA^\ast]{\herm{\mG} \vq} \norm[\setA]{\vz} \leq \norm[\setA]{\vz}
\]
where we used \eqref{eq:dualpolyinatmincon} for the last inequality. 
We thus have established that $\Re \innerprod{\vq}{\vy} = \norm[\setA]{\vz}$. 
Since $(\vz, \vq)$ is primal-dual feasible, it follows from strong duality that $\vz$ is a primal optimal solution and $\vq$ is a dual optimal solution. 

It remains to establish uniqueness. To this end, suppose that $\hat \vz = \sum_{\vr_n \in \hat \T} \hat b_n \va(\vr_n)$ with $\norm[\setA]{\hat \vz} = \sum_{\vr_n \in \hat \T} |\hat b_n|$ and $\hat \T \neq \T$ is another optimal solution. We then have 
\begin{align*}
\Re \innerprod{\vq}{\mG \hat \vz} 
&= \Re \innerprod{\vq}{\mG \sum_{\vr_n \in \hat \T} \hat b_n \va(\vr_n) }  \\
&= \sum_{\vr_n \in \T} \Re \left( \conj{\hat b}_n \innerprod{\vq}{\mG \va(\vr_n) } \right)
+ \sum_{\vr_n \in \hat \T \setminus \T} \Re \left( \conj{\hat b}_n \innerprod{\vq}{\mG \va(\vr_n) } \right)   \\
&< \sum_{\vr_n \in \T} |\hat b_n|
+ \sum_{\vr_n \in \hat \T \setminus \T} |\hat b_n| \\
&= \norm[\setA]{\hat \vz}
\end{align*}
where we used that $|Q(\vr)| < 1$ for $\vr \notin \T$. This contradicts strong duality and implies that all optimal solutions must be supported on $\T$. Since the set of atoms with $\vr_n \in \T$ are linearly independent, it follows that the optimal solution is unique.

\section{Comparison to MUSIC \label{app:MUSIC}}

As mentioned previously, the MUSIC algorithm (and related methods) can not be applied directly to the super-resolution radar problem, since MUSIC in general requires multiple measurements (snapshots). 
However, as pointed out by Peter Stoica (personal communication, June 2015), multiple measurements can be obtained from a single measurement by sending a periodic input signal. 
Specifically, let $L = M^2$, and suppose that the entries of $\vx \in \complexset^L$ are $M$-periodic, i.e., $\vx = \transp{[\transp{\tilde \vx}, ..., \transp{\tilde \vx}]}$, where $\tilde \vx \in \complexset^M$. 
We next define a version of our time and frequency shift operators in \eqref{eq:deftimefreqshifts} 
\begin{align}
[\mc T_{\tau}^{(L)} \vx ]_p
\defeq
\frac{1}{L}
\sum_{k=0}^{L-1} 
\left[ 
\left(
\sum_{\ell=0}^{L-1}
x_{\ell} e^{- i2\pi \frac{\ell k}{L}  }
\right)
e^{-i2\pi k  \tau }   
\right]
e^{i2\pi \frac{p k}{L}  }, 
\quad 
[\mc F_{\nu}^{(L)} \vx ]_p
\defeq x_p e^{i2\pi p \nu },\quad p=0,...,\L-1
\label{eq:deftimefreqshifts}
\end{align}
where we made the dependence on the length of $\vx$ (i.e., $L$), explicit. 
By using that the entries of $\vx$ are $M$-periodic, we obtain 
\begin{align}
\mc F_{\nu}^{(L)} \mc T_{\tau}^{(L)} \vx
 = 
 \begin{bmatrix}
 \mc F_{\nu}^{(M)} \mc T_{M\tau}^{(M)} \tilde \vx \\
e^{i2\pi \nu  M}  \mc F_{\nu}^{(M)} \mc T_{M\tau}^{(M)} \tilde \vx \\
 \hdots \\
e^{i2\pi \nu M (M-1)}   \mc F_{\nu}^{(M)} \mc T_{M\tau}^{(M)} \tilde \vx 
\end{bmatrix}.
\label{eq:periodicio}
\end{align}
We can now extract the vectors $\vy_p \in \complexset^M, p=0,...,M-1$,  from $\vy$ in \eqref{eq:periorel} according to $\vy = \transp{[\vy_0,...,\vy_{M-1}]}, \vy_i \in \complexset^M$, 
and obtain the input-output relation
\begin{align}
\vy_p 
= 
\sum_{j=1}^{S}
 \underbrace{b_j  e^{i2\pi \nu_j  M p} }_{b_{j,p}'}  \mc F_{\nu_j}^{(M)} \mc T_{M\tau_j}^{(M)} \tilde \vx
,
 \quad p=0,...,M-1.
 \label{eq:musicstoica}
\end{align}
The MUSIC algorithm can be applied to the measurements in \eqref{eq:musicstoica} in order to extract the time-frequency shifts $(\tau_j,\nu_j)$. However, the following three conditions are necessary for MUSIC to succeed: 
\begin{enumerate}
\item \label{mcond1} $S < M = \sqrt{L}$, i.e., the number of measurements must be quadratic in the number of time-frequency shifts. By Theorem \ref{thm:mainres}, the convex program $\AN(\vy)$ in \eqref{eq:primal} succeeds even when $S$ is (up to log-factors) linear in $L$.  
\item \label{mcond2} $\tau M \in [0,1]$, since $\mc T^{(M)}_{M\tau}$ is $1$-periodic in $M\tau$. 
The atomic norm minimization approach in \eqref{eq:primal} only requires $\tau \in [0,1]$. 
\item \label{mcond3} The $\nu_j$ need to be distinct. The atomic norm minimization approach \eqref{eq:primal} does only require the minimum separation condition to be satisfied either in time, or in frequency. 
\end{enumerate}
While in the noiseless case, recovery of the time-frequency shifts under conditions \ref{mcond1}, \ref{mcond2}, and \ref{mcond3}, succeeds provided the probing sequence $\tilde \vx$ is chosen appropriately (e.g., by drawing the entries $\tilde \vx$ i.i.d.~uniform from the complex unit disk), 
 recovery of the time-frequency shifts by application of the MUSIC algorithm to the measurement in \eqref{eq:musicstoica} is significantly more sensitive to noise than our convex programming approach is. 
This is a consequence of periodizing the input signal $\vx$, and 
is exemplified with the following numerical example. 

As in Section \ref{sec:discretesuperes}, we assume that the time-frequency shifts $(\tau_j,\nu_j), j=1,...,S$, lie on a fine grid with grid constant $\left( \frac{1}{L \cdot \SRF}, \frac{1}{L \cdot \SRF} \right)$, with super-resolution factor $\SRF=6$, and suppose that $(\tau_j, \nu_j) \in [0,1/M]^2$. 
We set $M = 17$, choose the time-frequency shifts as 
$\tau_j = \frac{j}{L}, \nu_j = \frac{j}{L}, j = 1,...,S$, $S \in \{1,4,16\}$, and draw the corresponding attenuation factors $b_j$ i.i.d.~uniformly at random from the complex unit disc. Note that by construction, the time-frequency shifts lie on the grid, and Conditions \ref{mcond1}, \ref{mcond2}, and \ref{mcond3} are satisfied. 

We generate the MUSIC-compatible measurements $\vy_p$ according to \eqref{eq:musicstoica}, where the entries of the probing signal $\tilde \vx \in \complexset^{M}$ are drawn i.i.d.~at random from the complex unit disc, and 
 consider the following variant of the MUSIC algorithm. 
Given the measurements $\vy_p, p=0,...,M-1$, we compute the matrix of eigenvectors $\mU \in \complexset^{M \times M-S}$ corresponding to the $M-S$ smallest eigenvalues of $\mZ = \sum_{p=0}^{M-1} \vy_p \herm{\vy}_p$. 
The time-frequency shifts are identified as the $S$ time-frequency shifts $(\tau,\nu) \in 
\{ \left( \frac{k}{L \cdot \SRF}, \frac{\ell}{L \cdot \SRF} \right) \colon (k,\ell) \in \left\{0,...,L\cdot \SRF/M-1\right\}^2 \}$ that minimize 
$\norm[2]{\herm{\mU} \mc F_{\nu}^{(M)} \mc T_{M\tau}^{(M)} \tilde \vx}$.
 
We compare the MUSIC algorithm to our convex programming approach. Specifically, we generate the probing signal $\vx \in \complexset^L$ by drawing its entries i.i.d. from the complex unit disc, and generate the measurement $\vy \in \complexset^L$ according to \eqref{eq:periorel}. We then recover the time-frequency shifts with $\text{L1-ERR}$ defined in \eqref{eq:BDDN}, where the 
 matrix $\mR$ in \eqref{eq:BDDN} has columns $\mc F_{\nu}
\mc T_{\tau}
\vx$, $(\nu,\tau) \in \{ \left( \frac{k}{L \cdot \SRF}, \frac{\ell}{L \cdot \SRF} \right) \colon (k,\ell) \in \{ 0,...,L\cdot \SRF/M -1 \}^2 \}$. 

The results are plotted in Figure \ref{fig:cmpmusicconvex}. 
While L1-ERR perfectly recovers the position of the time-frequency shifts on the grid as long as $\text{SNR} \leq 15$dB, the resolution error of MUSIC is quite large if there are many time-frequency shifts, even for large SNRs (i.e., for $\text{SNR} = 50$dB). 

This is not surprising, since MUSIC, due to the periodization of the input signal required to obtain the input-output relation \eqref{eq:musicstoica}, has to deal with a significantly worse conditioned matrix. To see this, set for simplicity the time shifts to $\tau_j=0$ and note that, with this choice, writing \eqref{eq:musicstoica} in matrix-vector form yields 
\[
\vy_p = \tilde \mF   \vb_p, \quad \tilde \mF \defeq [ \mc F_{\nu_1}^{(M)} \tilde \vx,..., \mc F_{\nu_S}^{(M)}  \tilde \vx ], 
\]
where $[\vb_p]_j = b_{j,p}$. The matrix $\tilde \mF \in \reals^{M\times S}$ is equivalent to the upper left corner of a $M^2\times M^2$ DFT matrix with scaled rows (scaled by the $[\tilde \vx]_p$), i.e., $[\tilde \mF]_{p,j} = [\tilde \vx]_p e^{i2\pi \frac{p j}{M^2}}, p = 0,..,M-1, j= 1,..., S$, and as such is ill-conditioned, in particular if $M$ and $S$ are large.

\begin{figure}[!ht]
\centering
\includegraphics{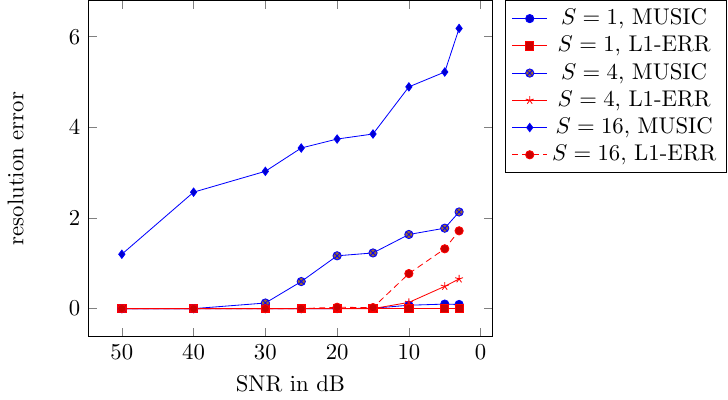}
%
%
%
%
%
%
%
%

\caption{
\label{fig:cmpmusicconvex}
Resolution error, defined as 
the average over $j=1,...,S$ of $\L \sqrt{ (\hat \tau_j - \tau_j)^2 + (\hat \nu_j - \nu_j)^2}$, where the $(\hat \tau_j, \hat \nu_j)$ are the time-frequency shifts obtained by i) application of the MUSIC algorithm to the samples $\vy_p + \vn_p, p=0,...,M-1$ in \eqref{eq:musicstoica}, and ii)
application of \text{L1-ERR} to the sample $\vy + \vn$ in \eqref{eq:periorel}. Here, $\vn_p\in \complexset^M$ and $\vn \in \complexset^L$ is additive Gaussian noise, such that the signal-to-noise ratio is $\text{SNR} \defeq \left( \sum_{p=0}^{M-1} \norm[2]{ \vy_p }^2 \right) / \left( \sum_{p=0}^{M-1} \norm[2]{ \vn_p}^2 \right)$ 
 and
 $\text{SNR} \defeq \norm[2]{\vy}^2 / \norm[2]{ \vn }^2$.   
}

\end{figure}

\section{Proof of Theorem~\ref{thm:mtxHW}}

\newcommand{\E}{\operatorname{\mathbb{E}}}

In this section, we prove the matrix Hanson-Wright inequality given in Theorem~\ref{thm:mtxHW}, restated below:
\begin{theorem}
Let $g_1,\ldots,g_{L}$ be iid zero-mean and unit variance Gaussian random variables, and let $\mB_{\ell,\ell'} \in \reals^{d_1\times d_2}$ be fixed matrices. 
Then, for $t \geq 0$,
\begin{align*}
\PR{
\norm{\sum_{\ell, \ell'=1}^L
\mtx{B}_{\ell,\ell'} (g_\ell g_{\ell'} - \EX{g_\ell g_{\ell'} } ) }\ge t
}
\leq
2(d_1+d_2)e^{-c\min\left(\frac{t}{\max_\ell\norm{\mtx{B}_\ell}},\frac{t^2}{\max\left(\norm{\sum_{\ell=1}^L \mtx{B}_\ell^T\mtx{B}_\ell},\norm{\sum_{\ell=1}^L \mtx{B}_\ell\mtx{B}_\ell^T}\right)}\right)}\\
+
4(d_1+d_2)(L+1)e^{-\frac{t}{8\sqrt{e L \max_{\ell'}\max\left(\norm{\sum_{\ell\neq\ell'}\mtx{B}_{\ell,\ell'}^T\mtx{B}_{\ell,\ell'}},\norm{\sum_{\ell\neq\ell'}\mtx{B}_{\ell,\ell'}\mtx{B}_{\ell,\ell'}^T}\right)}}}
\end{align*}
\end{theorem}

The proof of this statement is a generalization of the decoupling proof of the classical Hanson-Wright inequality but with some intricate modifications (see \cite[Section 6.2]{Vershynin_2018} for a nice overview of the classical decoupling proof). 
See also~\cite{Adamczak_Latala_Meller_2020} for a general Hanson-Wright inequality for Banach spaces. 

Define the double-sided tail
\begin{align*}
p:=
\PR{
\norm[\opnormss]{\sum_{\ell=1}^L\sum_{\ell'=1}^Lg_\ell g_{\ell'}\mtx{B}_{\ell,\ell'}-\sum_{\ell=1}^L\mtx{B}_{\ell,\ell}}\ge t }
\end{align*}
Now consider the following decomposition
\begin{align*}
\sum_{\ell=1}^L\sum_{\ell'=1}^Lg_\ell g_{\ell'}\mtx{B}_{\ell,\ell'}-\sum_{\ell=1}^L\mtx{B}_{\ell,\ell}=\sum_{\ell=1}^L (g_\ell^2-1)\mtx{B}_{\ell,\ell}+\sum_{\ell,\ell':\text{ }\ell\neq \ell'} g_\ell g_{\ell'} \mtx{B}_{\ell,\ell'}.
\end{align*}
Therefore, the problem reduces to estimating the diagonal and off-diagonal sums:
\begin{align*}
p\le \PR{ \norm[\opnormss]{\sum_{\ell=1}^L (g_\ell^2-1)\mtx{B}_{\ell,\ell}}\ge \frac{t}{2}}
+
\PR{ 
\norm[\opnormss]{\sum_{\ell,\ell':\text{ }\ell\neq \ell'} g_\ell g_{\ell'} \mtx{B}_{\ell,\ell'}}\ge \frac{t}{2} } := p_1+p_2.
\end{align*}
We proceed by bounding each of these terms.

\paragraph{Diagonal sum:} First define the symmetric matrices
\begin{align*}
\mtx{A}_{\ell}=\begin{bmatrix} \mtx{0} & \mtx{B}_{\ell,\ell} \\ \mtx{B}_{\ell,\ell}^T & \mtx{0} \end{bmatrix}
\end{align*}
and note that 
\begin{align*}
\norm[\opnormss]{\sum_{\ell=1}^L(g_\ell^2-1)\mtx{B}_{\ell,\ell} }=\norm[\opnormss]{\sum_{\ell=1}^L (g_\ell^2-1) \mtx{A}_{\ell}}.
\end{align*}
Furthermore, note that 
\begin{align*}
\norm[\opnormss]{\sum_{\ell=1}^L(g_\ell^2-1) \mtx{A}_{\ell}}=\max\left(\lambda_{\max}\left(\sum_{\ell=1}^L(g_\ell^2-1) \mtx{A}_{\ell}\right),
\lambda_{\max}\left(-\sum_{\ell=1}^L(g_\ell^2-1) \mtx{A}_{\ell}\right)
\right).
\end{align*}
Thus, using Markov's inequality
\begin{align*}
\PR{ \norm[\opnormss]{\sum_{\ell=1}^L(g_\ell^2-1) \mtx{B}_{\ell,\ell}}\ge \frac{t}{2} }
=&\PR{\norm[\opnormss]{\sum_{\ell=1}^L(g_\ell^2-1) \mtx{A}_{\ell}}\ge \frac{t}{2} } \\
=&2\PR{ \lambda_{\max}\left(\sum_{\ell=1}^L(g_\ell^2-1) \mtx{A}_{\ell}\right)\ge \frac{t}{2} }\\
=&2 \PR{ e^{\theta\cdot \lambda_{\max}\left(\sum_{\ell=1}^L(g_\ell^2-1) \mtx{A}_{\ell}\right)}\ge e^{\frac{\theta t}{2}} }\\
\le&2e^{-\frac{\theta t}{2}}\E_g\Bigg[e^{\theta\cdot\lambda_{\max}\left(\sum_{\ell=1}^L(g_\ell^2-1) \mtx{A}_{\ell}\right)}\Bigg]\\
\le&2e^{-\frac{\theta t}{2}}\E_g\Bigg[\text{trace}\left(e^{\theta\cdot\sum_{\ell=1}^L(g_\ell^2-1) \mtx{A}_{\ell}}\right)\Bigg]\\
\overset{(a)}{\le}&2e^{-\frac{\theta t}{2}}\text{trace}\left(e^{\sum_{\ell=1}^L\log \left(\E_g\Big[ e^{\theta(g_\ell^2-1) \mtx{A}_{\ell}}\Big]\right)}\right),
\end{align*}
where (a) follows from using \cite[Equation (5.15)]{Vershynin_2018}. 
To continue note that for a matrix with eigenvalue decomposition $\mtx{A}_\ell=\mtx{V}_\ell\mtx{\Lambda}_\ell\mtx{V}_\ell^T$ for $\theta\ge 0$ obeying $\theta\le \frac{c}{\lambda_{\max}(\mtx{A}_\ell)}=\frac{c}{\norm[\opnormss]{\mtx{B}_\ell}}$ we have
\begin{align*}
\E_g\Big[ e^{\theta(g^2-1) \mtx{A}_{\ell}}\Big]=\mtx{V}_\ell \E_g\Big[ e^{\theta(g^2-1) \mtx{\Lambda}_{\ell}}\Big]\mtx{V}_\ell^T\preceq \mtx{V}_\ell e^{-C\theta^2\Lambda_\ell^2}  \mtx{V}_\ell^T=e^{-C\theta^2\mtx{A}_\ell^2},
\end{align*}
where in the inequality we used the fact that $g^2-1$ is a subexponential random variable and hence \cite[Lemma 5.15]{vershynin_introduction_2012} on the moment generating function of subexponential random variable applies.

Thus, continuing from the above and minimizing the expression over $0\le \theta\le \frac{1}{4\max_\ell \norm[\opnormss]{\mtx{B}_\ell}}$ we have
\begin{align*}
\PR{ \norm[\opnormss]{\sum_{\ell=1}^L(g_\ell^2-1) \mtx{B}_{\ell,\ell}}\ge \frac{t}{2} }
\le& 2(d_1+d_2)e^{-c\min\left(\frac{t}{\max_\ell\norm[\opnormss]{\mtx{B}_\ell}},\frac{t^2}{\max\left(\norm[\opnormss]{\sum_{\ell=1}^L \mtx{B}_\ell^T\mtx{B}_\ell},\norm[\opnormss]{\sum_{\ell=1}^L \mtx{B}_\ell\mtx{B}_\ell^T}\right)}\right)}.
\end{align*}

\paragraph{Off-diagonal sum:} First define the symmetric matrices
\begin{align*}
\mtx{A}_{\ell,\ell'}=\begin{bmatrix} \mtx{0} & \mtx{B}_{\ell,\ell'} \\ \mtx{B}_{\ell,\ell'}^T & \mtx{0} \end{bmatrix}\quad\text{and}\quad \mtx{A}_{\ell,\ell}=\mtx{0}
\end{align*}
and note that 
\begin{align*}
\norm[\opnormss]{\sum_{\ell,\ell':\text{ }\ell\neq \ell'} g_\ell g_{\ell'} \mtx{B}_{\ell,\ell'}}=\norm[\opnormss]{\sum_{\ell=1}^L\sum_{\ell'=1}^L g_\ell g_{\ell'} \mtx{A}_{\ell,\ell'}},
\end{align*}
Furthermore, note that 
\begin{align*}
\norm[\opnormss]{\sum_{\ell=1}^L\sum_{\ell'=1}^L g_\ell g_{\ell'} \mtx{A}_{\ell,\ell'}}=\max\left(\lambda_{\max}\left(\sum_{\ell=1}^L\sum_{\ell'=1}^L g_\ell g_{\ell'} \mtx{A}_{\ell,\ell'}\right),
\lambda_{\max}\left(-\sum_{\ell=1}^L\sum_{\ell'=1}^L g_\ell g_{\ell'} \mtx{A}_{\ell,\ell'}\right)
\right).
\end{align*}
Thus, using Markov's inequality
\begin{align*}
\PR{\norm[\opnormss]{\sum_{\ell,\ell':\text{ }\ell\neq \ell'} g_\ell g_{\ell'} \mtx{B}_{\ell,\ell'}}\ge \frac{t}{2} }
=&
\PR{\norm[\opnormss]{\sum_{\ell=1}^L\sum_{\ell'=1}^L g_\ell g_{\ell'} \mtx{A}_{\ell,\ell'}}\ge \frac{t}{2} }
\\
=&2\PR{\lambda_{\max}\left(\sum_{\ell=1}^L\sum_{\ell'=1}^L g_\ell g_{\ell'} \mtx{A}_{\ell,\ell'}\right)\ge \frac{t}{2} }\\
=&2\PR{e^{\theta\cdot \lambda_{\max}\left(\sum_{\ell=1}^L\sum_{\ell'=1}^L g_\ell g_{\ell'} \mtx{A}_{\ell,\ell'}\right)}\ge e^{\frac{\theta t}{2}} }\\
\le&2e^{-\frac{\theta t}{2}}\E_g\Bigg[e^{\theta\cdot\lambda_{\max}\left(\sum_{\ell=1}^L\sum_{\ell'=1}^L g_\ell g_{\ell'} \mtx{A}_{\ell,\ell'}\right)}\Bigg]\\
\le&2e^{-\frac{\theta t}{2}}\E_g\E_{\widetilde{g}}\Bigg[e^{4\theta\cdot\lambda_{\max}\left(\sum_{\ell=1}^L\sum_{\ell'=1}^L g_\ell \widetilde{g}_{\ell'} \mtx{A}_{\ell,\ell'}\right)}\Bigg],
\end{align*}
where in the last line we used the result in \cite[Exercise 6.1.5]{Vershynin_2018} and $\widetilde{g}_\ell$ are independent copies of $g_\ell$. Continuing the above chain of inequalities we have
\begin{align*}
\PR{
\norm[\opnormss]{\sum_{\ell,\ell':\text{ }\ell\neq \ell'} g_\ell g_{\ell'} \mtx{B}_{\ell,\ell'}}\ge \frac{t}{2} }
\le&2e^{-\frac{\theta t}{2}}\E_g\E_{\widetilde{g}}\Bigg[e^{4\theta\cdot\lambda_{\max}\left(\sum_{\ell=1}^L\sum_{\ell'=1}^L g_\ell \widetilde{g}_{\ell'} \mtx{A}_{\ell,\ell'}\right)}\Bigg]\\
=&2e^{-\frac{\theta t}{2}}\E_g\E_{\widetilde{g}}\Bigg[\lambda_{\max}\left(e^{4\theta\cdot\sum_{\ell=1}^L\sum_{\ell'=1}^L g_\ell \widetilde{g}_{\ell'} \mtx{A}_{\ell,\ell'}}\right)\Bigg]\\
\le&2e^{-\frac{\theta t}{2}}\E_g\E_{\widetilde{g}}\Bigg[\text{trace}\left( e^{4\theta\cdot\sum_{\ell=1}^L\sum_{\ell'=1}^Lg_\ell \widetilde{g}_{\ell'} \mtx{A}_{\ell,\ell'}}\right)\Bigg]\\
=&2e^{-\frac{\theta t}{2}}\E_{\widetilde{g}}\E_g\Bigg[\text{trace}\left( e^{4\theta\cdot\sum_{\ell=1}^L g_\ell \left(\sum_{\ell'=1}^L\widetilde{g}_{\ell'} \mtx{A}_{\ell,\ell'}\right)}\right)\Bigg].
\end{align*}
Now note that using \cite[Equation (5.15)]{Vershynin_2018}, we have
\begin{align*}
\E_g\Bigg[\text{trace}\left( e^{4\theta\cdot\sum_{\ell=1}^L g_\ell \left(\underset{\ell': \ell'\neq \ell}{\sum}\widetilde{g}_{\ell'} \mtx{A}_{\ell,\ell'}\right)}\right)\Bigg]\le&\text{trace}\left(e^{\sum_{\ell=1}^L \log  \E_{g_\ell}\Big[e^{4\theta g_\ell\cdot \left(\sum_{\ell'=1}^L\widetilde{g}_{\ell'} \mtx{A}_{\ell,\ell'}\right)}\Big]}\right)\\
=&\text{trace}\left(e^{8\theta^2\sum_{\ell=1}^L\left(\sum_{\ell'=1}^L\widetilde{g}_{\ell'} \mtx{A}_{\ell,\ell'}\right)^2}\right),
\end{align*}
where in the last line we used the fact that for a matrix $\mtx{A}$ and a Gaussian random variable $g$ we have $\log E_g e^{g\theta \mtx{A}}=\frac{\theta^2}{2}\mtx{A}^2$. Combining the last two identities we conclude that
\begin{align*}
\mathbb{P}\Bigg\{\norm[\opnormss]{\sum_{\ell,\ell':\text{ }\ell\neq \ell'} g_\ell g_{\ell'} \mtx{B}_{\ell,\ell'}}\ge \frac{t}{2}\Bigg\}\le 2e^{-\frac{\theta t}{2}}\E_{\widetilde{g}}\Bigg[\text{trace}\left(e^{8\theta^2\sum_{\ell=1}^L\left(\sum_{\ell'=1}^L\widetilde{g}_{\ell'} \mtx{A}_{\ell,\ell'}\right)^2}\right)\Bigg].
\end{align*}
To continue, define the matrices 
\begin{align*}
\mtx{C}_{\ell,\ell'}=\mtx{e}_{\ell}^T\otimes \mtx{A}_{\ell,\ell'}\quad\text{and}\quad
 \mtx{D}_{\ell'}=
 \begin{bmatrix} \mtx{0} & \sum_{\ell=1}^L\mtx{C}_{\ell,\ell'}^T\\
 \sum_{\ell=1}^L\mtx{C}_{\ell,\ell'} & \mtx{0}
 \end{bmatrix}
\end{align*}
and note that
\begin{align*}
\sum_{\ell=1}^L\left(\sum_{\ell'=1}^L\widetilde{g}_{\ell'} \mtx{A}_{\ell,\ell'}\right)^2
=\left(\sum_{\ell'=1}^L\widetilde{g}_{\ell'} \left(\sum_{\ell=1}^L\mtx{C}_{\ell,\ell'}\right)\right)\left(\sum_{\ell'=1}^L\widetilde{g}_{\ell'} \left(\sum_{\ell=1}^L\mtx{C}_{\ell,\ell'}\right)\right)^T
\end{align*}
Furthermore, the matrix $\left(\sum_{\ell'=1}^L\widetilde{g}_{\ell'}  \mtx{D}_{\ell'}\right)^2$ is block diagonal with the two matrices
\begin{align*}
 \left(\sum_{\ell'=1}^L\widetilde{g}_{\ell'} \left(\sum_{\ell=1}^L\mtx{C}_{\ell,\ell'}\right)\right)^T\left(\sum_{\ell'=1}^L\widetilde{g}_{\ell'} \left(\sum_{\ell=1}^L\mtx{C}_{\ell,\ell'}\right)\right)
 ,\quad  
 \left(\sum_{\ell'=1}^L\widetilde{g}_{\ell'} \left(\sum_{\ell=1}^L\mtx{C}_{\ell,\ell'}\right)\right)\left(\sum_{\ell'=1}^L\widetilde{g}_{\ell'} \left(\sum_{\ell=1}^L\mtx{C}_{\ell,\ell'}\right)\right)^T
\end{align*}
on its diagonal, 
so that
\begin{align*}
\text{trace}\left(e^{8\theta^2\sum_{\ell=1}^L\left(\sum_{\ell'=1}^L\widetilde{g}_{\ell'} \mtx{A}_{\ell,\ell'}\right)^2}\right)=&\text{trace}\left(e^{4\theta^2\left(\sum_{\ell'=1}^L\widetilde{g}_{\ell'}  \mtx{D}_{\ell'}\right)^2}\right).
\end{align*}
Thus
\begin{align*}
&\E_{\widetilde{g}}\Bigg[\text{trace}\left(e^{8\theta^2\sum_{\ell=1}^L\left(\sum_{\ell'=1}^L\widetilde{g}_{\ell'} \mtx{A}_{\ell,\ell'}\right)^2}\right)\Bigg]\\
=&\E_{\widetilde{g}}\Bigg[\text{trace}\left(e^{4\theta^2\left(\sum_{\ell'=1}^L\widetilde{g}_{\ell'}  \mtx{D}_{\ell'}\right)^2}\right)\Bigg]\\
=&\sum_{s=0}^{+\infty}\frac{(4\theta^2)^s}{s!}\E_{\widetilde{g}}\Bigg[\text{trace}\left(\left(\sum_{\ell'=1}^L\widetilde{g}_{\ell'}  \mtx{D}_{\ell'}\right)^{2s}\right)\Bigg]\\
\le&\text{trace}(\mtx{I})+\sum_{s=1}^{+\infty}\frac{(4\theta^2(2s-1))^s}{s!}\text{trace}\left(\left(\sum_{\ell'=1}^L\mtx{D}_{\ell'}^2\right)^{s}\right)\\
=&\text{trace}\left(\mtx{I}+\sum_{s=1}^{+\infty}\frac{(2s-1)^s}{s!}\left(4\theta^2\sum_{\ell'=1}^L\mtx{D}_{\ell'}^2\right)^{s}\right)\\
\le&(d_1+d_2)(L+1)\lambda_{\max}\left(\mtx{I}+\sum_{s=1}^{+\infty}\frac{(2s-1)^s}{s!}\left(4\theta^2\sum_{\ell'=1}^L\mtx{D}_{\ell'}^2\right)^{s}\right)\\
\le&(d_1+d_2)(L+1)\left(1+\sum_{s=1}^{+\infty}\frac{(2s-1)^s}{s!}\left(4\theta^2\norm[\opnormss]{\sum_{\ell'=1}^L\mtx{D}_{\ell'}^2}\right)^{s}\right)\\
\le&(d_1+d_2)(L+1)\left(1+\sum_{s=1}^{+\infty}\frac{(2s-1)^s}{s!}\left(4\theta^2L\cdot\max_{\ell'}\norm[\opnormss]{\mtx{D}_{\ell'}^2}\right)^{s}\right)\\
=&(d_1+d_2)(L+1)\left(1+\sum_{s=1}^{+\infty}\frac{(2s-1)^s}{s!}\left(4\theta^2L\cdot\max_{\ell'}\max\left(\norm[\opnormss]{\sum_{\ell\neq\ell'}\mtx{B}_{\ell,\ell'}^T\mtx{B}_{\ell,\ell'}},\norm[\opnormss]{\sum_{\ell\neq\ell'}\mtx{B}_{\ell,\ell'}\mtx{B}_{\ell,\ell'}^T}\right)\right)^{s}\right)\\
\le&(d_1+d_2)(L+1)\left(1+\sum_{s=1}^{+\infty}\frac{(2s)^s}{s!}\left(4\theta^2L\cdot\max_{\ell'}\max\left(\norm[\opnormss]{\sum_{\ell\neq\ell'}\mtx{B}_{\ell,\ell'}^T\mtx{B}_{\ell,\ell'}},\norm[\opnormss]{\sum_{\ell\neq\ell'}\mtx{B}_{\ell,\ell'}\mtx{B}_{\ell,\ell'}^T}\right)\right)^{s}\right)\\
\le&(d_1+d_2)(L+1)\left(1+\sum_{s=1}^{+\infty}\left(8e\theta^2L\cdot\max_{\ell'}\max\left(\norm[\opnormss]{\sum_{\ell\neq\ell'}\mtx{B}_{\ell,\ell'}^T\mtx{B}_{\ell,\ell'}},\norm[\opnormss]{\sum_{\ell\neq\ell'}\mtx{B}_{\ell,\ell'}\mtx{B}_{\ell,\ell'}^T}\right)\right)^{s}\right)\\
\le&\frac{(d_1+d_2)(L+1)}{1-8e\theta^2L\cdot\max_{\ell'}\max\left(\norm[\opnormss]{\sum_{\ell\neq\ell'}\mtx{B}_{\ell,\ell'}^T\mtx{B}_{\ell,\ell'}},\norm[\opnormss]{\sum_{\ell\neq\ell'}\mtx{B}_{\ell,\ell'}\mtx{B}_{\ell,\ell'}^T}\right)}\\
\le&2(d_1+d_2)(L+1)
\end{align*}
where in the first inequality we used \cite[Proposition 7.1]{Tropp_2018}  and the last inequality holds when $16e\theta^2L\cdot\max_{\ell'}\max\left(\norm[\opnormss]{\sum_{\ell\neq\ell'}\mtx{B}_{\ell,\ell'}^T\mtx{B}_{\ell,\ell'}},\norm[\opnormss]{\sum_{\ell\neq\ell'}\mtx{B}_{\ell,\ell'}\mtx{B}_{\ell,\ell'}^T}\right)\le 1$ which is equivalent to 
\begin{align*}
\theta\le \frac{1}{4\sqrt{e L \max_{\ell'}\max\left(\norm[\opnormss]{\sum_{\ell\neq\ell'}\mtx{B}_{\ell,\ell'}^T\mtx{B}_{\ell,\ell'}},\norm[\opnormss]{\sum_{\ell\neq\ell'}\mtx{B}_{\ell,\ell'}\mtx{B}_{\ell,\ell'}^T}\right)}}
\end{align*}
Thus, using the latter chain of inequalities with the above upper bound on $\theta$ we have
\begin{align*}
\PR{ \norm[\opnormss]{\sum_{\ell,\ell':\text{ }\ell\neq \ell'} g_\ell g_{\ell'} \mtx{B}_{\ell,\ell'}}\ge \frac{t}{2} }
\le 4(d_1+d_2)(L+1)e^{-\frac{t}{8\sqrt{e L \max_{\ell'}\max\left(\norm[\opnormss]{\sum_{\ell\neq\ell'}\mtx{B}_{\ell,\ell'}^T\mtx{B}_{\ell,\ell'}},\norm[\opnormss]{\sum_{\ell\neq\ell'}\mtx{B}_{\ell,\ell'}\mtx{B}_{\ell,\ell'}^T}\right)}}},
\end{align*}
which concludes the bound on the off-diagonal term.

Combining the bounds of the diagonal and off-diagonal terms concludes the proof.

\end{document}